\documentclass[12pt]{article}

\usepackage[margin=1in]{geometry}
\usepackage{graphicx}
\usepackage{xcolor}

\usepackage{amsmath, amsfonts, amsthm, mathtools}

\usepackage{fontspec}
\usepackage{unicode-math}

\theoremstyle{definition}
\newtheorem{definition}{Definition}
\newtheorem{assumption}{Assumption}

\theoremstyle{plain}

\newtheorem{lemma}{Lemma}
\newtheorem{theorem}{Theorem}

\theoremstyle{remark}
\newtheorem{remark}{Remark}

\usepackage[
    style=authoryear, 
    maxbibnames=40,
    url=false,
    isbn=false,
    natbib=true,
    uniquename=false
]{biblatex}
\usepackage[colorlinks, allcolors=blue]{hyperref}

\usepackage{subcaption}
\usepackage{authblk}
\usepackage{papermacros}

\begin{document}

\title{Estimating peer effects in noisy, low-rank networks via network smoothing}
\author[1]{Alex Hayes}
\author[2]{Keith Levin}
\date{\today}
\affil[1]{\small Department of Economics, Stanford University, USA}
\affil[2]{\small Department of Statistics, University of Wisconsin--Madison, USA}

\maketitle

\begin{quote}
    Peer effect estimation requires precise network measurement, yet most empirical networks are noisy, rendering standard estimators inconsistent. To address measurement error in networks, we propose a method to estimate peer effects in networks whose expected adjacency matrix is low-rank. Our key result shows that peer effects over a true unobserved network are asymptotically equivalent to peer effects over the expected adjacency matrix. This result reduces peer effect estimation in noisy networks to low-rank matrix estimation targeting the expected adjacency matrix. We develop our theory for weighted networks observed with additive noise, but simulations suggest approach can be applied more generally when there is a low-rank estimation method suited to a particular noise structure. We demonstrate via simulations that our approach applies to egocentric samples, aggregated relational data, and networks with missing edges, each requiring a different low-rank estimation method.
\end{quote}

\section{Introduction}

One of the fundamental challenges in estimating social effects, such as contagion, is measuring and quantifying relationships in a reliable manner. These measurement issues pose a fundamental problem, as most current techniques for estimating peer effects in social networks rely on data that precisely encodes potential social influences. For instance, in the social sciences, a popular approach is to use the ``name generator'' method, where each study participant is asked to name their friends or other social ties.  This method is replete with complications: participants report different friends depending on the precise language of the name generator \citep{shakya2017, bidart2011}, participants often forget day-to-day social contacts \citep{smieszek2012}, and reports from different participants can disagree \citep{debacco2021}.

When relationship measurements are noisy, traditional peer effect estimation encounters a problem. Standard network autoregressive models assume that we observe a true, noiseless network, which precisely encodes the pathways along which peer influence operates. If we observe a network via noisy measurement, these estimators may produce inconsistent or misleading results. We propose a solution to this problem: replace noisy measurements of the adjacency matrix with smoothed, de-noised estimates of network structure. Our approach is motivated by a key theoretical insight: under low-rank network models, contagion over a true, noiseless network is asymptotically equivalent to contagion over a smoothed, latent adjacency matrix.

To formalize matters, let us consider two network autoregressive models. The first is a standard \emph{peer contagion model}, where influence operates over the degree-normalized adjacency matrix:
\begin{equation}
    Y_i =  \betanaught + \W_i \betaw + \Xpop_i \betax + \betay \sum_{j \neq i} \frac{\A_{ij}}{d_i}  Y_j + \varepsilon_i,
\end{equation}
where $Y_i$ is an outcome for node $i$, $\W_i \in \R^p$ are observed covariates, $\Xpop_i \in \R^d$ is a latent vector encoding the behavior of node $i$, $\A_{ij} \in \R$ is the strength of the edge between nodes $i$ and $j$, and $d_i = \sum_j \A_{ij}$ is the degree of node $i$.

The second model is a \emph{latent contagion model}, where influence operates over a smoothed, expected adjacency matrix:
\begin{equation}
    Y_i = \thetanaught + \W_i \thetaw + \Xpop_i \thetax + \thetay \sum_{i \neq j} \frac{\E[\Xpop]{\A_{ij}}}{\E[\Xpop_i]{d_i}} Y_j + \varepsilon_i.
\end{equation}
Under a random dot product graph \citep[RDPG;][]{athreya2018} model, where $\E[\Xpop]{\A_{ij}} = \Xpop_i^\top \Xpop_j$, influence is inversely proportional to the cosine distance between nodes in latent space. We develop two-stage least squares estimators for both of these models, and prove that they are consistent under misspecification. That is, estimators derived under the latent contagion model consistently recover peer contagion parameters when the true data generating process follows the peer contagion model, and vice versa. In a formal and quantifiable sense, contagion in peer and latent spaces are (asymptotically) equivalent under a random dot product graph model.

This equivalence has important practical implications. The latent contagion estimators we propose are functions of the nodal covariates $\W$ and the latent positions $\Xpop$, but not the adjacency matrix $\A$ itself. This means that even when $\A$ is observed with noise, our estimators remain consistent provided we can obtain a sufficiently good estimate $\Xhat$ of $\Xpop$. Crucially, in random dot product graphs, it is often possible to obtain high-quality estimates of $\Xpop$ even when the network is observed with measurement error or when data is missing.
We leverage standard tools from spectral network analysis---specifically, the adjacency spectral embedding---to estimate latent positions \citep{lyzinski2014}.
Under a sub-gamma model of edge-level noise, which includes binary, weighted, and count-valued networks as special cases, we show that our estimators achieve $\sqrt{n}$ convergence rates and are asymptotically normal. Since these results hold when the observed network contains sub-gamma noise, this enables inference when network measurements are corrupted by random errors.

The flexibility of our approach extends beyond the specific case of sub-gamma edge noise. While we focus on this setting in our theoretical results, the low-rank smoothing approach that we advocate opens the door to a wide range of robust estimation strategies. Any method that can reliably estimate the low-rank structure of $\E{\A}$---whether through matrix completion for missing data, debiasing for measurement error, or other spectral methods---can be combined with our latent contagion model to estimate peer effects in challenging data settings. We demonstrate this flexibility through simulations and an empirical application studying the contagiousness of smoking in an adolescent social network.

\subsection*{Notation}
\label{subsec:notation}

For a positive integer $n$, we write $[n] = \{1,2,\dots,n\}$. We denote the identity matrix and the zero matrix by $\mI$ and $\symbf 0$, respectively, subscripting by the intended dimension when it is not immediately clear from context. For an $n_1 \times n_2$ matrix $\H$, we denote by $\H_{\cdot j}$ the column vector formed by the $j$-th column of $\H$, and we denote by $\H_{i \cdot}$ the row vector formed by the $i$-th row of $\H$. Abusing notation slightly, we also let $\H_i \in \R^{n_2}$ denote the column vector formed by transposing the $i$-th row of $\H$, that is, $\H_i = (\H_{i \cdot})^\top$. We often consider sequences of matrices indexed by the number of nodes $n$, but suppress the $n$ in our notation to avoid notational clutter. Given any suitably specified ordering on eigenvalues of a square matrix $\H$, we let $\lambda_i(\H)$ denote the $i$-th eigenvalue of $\H$ under that ordering. Similarly, $\sigma_i(\H)$ denotes the $i$-th singular value of $\H$. We let $\norm{\H}$ denote the spectral norm of $H$ and $\norm{\H}_F$ denote the Frobenius norm. We let $\norm{\H}_{2 \rightarrow \infty}$ denote the maximum of the Euclidean norms of the rows of $\H$, so that $\norm{\H}_{2 \rightarrow \infty} = \max_i \norm*{\H_i}$. We use standard Landau notation to indicate convergence rates, as well as typical probabilistic variants $o_p$ and $\mathcal{O}_p$. Throughout, $C > 0$ will denote a constant independent of $n$ that may change from line to line.

\section{Related work}

Our work builds on several strands in the literature. Most directly related are works that model contagion in latent spaces. \citet{sweet2020} suggest a Hoff model for social influence in a latent space, while \citet{chen2023c} model correlations between stock returns via the latent structure encoded by a stochastic blockmodel in a high-dimensional time series setting. These approaches are similar to our own in that they model social influence in a latent space, but they differ from the present work in that they do not consider equivalence between latent and traditional social influence, and thus cannot use the latent space models to account for measurement error in the network.

A substantial body of prior work considers network autoregression with endogeneity, where the network depends on nodal features. \citet{hayes2025b} considers identification and asymptotic signal-to-noise ratios in network autoregression models, proving minimax rates that apply to the models under consideration here. \citet{mcfowland2021} establishes consistency of ordinary least squares in network autoregression models unrolled in time, which \citet{chang2024a} extends to show asymptotic normality in the longitudinal setting. \citet{paul2022a} presents a network autoregressive model that accounts for estimation error in latent positions, but not in the network itself, considering a quasi-maximum likelihood estimator very similar to ours. \citet{johnsson2019} proposes a non-parametric approach to account for endogeneity in cross-sectional networks using sieve estimators, and \citet{egami2021a} proposes a generalized method of moments estimator using double negative controls to account for contextual effects and homophily. All of these approaches assume precise observation of a social network that encodes possible pathways for social influence, an assumption that we find unrealistic and seek to avoid in this paper.

Most relevant to our focus on measurement error is a growing literature on linear-in-means models with network mismeasurement. \citet{griffith2022} and \citet{griffith2024a} consider bias in peer effect estimates when individuals are capped at reporting a maximum number of friends. To the best of our knowledge, our estimators cannot handle this kind of degree censoring and are subject to similar bias. \citet{boucher2025} considers estimation of linear-in-means models when only partial network data is available but the network distribution is known. They propose a simulated generalized method of moments approach to estimating peer effects, given a consistent estimate of edge probabilities for a network, introducing a bias-correction to account for noise due to the simulation process. In contrast, we propose a plug-in estimator for peer effects that does not require bias adjustment. However, low-rank model estimates could naturally be used with the simulated generalized method of moments. \citet{lewbel2024b} shows that two-stage least squares estimators remain consistent under small amounts of network mismeasurement. The settings considered here have substantially more error than accommodated by their results. \citet{lewbel2025} proposes an adjustment for larger amounts of measurement error, based on estimating false positive and false negative rates from either an asymmetric observation of a network or multiple networks. Mechanically, both \citet{lewbel2025} and our own approach work by plugging in an estimate of the expected adjacency matrix to account for noise. Since \citet{lewbel2025} are interested in i.i.d. false positives and false negatives, they require auxiliary data to estimate false positive and false negative rates, whereas we require only a single realization of the network. \citet{li2022} also consider how multiple measurements of a network can be used to estimate and adjust for measurement error rates, in a causal exposure-mapping framework.

Related work on causal inference with network misspecification includes \citet{savje2024b}, which considers causal estimation when exposure maps are misspecified, and \citet{hardy2024}, who propose a mixture model over potentially misspecified treatments in linear-in-means models. These works highlight the close relationship between misspecified treatments and mismeasured networks, which can be thought of as conceptual duals. \citet{zhang2024b} presents an estimator for mismeasured networks with bounded degrees in experimental settings.  \citet{li2025h} considers ordinary least squares estimators of linear-in-means models under parametric misspecification. \citet{lewbel2023} and \citet{yu2022a} consider estimating spillover effects when the network is entirely unobserved. \citet{spohn2023}, \citet{chin2019a}, and \citet{leung2022} similarly consider linear models for causal interference in precisely observed networks. \cite{liu2013} considers two-stage estimators for linear-in-means models in sampled networks, and propose instruments that account for this sampling.

Finally, our work connects to the broader literature on network autoregressive models. \citet[][Chapter 2]{lesage2009} describe how the typical marginal effect interpretation does not apply to linear-in-means models due to non-linearity, and present methods to compute impact scores that retain this interpretation. \citet{vazquez-bare2023} discusses when coefficients have a causal interpretation (see also \citealt{leung2022,mcfowland2021}). Estimation approaches are given by \citet{ord1975, kelejian1998,lee2002,lee2003,lee2004,kelejian2007,lee2010, su2012, drukker2013, lin2010a} and surveyed in \citet{bivand2021}. Key identification results were given by \citet{bramoulle2009}, and \citet{bramoulle2020} surveys identification in network autoregressive and linear-in-means models.

\section{Models and Estimators}
\label{sec:theory}

In this section, we formalize two network autoregressive models for peer influence: a standard peer contagion model and our proposed latent contagion model. We then develop two-stage least squares estimators for both models, and establish our main theoretical results showing that these estimators remain consistent under model misspecification. Our results essentially show that these models and estimators are interchangeable, which is useful because the estimators derived under the latent contagion model can easily be extended to handle noisy networks.

\subsection{Setup} \label{subsec:Gdefs}

Consider a network with $n$ nodes encoded by a symmetric, non-negative adjacency matrix $\A \in \R_+^{n \times n}$. For each node $i \in [n]$, we observe an outcome $Y_i \in \R$ and covariates $\W_i \in \R^p$. We assume each node has an unobserved latent position $\Xpop_i \in \R^d$ that governs how node $i$ forms connections to other nodes. This framing draws on a large literature on latent space models, particularly the random dot product graph \citep[RDPG;][]{athreya2018}, where connection probabilities are functions of these latent positions. Let $d_i = \sum_{j \neq i} \A_{ij}$ denote the degree of node $i$. For notational convenience, we define $\D = \diag(d_1, d_2, \dots, d_n)$ and the degree-normalized adjacency matrix $\G = \D^{-1} \A \in \R^{n \times n}$. If node $i$ is isolated with $d_i = 0$, we set $\G_{i \cdot} = \symbf 0$. Similarly, in the latent space, we define $\Apop = \E[\Xpop]{\A}$, $\dtilde_i = \sum_j \Apop_{ij} = \E[\Xpop_i]{d_i}$, $\Dtilde = \diag(\dtilde_1, \dots, \dtilde_n)$, and $\Gtilde = \Dtilde^{-1} \Apop \in \R^{n \times n}$.

\subsection{Two Models of Contagion}

We now present two complementary specifications for how peer influence propagates through social networks.

\paragraph{Peer contagion model.} The first specification is a standard spatial autoregressive model adapted to networks via latent positions:
\begin{align}
    \label{eq:lim-peer}
    Y_i & =  \betanaught + \W_i \betaw + \Xpop_i \betax + \betay \sum_{j: j \neq i} \frac{\A_{ij}}{d_i}  Y_j + \varepsilon_i \\
    \label{eq:lim-peer-red}
    \Y  & = \paren*{\mI - \betay \G}^{-1} \paren*{\1_n \betanaught + \W \betaw + \Xpop \betax  + \be}
\end{align}
Here, $\betay \in (-1, 1)$ measures how peer outcomes influence focal outcomes through the observed network structure. The latent positions $\Xpop_i$ model homophily, which is crucial for both statistical identification \citep{hayes2025b} and causal inference \citep{shalizi2011}. We call this the \emph{peer contagion} model because influence operates over the degree-normalized adjacency matrix $\G$. We assume that $\varepsilon_i$ i.i.d. and mean zero with variance $\sigma_\varepsilon^2$.

\paragraph{Latent contagion model.} Our proposed alternative replaces the observed adjacency matrix with its conditional expectation:
\begin{align}
    \label{eq:lim-latent}
    Y_i & = \thetanaught + \W_i \thetaw + \Xpop_i \thetax + \thetay \sum_{j:j \neq i} \frac{\E[\Xpop]{\A_{ij}}}{\E[\Xpop_i]{d_i}} Y_j + \varepsilon_i \\
    \label{eq:lim-latent-red}
    \Y  & = \paren*{\mI - \thetay \Gtilde}^{-1} \paren*{\1_n \thetanaught + \W \thetaw + \Xpop \thetax + \be}
\end{align}
Under the RDPG, where $\E[\Xpop_i, \Xpop_j]{\A_{ij}} = \Xpop_i^\top \Xpop_j$, the latent contagion model takes a particularly interpretable form:
\begin{equation}
    \label{eq:lim-latent-red-2}
    Y_i = \thetanaught + \W_i \thetaw + \Xpop_i \thetax + \thetay \sum_{j \neq i} \frac{\Xpop_i^\top \Xpop_j}{\sum_k \Xpop_i^\top \Xpop_k} Y_j + \varepsilon_i.
\end{equation}
\noindent In this specification, peer influence is inversely proportional to the cosine distance between nodes in latent space. All pairs of nodes exert influence on one another, with the strength of influence determined by latent proximity rather than the realization of individual edges.
The latent contagion model reflects a fundamentally different view of social influence: rather than contagion traveling along realized edges, it diffuses based on the propensity for connection. Nodes close in latent space exert strong influence regardless of whether a specific edge forms, while distant nodes exert negligible influence.

\begin{remark}
    The parameters $\symbf \beta$ and $\symbf \theta$ have fundamentally similar interpretations, and we introduce these differing parameters in an attempt to clarify notation. We use $\symbf \beta$ to indicate parameters in the peer contagion model, and $\widehat{\symbf \beta}$ to denote estimators derived under a peer contagion working model. Correspondingly, $\symbf \theta$ indicates parameters in the latent contagion model, and $\widehat{\symbf \theta}$ to denote estimators derived under a latent contagion working model. Later on, we will make statements of the form $\sqrt{n} (\widehat{\symbf \theta} - \symbf \beta) \to \mathcal{N}(0, \Sigma)$. Here, the use of $\beta$ as the underlying parameter indicates that the true data generating process is peer contagion, but $\widehat{\symbf \theta}$ indicates that we are estimating $\symbf \beta$ using an estimator derived under a latent contagion working model.
\end{remark}

Both the peer and latent contagion models require a stochastic model for how the network $\A$ is generated. We adopt a flexible sub-gamma framework that encompasses binary, weighted, and count-valued networks.

\begin{definition}[Sub-gamma network model]
    \label{def:subgamma-network}
    Let $\A \in \R^{n \times n}$ be a random symmetric matrix. Let $\Apop = \E[\Xpop]{\A}$ be the expectation of $\A$ conditional on $\Xpop \in \R^{n \times d}$, which has independent and identically distributed rows $\Xpop_1, \dots, \Xpop_n \in \R^d$. Assume $\Apop$ has rank $d$ and is positive semi-definite with eigenvalues $\spop_1 \ge \spop_2 \ge \cdots \ge \spop_d > 0 = \spop_{d+1} = \cdots = \spop_n$. Conditional on $\Xpop$, the upper-triangular elements of $\A - \Apop$ are independent $(\nu_n, b_n)$-sub-gamma random variables.
\end{definition}

The sub-gamma family is broad, including Bernoulli, Poisson, Exponential, Gamma, and Gaussian distributions, as well as all bounded distributions \citep{boucheron2013,vershynin2020}. We provide a formal definition of sub-gamma random variables along with a handful of related technical results in Appendix~\ref{apx:technical}. The generality of this framework allows us to handle diverse edge types, including binary friendships, weighted interaction frequencies, or count-valued communication volumes. As a consequence of this generality, our results require a comparatively high network density; forthcoming work by Hayes, Chandrasekhar, McCormick and Breza shows that the density requirements are less stringent in binary networks.

\subsection{Identification}

Identification in spatial autoregressive models is subtle, owing to their conditional specification via $\Y_i \mid \Y_{-i}, \W, \Xpop, \A$. The model can equivalently be written in its total law form (Equations \ref{eq:lim-peer-red} and \ref{eq:lim-latent-red}), as is standard in the Markov random field literature \citep{besag1974, rue2005a}. The following identification result is crucial for both models.

\begin{lemma}[\citealt{martellosio2022}]
    \label{lem:martellioso2022}
    Consider a network with finite numbers of nodes $n$. For the peer contagion model \eqref{eq:lim-peer-red}, suppose $\bbE[\be \mid \W, \Xpop, \G] = \symbf 0$ and $\abs*{\betay} < 1$. Then $\betanaught, \betaw, \betax, \betay$ are identified if and only if
    \begin{equation*}
        \rank \begin{bmatrix} \1_n \, \, \W \, \, \Xpop \, \, \G \W \, \, \G \Xpop \end{bmatrix} > \rank \begin{bmatrix} \1_n \, \, \W \, \, \Xpop \end{bmatrix}
    \end{equation*}
    An analogous result holds for the latent contagion model \eqref{eq:lim-latent-red} with $\Gtilde$ replacing $\G$, under the additional assumption that that $d \ge 2$, which rules out collinearity between $\Xpop$ and $\Gtilde \Xpop$.
\end{lemma}

\begin{remark}
    The latent positions $\Xpop$ are only identified up to an orthogonal transformation $\Q$, which implies $\betax$ is also only identified up to $\Q$ \citep{hayes2025b}. In some models, $\1_n$ may lie in the column space of $\Xpop$, in which case the intercept should be dropped.
\end{remark}

\begin{remark}
    Linear-in-means models can suffer from asymptotic degeneracy when $\G \Y$ (or $\Gtilde \Y$) becomes collinear with $\begin{bmatrix} \1_n \, \, \W \, \, \Xpop \end{bmatrix}$, leading to parameters that are identified but inestimable \citep{hayes2025b}. We assume throughout that the network structure prevents such degeneracy. For random dot product graphs, this requires either sufficient sparsity to prevent concentration of $\G \Y$ around its expectation, or sufficient degree heterogeneity to ensure linear independence. See Example 4 in \citet{hayes2025b} for details.
\end{remark}

The following Lemma clarifies the need for a bound on the magnitude of the contagion coefficient. If there is no such bound, the spillover can cause the responses in $\Y$ to diverge.

\begin{lemma} \label{lem:IbgyG:invertible}
    If $|\beta| < 1$, then $\mI-\beta \G$ is invertible with all eigenvalues in the interval $(1 - \beta, 1 + \beta)$.
\end{lemma}
\begin{proof}
    Since $\G = \D^{-1} \A$ is row stochastic, all its eigenvalues have absolute value at most $1$. Therefore, all eigenvalues of $\beta \G$ have absolute value at most $|\beta|$, implying that all eigenvalues of $\mI - \beta \G$ lie in $1 \pm \beta$ and are bounded away from zero.
\end{proof}

\subsection{Estimators}

To construct feasible estimators, we require an estimate $\Xhat$ of the latent positions $\Xpop$. We use the adjacency spectral embedding \citep[ASE;][]{sussman2014}, a spectral estimate appropriate both for precisely observed networks and networks observed with additive noise. Other types of measurement error will require distinct estimators of $\Xpop$. We consider the empirical performance of some of these estimators in the simulation study of Section~\ref{sec:simulations}, but leave detailed theoretical investigation of these settings to future work.

\begin{definition}[Adjacency spectral embedding] \label{def:ASE}
    Given a network $\A$, the $d$-dimensional \emph{adjacency spectral embedding} is $\Xhat = \Uhat \Shat^{1/2} \in \R^{n \times d}$, where $\Uhat \Shat \Vhat^\top$ is the rank-$d$ truncated singular value decomposition of $\A$.
\end{definition}

The adjacency spectral embedding provides a consistent estimate of $\Xpop$ in several settings, with rate varying depending on the precise assumptions \citep{sussman2012a, lyzinski2017, levin2022a, athreya2018}.

\begin{lemma}
    \label{lem:ase-consistency}
    Under suitable regularity conditions, there exists a $d \times d$ orthogonal matrix $\Q$ such that
    \[
        \max_{i \in [n]} \, \norm*{\Q \Xhat_i - \Xpop_i} = \op{1}.
    \]
\end{lemma}

If we knew $\Xpop$, we could construct two-stage least squares in the typical way for linear-in-means models \citep{kelejian1998, lee2003}, but since $\Xpop$ is observed only via $\mA$, we instead use plug-in estimators that replace $\Xpop$ with $\Xhat$ wherever it appears.

\begin{definition}[Peer contagion estimator]
    \label{def:feasibleEstimators:peer}
    Let $\Zcheck = \begin{bmatrix} \1_n \, \, \W \, \, \Xhat \, \, \G \y \end{bmatrix} \in \R^{n \times (p + d + 2)}$ and $\Hcheck = \begin{bmatrix} \W \, \, \Xhat \, \, \G \W \, \, \G \Xhat \, \,  \G^2 \W \, \, \G^2 \Xhat \end{bmatrix} \in \R^{n \times (3 p + 3 d)}$, and let $\Mcheck = \Hcheck (\Hcheck^\top \Hcheck)^{-1} \Hcheck^\top$ denote the projection matrix onto the column space of $\Hcheck$.  We define our peer contagion estimator according to
    \begin{equation} \label{eq:def:betahattsls}
        \betahattsls = (\Zcheck^\top \Mcheck \Zcheck)^{-1} \Zcheck^\top \Mcheck \y .
    \end{equation}
\end{definition}

For the latent contagion estimator, we construct $\Apophat = \Xhat \Xhat^\top$, $\widehat{d}_i = \sum_j \Apophat_{ij}$, $\Dhat = \diag(\widehat{d}_1, \dots, \widehat{d}_n)$, and
\begin{equation}
    \label{eq:def:Ghat}
    \Ghat = \Dhat^{-1} \Apophat \in \R^{n \times n}.
\end{equation}

\begin{definition}[Latent contagion estimator] \label{def:feasibleEstimators:latent}
    Let $\Zhat = \begin{bmatrix} \1_n \, \W \, \, \Xhat \, \, \Ghat \y \end{bmatrix} \in \R^{n \times (p + d + 2)}$ and $\Hhat = \begin{bmatrix} \W \, \, \Xhat \, \, \Ghat \W \, \, \Ghat \Xhat \, \, \Ghat^2 \W \, \, \Ghat^2 \Xhat \end{bmatrix} \in \R^{n \times (3 p + 3 d)}$, and let $\Mhat = \Hhat (\Hhat^\top \Hhat)^{-1} \Hhat^\top$ denote the projection matrix onto the column space of $\Hhat$. Our latent contagion estimator is given by
    \begin{equation} \label{eq:def:thetahattsls}
        \thetahattsls = (\Zhat^\top \Mhat \Zhat)^{-1} \Zhat^\top \Mhat \y .
    \end{equation}
\end{definition}

In some cases, $\Hcheck$ and $\Hhat$ will have collinear columns; any subset of columns with rank $p + d + 1$ or greater is sufficient for identification \citep[see][for details on the construction of the instruments matrix]{bramoulle2009}.

\subsection{Main Results}

We begin with some preliminary theoretical results, which show that under correctly-specified models, peer and latent contagion models that adjust for latent positions can be estimated by plugging in $\Xhat$ as an estimate for $\Xpop$. These results shows that it is possible to distinguish between localized effects on a network, as parameterized by $\betax$ and $\thetax$, and diffusions, as parameterized by $\betay$ and $\thetay$.

\begin{theorem}[Peer contagion estimators under peer contagion]
    \label{thm:peertruepeerfit}
    Suppose the data are generated according to peer contagion as in Equation~\eqref{eq:lim-peer-red} and $\A$ follows a sub-gamma model as in Definition~\ref{def:subgamma-network}. Under Assumptions~\ref{assum:Apop:spectrum},~\ref{assum:degrees},~\ref{assum:latentpositions},~and~\ref{assum:peer-oracle}, detailed in the Appendix, there exists a sequence of orthogonal matrices $\Qdes \in \R^{(p+d+2) \times (p+d+2)}$ such that
    \begin{align*}
        \sqrt{n} \paren*{\Qdes \betahattsls - \symbf \beta} & \to \mathcal{N} \paren*{\symbf 0, \sigmaeps^2 \paren*{\ZpeerOracle^\top \MpeerOracle \ZpeerOracle}^{-1}}
    \end{align*}
    where $\symbf \beta = (\betanaught, \betaw^\top, \betax^\top, \betay)^\top$.
\end{theorem}

\begin{theorem}[Latent contagion estimators under latent contagion]
    \label{thm:lattruelatfit}
    Suppose the data are generated according to latent contagion as in Equation~\eqref{eq:lim-latent-red} and $\A$ follows a sub-gamma model as in Definition~\ref{def:subgamma-network}. Under Assumptions~\ref{assum:Apop:spectrum},~\ref{assum:degrees},~\ref{assum:latentpositions},~\ref{assum:stronger4dhat},~and~\ref{assum:latent-oracle}, detailed in the Appendix, there exists a sequence of orthogonal matrices $\Qdes \in \R^{(p+d+2) \times (p+d+2)}$ such that
    \begin{align*}
        \sqrt{n} \paren*{\Qdes \thetahattsls - \symbf \theta} & \to \mathcal{N} \paren*{\symbf 0, \sigmaeps^2 \paren*{\ZlatOracle^\top \MlatOracle \ZlatOracle}^{-1}},
    \end{align*}
    where $\symbf \theta = (\thetanaught, \thetaw^\top, \thetax^\top, \thetay)^\top$.
\end{theorem}

Our central theoretical contribution establishes that estimators derived under one model remain consistent when the true data generating process follows the other model, if the observed network is distributed as a random dot product graph.
Said another way, correctly-specified two-stage least squares estimators are asymptotically normal in both the peer and latent contagion models.

\begin{theorem}[Latent contagion estimators under peer contagion]
    \label{thm:peertruelatfit}
    Suppose the data are generated according to peer contagion as in Equation~\eqref{eq:lim-peer-red} but we use latent contagion estimators from Definition~\ref{def:feasibleEstimators:latent}. Suppose $\A$ follows a sub-gamma model as in Definition~\ref{def:subgamma-network}. Under Assumptions~\ref{assum:Apop:spectrum},~\ref{assum:degrees},~\ref{assum:latentpositions},~\ref{assum:stronger4dhat}~and~\ref{assum:peer-oracle}, detailed in the Appendix, there exists a sequence of orthogonal matrices $\Qdes \in \R^{(p+d+2) \times (p+d+2)}$ such that
    \begin{align*}
        \sqrt{n} \paren*{\Qdes \thetahattsls - \symbf \beta} & \to \mathcal{N} \paren*{\symbf 0, \sigmaeps^2 \paren*{\ZlatOracle^\top \MlatOracle \ZlatOracle}^{-1}} .
    \end{align*}
\end{theorem}

\begin{theorem}[Peer contagion estimators under latent contagion]
    \label{thm:lattruepeerfit}
    Suppose the data are generated according to latent contagion as in Equation~\eqref{eq:lim-latent-red} but we use peer contagion estimators from Definition~\ref{def:feasibleEstimators:peer}. Suppose $\A$ follows a sub-gamma model as in Definition~\ref{def:subgamma-network}. Under Assumptions~\ref{assum:Apop:spectrum},~\ref{assum:degrees},~\ref{assum:latentpositions}~and~\ref{assum:latent-oracle}, detailed in the Appendix, there exists a sequence of orthogonal matrices $\Qdes \in \R^{(p+d+2) \times (p+d+2)}$ such that
    \begin{align*}
        \sqrt{n} \paren*{\Qdes \betahattsls - \symbf \theta} & \to \mathcal{N} \paren*{\symbf 0, \sigmaeps^2 \paren*{\ZpeerOracle^\top \MpeerOracle \ZpeerOracle}^{-1}} .
    \end{align*}
\end{theorem}

Proofs of Theorems~\ref{thm:peertruepeerfit},~\ref{thm:lattruelatfit},~\ref{thm:peertruelatfit} and~\ref{thm:lattruepeerfit} can be found in the Appendix. The general proof strategy is to show that $\betahattsls$ and $\thetahattsls$ are close to ``oracle'' estimates based on using the true latent positions. Appendix~\ref{apx:oracle-convergence} discusses convergence of these oracle estimators to the true parameters. Appendices~\ref{apx:latcon} and~\ref{apx:peercon} show convergence of our estimators defined above to these oracle estimators. We include, additionally, proofs showing that the feasible ordinary least squares estimators converge to their oracle counterparts, although oracle ordinary least squares estimators are only consistent and asymptotically normal in dense networks \citep{lee2002,gupta2019}.

Our misspecification results generalize to a statement about asymptotic equivalence of the corresponding functionals. Recall that the population design matrices of the peer and latent contagion models are, respectively,
\begin{equation} \label{eq:def:Zmxs}
    \ZpeerOracle = \begin{bmatrix} \1_n \, \, \W \, \, \Xpop \, \, \G \y \end{bmatrix}
    ~\text{ and }~
    \ZlatOracle  = \begin{bmatrix} \1_n \, \, \W \, \, \Xpop \, \, \Gtilde \y. \end{bmatrix} .
\end{equation}
Consider the corresponding ``projection parameters'', or the population coefficients corresponding to projection of outcomes onto these design matrices, and letting $F$ be an appropriately supported cumulative distribution function (see \citealt{li2025h} for a discussion of these parameters in the network regression context),
\begin{align}
    \label{eq:proj:beta}
    \tau (F)             & = \argmin_{b} \bbE_F \brac*{\paren*{\Y - \ZpeerOracle \, b}^2} = \bbE_F \brac*{\ZpeerOracle^\top \ZpeerOracle}^{-1} \bbE_F \brac*{\ZpeerOracle^\top \Y}, \text{ and } \\
    \label{eq:proj:theta}
    \widetilde{\tau} (F) & = \argmin_{t} \bbE_F \brac*{\paren*{\Y - \ZlatOracle \, t}^2} = \bbE_F \brac*{\ZlatOracle^\top \ZlatOracle}^{-1} \bbE_F \brac*{\ZlatOracle^\top \Y}.
\end{align}
Let $F_\beta$ denote the law of the peer contagion process and $F_\theta$ the law of the latent contagion process. When the expectation is taken with respect to the true generating process, the projection parameters simplify to the corresponding regression coefficients,
\begin{equation} \label{eq:def:projectionParams}
    \begin{aligned}
        \tau (F_\beta)
        = \symbf \beta, ~~~\text{ and }~~~
        \widetilde{\tau} (F_\theta)
        = \symbf \theta.
    \end{aligned} \end{equation} %
These projection parameters are not, in general, equal under the latent and peer contagion models. However, fixing either the peer or latent contagion model, the following theorem shows that they are asymptotically equivalent, in the sense that the difference between the two parameters is negligible relative to the typical estimation error.

\begin{theorem} \label{thm:proj-equivalence}
    Under Assumptions~\ref{assum:latent-oracle},~\ref{assum:equiv:limcov},~\ref{assum:equiv:edges}~and~\ref{assum:equiv:LandX}, and supposing the latent contagion model in Equation~\eqref{eq:lim-latent-red} holds, then
    \begin{align*}
        \norm*{\widetilde{\tau} (F_\theta) - \tau (F_\theta)}
        = \norm*{\, \symbf \theta - \tau (F_\theta)}
        = \norm*{\bbE \brac*{\ZlatOracle^\top \ZlatOracle}^{-1} \bbE \brac*{\ZlatOracle^\top \Y} - \bbE \brac*{\ZpeerOracle^\top \ZpeerOracle}^{-1} \bbE \brac*{\ZpeerOracle^\top \Y}}
        = o \paren*{n^{-1/2}}.
    \end{align*}
    If the peer contagion model in Equation~\eqref{eq:lim-peer-red} holds instead of the model in Equation~\eqref{eq:lim-latent-red}, then, with Assumption~\ref{assum:peer-oracle} in place of Assumption~\ref{assum:latent-oracle},
    \begin{align*}
        \norm*{\tau (F_\beta) - \widetilde{\tau} (F_\beta)}
        = \norm*{\, \symbf \beta - \widetilde{\tau} (F_\beta)}
        = \norm*{\bbE \brac*{\ZpeerOracle^\top \ZpeerOracle}^{-1} \bbE \brac*{\ZpeerOracle^\top \Y} - \bbE \brac*{ \ZlatOracle^\top \ZlatOracle}^{-1} \bbE \brac*{\ZlatOracle^\top \Y}}
        = o \paren*{n^{-1/2}}.
    \end{align*}
\end{theorem}
A proof can be found in Appendix~\ref{apx:equivalence}. In the above theorem, we use the ``assumption lean'' representation of regression coefficients, or the ``projection estimands'', which are the definitions of $\tau$ and $\widetilde{\tau}$ given in \eqref{eq:proj:beta} and \eqref{eq:proj:theta}, respectively. These functionals are non-parametric. That is, $\widetilde{\tau}$ is a projection parameter that one might want to estimate even if the true data generating model is the peer contagion model and the model is indexed by $\symbf \beta$, with $\symbf \theta$ undefined. The key idea of Theorem~\ref{thm:proj-equivalence} is that the true parameter and the projection parameter are within parametric estimation error of one another, so they are functionally indistinguishable in the asymptotic limit.

This result has substantial implications for estimating contagion effects in noisy networks. Suppose that outcomes are generated according to the peer contagion model in Equation~\eqref{eq:lim-peer}, but the network $\A$ is unobserved and one only has access to a noisy variant $\Anoise$. In this case, treating $\Anoise$ as the true network and plugging it into typical estimators such as $\betahattsls$ will ignore measurement error in $\Anoise$ and typically lead to inconsistent estimation, among other issues.
In particular, it is challenging to target the parameter $\symbf{\beta}$, because this requires projecting onto $\ZpeerOracle$, which is a function of an unknown network $\A$. Theorem~\ref{thm:proj-equivalence}, however, suggests a path forward: we can instead target the parameter $\widetilde{\tau}$, which does not equal $\symbf{\beta}$, but is functionally indistinguishable. Crucially, to estimate $\widetilde{\tau}$ we do not need to project outcomes onto $\ZpeerOracle$, but rather onto $\ZlatOracle$, which
depends on the low-rank expectation $\Apop$ rather than the precise network $\A$. This is a substantially easier task: we do not need $\A$ proper, only its principal subspace. Conveniently, given a noisy observation  of a network $\Anoise$, there are many off-the-shelf methods to estimate the principal subspace of $\A$.

This leads to a potentially generic recipe for estimating contagion effects in noisy networks: find a subspace estimator tailored to the noise process, and target $\widetilde{\tau}$ using an estimator like $\thetahattsls$. The variance of the corresponding estimator will depend on how fast the subspace estimator converges. For sufficiently fast estimators, such as the adjacency spectral embedding, there is no asymptotic variance penalty due to estimating the principal subspace.
In fact, this approach and our results thus far are sufficient to develop an estimator for contagion effects in noisy weighted networks. Suppose that $\Anoise = \A + \mE$, where $\mE$ is mean-zero sub-gamma noise. In this case, $\Anoise$ still satisfies the assumptions of the subgamma network model (Definition \ref{def:subgamma-network}), and thus the adjacency spectral embedding is still a consistent estimator of the principal subspace of $\A$, with only a slight adjustment to the sub-gamma parameters in the convergence rate. Thus, $\thetahattsls$ immediately accommodates sub-gamma noise in the adjacency matrices.

Figure~\ref{fig:concentration} offers some visual intuition for the generic nature of the network smoothing strategy, showing estimates of $\Apop$ obtained from networks with various types of corruption. Even when 10\% of edges are missing, matrix completion methods produce estimates $\Apophat$ that closely approximate the true latent structure.
In the subsequent section, we provide evidence via simulations that network smoothing is a viable strategy to estimate contagion in networks with missing edges, in ego-centrically sampled networks, and in networks where only aggregated relational data is available.

\begin{figure}[htbp]
    \centering
    \includegraphics{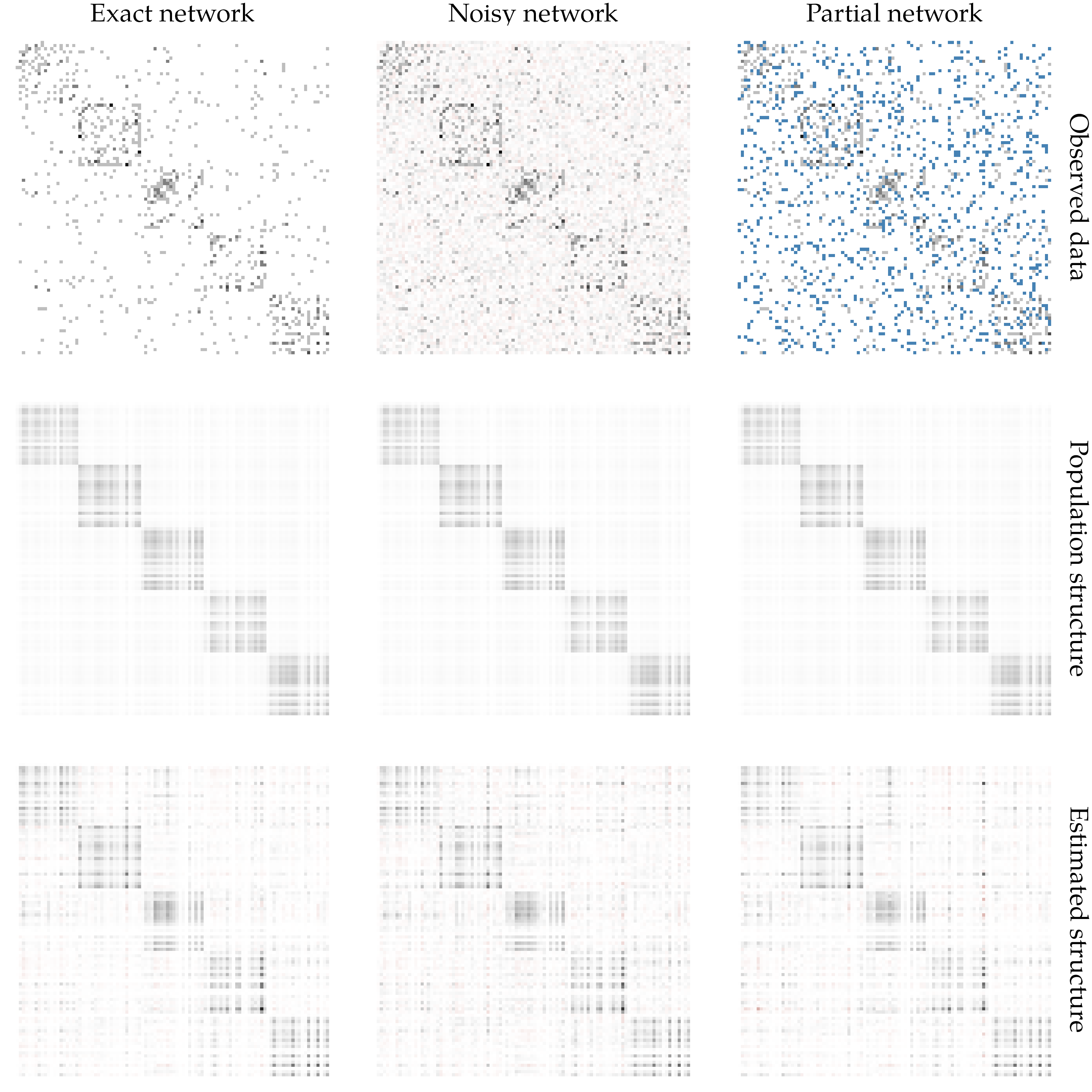}
    \caption{Adjacency matrices, corresponding expectations, and estimated expectations for observed network data corresponding to the same underlying population structure in a random dot product graph, as well as various estimates of the latent structure, based on the adjacency spectral embedding in the left and middle columns, and nuclear-norm based matrix completion in the right column. Blue entries in the right column indicate unobserved entries of $\A$.}
    \label{fig:concentration}
\end{figure}

\section{Simulation study}
\label{sec:simulations}

We now verify via simulation that $\thetahattsls$ and $\betahattsls$ are consistent and asymptotically normal estimates of contagion effects (Fig.~\ref{fig:mse}), and that they obtain the expected $n^{-1/2}$ convergence rate predicted by our theoretical results. All networks in our simulations below are generated from a degree-corrected stochastic blockmodel (Definition~\ref{def:sbm}) with $n$ nodes and five equally probable blocks. In particular, we use a Poisson stochastic blockmodel where edges are Poisson distributed.

\begin{figure}[t]
    \centering
    \includegraphics{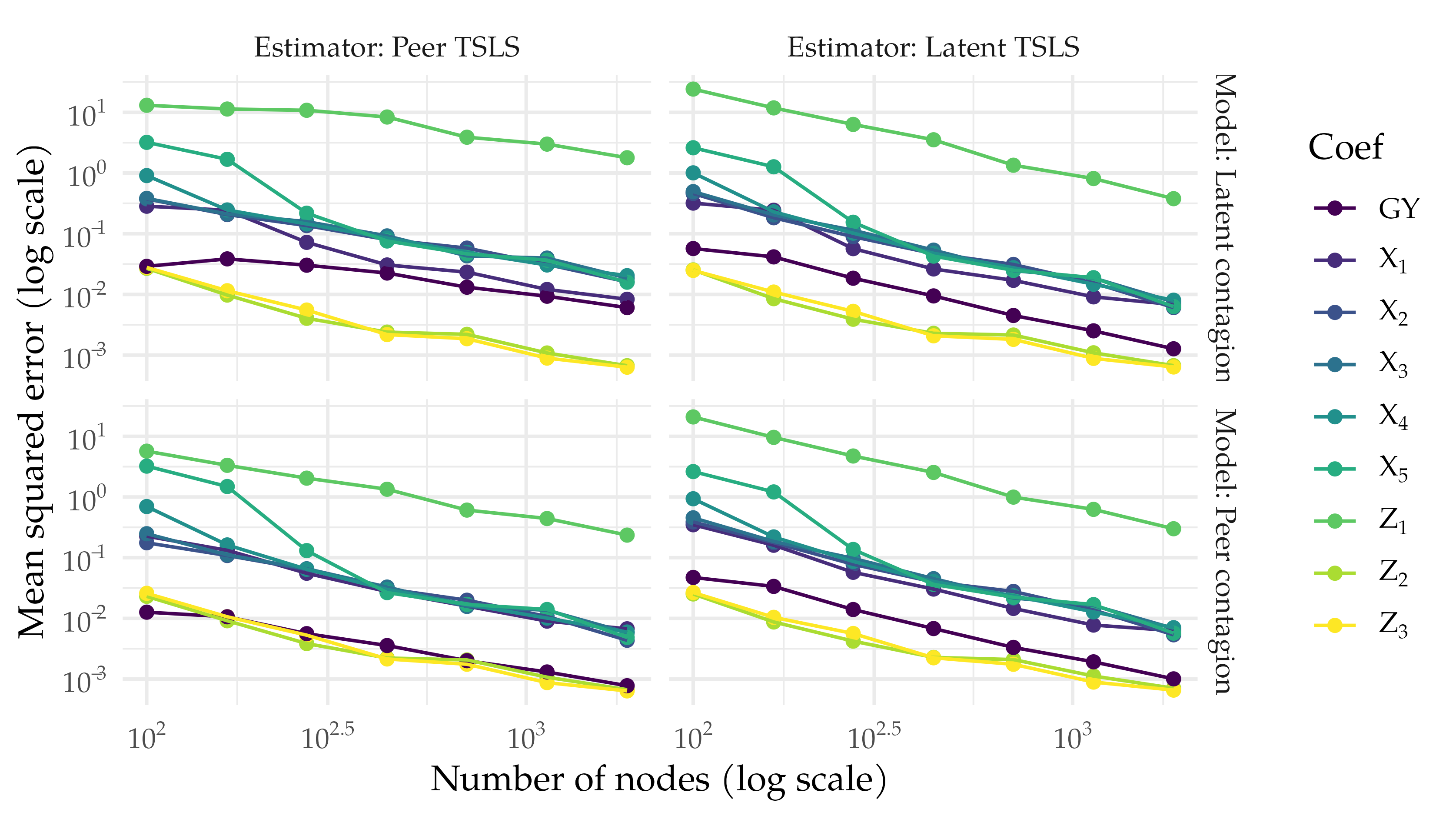}
    \caption{Monte Carlo estimates of mean squared error of $\thetahattsls$, and $\betahattsls$.
        The top panels consider the setting where latent contagion is the true data generating process, and the bottom panels considers the setting where peer contagion is the true data generating process. The left panels display results the peer contagion estimator, and the right panels display results for the latent contagion estimator. In each panel, thex-axis represents the number of nodes in the network on a log scale, and the y-axis represents the mean-squared error of the estimate, also on a log scale. Each line corresponds to a different regression coefficient, differentiated by color. The average degree in these simulations is $n^{3/4}$. Additional simulation results when the average degrees is $n^{1/2}$ and $n^{1/4}$ are available in Figures~\ref{fig:mse-n12}~and~\ref{fig:mse-n14}.}
    \label{fig:mse}
\end{figure}

\begin{definition}[Poisson Degree-Corrected Stochastic Blockmodel] \label{def:sbm}
    The Poisson degree-corrected stochastic blockmodel \citep{rohe2018, karrer2011} is an undirected model of community membership, with $d$ communities. Each node, indexed by $i \in [n]$, is assigned a block $z_i \in [d]$ with probability $\Pr({z_i = k}) = \pi_k$, and a degree-correction parameter $\xi_i$, which describes the propensity of vertex $i$ to connect with other nodes. Conditional on block memberships and degree-correction parameters, edges are generated independently between every pair of vertices in the network according to a Poisson distribution with parameter $\rho_n \, \xi_i \symbf B_{z_i, z_j} \xi_j$. That is, the expected number of edges between two vertices depends on their community memberships, their degree correction parameters, a positive semi-definite matrix $\symbf B \in [0, 1]^{d \times d}$ of inter-block edge formation probabilities, and a scaling factor $\rho_n \in [0, 1]$, which may vary with $n$.
\end{definition}
For our simulation study, we take there to be $d = 5$ blocks, and sample diagonal elements of $\symbf B$ from a $\mathrm{Uniform}(0.75, 0.85)$ distribution and off-diagonal elements of $\symbf B$ from a $\mathrm{Uniform(0.01, 0.05)}$ distribution, such that networks are strongly assortatively, mostly forming edges within blocks. We sample degree correction parameters according to $\xi_i \sim \mathrm{Exponential(1/3)} + 1$. Shifting the distribution of $\xi_i$ away from zero limits the number of isolated nodes. The sparsity parameter $\rho_n$ is set so that the expected mean degree of the network is $n^{3/4}, n^{1/2}$ or $n^{1/4}$. Once we have generated these parameters, we compute $\Xpop = \Upop \Spop^{1/2}$ where $\Upop \Spop \Upop^\top$ is the eigendecomposition of the low-rank expectation $\E[z_1,z_2,\dots, z_n, \theta]{A}$.
Errors $\be$ are sampled from a standard normal distribution, and we include three covariates $\W_1, \W_2, \W_3$, also sampled from a standard normal. We set $\betanaught = \thetanaught = 0$, $\betay = \thetay = 0.2$, $\betaw = \thetaw = (5, 5, 5)$ and $\betax = \thetax = (2, 2, 2, 2, 2)$. Then $\Y$ is generated according to Equation~\eqref{eq:lim-peer-red} or Equation~\eqref{eq:lim-latent-red}, depending on whether we are under peer or latent contagion, respectively.
For both models, we compute $\betahattsls$ and $\thetahattsls$, and measure the estimation error to the corresponding model coefficients. Since $\betax$ is only identified up to orthogonal rotation, we perform Procrustes alignment \citep[see][for discussion]{cape2019c} between $\Xhat$ and $\Xpop$ to investigate recovery of $\betax$. Note that we set $\betanaught = \thetanaught = 0$ to simplify this alignment step.

In our experiments, we vary the sample size $n$ (i.e., the number of vertices) on a logarithmic scale, considering $n \in \set{100, 163, 264, 430, 698, 1135, 1845, 3000}$, and replicate our experiments $100$ times for each simulation setting.
Figure~\ref{fig:mse} shows the mean squared error of the estimated coefficients as a function of the number of nodes. Mean squared error for $\thetahattsls$ and $\betahattsls$ decreases at $n^{-1/2}$ rates in both the latent and peer contagion models, exactly as dictated by our theory.

\subsection{Noisy networks}
\label{subsec:noisy-networks}

We additionally investigate how network smoothing performs when the network is mismeasured. We consider the exact same simulation setting as before, but now we observe a noisy adjacency matrix $\Anoise$ rather than $\A$. We consider only the latent contagion estimator $\thetahattsls$, since the latent space formulation allows us to smooth away noise in the network. In some cases, estimating $\Upop$ and $\Spop$ via the singular value decomposition is impossible because of the noise process (for instance, due to missing data). When this is the case, or when there is an estimator of $\Upop$ and $\Spop$ designed to handle the particular form of the noisy matrix $\A$, we use that more appropriate estimator instead of the singular value decomposition. The estimators are introduced below, alongside the various noise processes under consideration. In the first four cases, existing estimators are capable of recovering the singular values and singular vectors of $\Apop$. In the last two settings, we are unaware of principal subspace estimators, and expect difficulties with estimation.

\begin{enumerate}
    \item \textbf{Baseline}: The network is observed without noise, so $\Anoise = \A$ and $\Apophat$ is estimated by taking the rank $k$ singular value decomposition of $\Anoise$. $\Xhat$ is constructed via $k$-dimensional ASE applied to $\Apophat = \Uhat \Shat \Uhat$, so that $\Xhat = \Uhat \Shat^{1/2}$. $\Ghat$ is constructed by row-normalizing $\Apophat$ and the setting the diagonal elements to zero. $\Xhat$ and $\Ghat$ are then plugged in as estimates of $\Xpop$ and $\Gtilde$ in $\thetahattsls$.

    \item \textbf{Gaussian noise}: We observe $\Anoise = \A + \symbf E$, where $\symbf E$ is a symmetric matrix with i.i.d. entries $\varepsilon_{ij} \sim \mathcal{N}(0, \sigma^2)$ for $i < j$ and $\varepsilon_{ii} = 0$. We compute $\Apophat$ as the rank-$k$ truncated singular value decomposition of $\Anoise$, and construct $\Ghat$, $\Xhat$ and then $\thetahattsls$ as in the baseline case. \cite{levin2022a} justifies this singular value decomposition as applied to $\Anoise$ rather than $\A$.

    \item \textbf{Missing edges}: We observe $\Anoise = \A$, except $30\%$ of entries of $\Anoise$ are missing at random, with missingness independent of the network. $\Upop$ and $\Spop$ are estimated via matrix completion, in particular the \texttt{AdaptiveImpute} algorithm of \citet{cho2019}. Estimated singular values and vectors are directly substituted for the more typical $\Uhat$ and $\Shat$ from singular value decomposition in the known $\A$ case, which in turn allows computation of $\Apophat, \Ghat$ and $\Xhat$ in the same fashion as the previous two settings.

    \item \textbf{Aggregated relational data}: We observe $\Anoise = \A \W$ where $\W$ are traits (also observed) that are correlated with latent $\Xpop$. $\W$ has the same dimensions as $\Xpop$, and each column of $\W$ is sampled from a multivariate normal distribution with correlation 0.8 to the corresponding column of $\Xpop$. In the aggregated relational data setting, $\Upop$ is estimated by the left singular vectors of $\Anoise$, and we denote these estimates $\bar{\Upop}$. $\Spop$ is then estimated via $\bar{\Upop}^\top \Y \W^\top \bar{\Upop} (\bar{\Upop}^\top \W \W^\top \bar{\Upop})^{-1}$, and these estimates are used as plug-in replacements for $\Uhat$ and $\Shat$. Theory and motivation for these estimators is under development in forthcoming work by Hayes, Chandrasekhar, McCormick and Breza. The aggregated relational data framework is detailed in \citet{breza2020, breza2023a}.

    \item \textbf{Ego-centric data}: In ego-centric sampling, half of the nodes in the network are selected at random, and only edges incident to those nodes are observed. Unobserved edges are imputed using Algorithm 1 of \citet{chan2023a}, and then $\Apop$ is estimated via the full matrix recovery technique presented in the same manuscript. Let $\A_{11}$ be the ego-ego block and $\A_{12}$ be the ego-nonego block. We compute a rank-$k$ approximation $\tilde{\Apop}_{11} = \Upop_1 \symbf \Spop_1 \Vpop_1^\top$ of $\A_{11}$, and then estimate the nonego-nonego block as $\Apophat_{22} = \A_{12}^\top \tilde{\Apop}_{11}^\dagger \A_{12}$. The remaining blocks $\Apophat_{11}$ and $\Apophat_{12}$ are recovered via a rank-$k$ approximation of the observed part of the network $\begin{bmatrix} \A_{11} & \A_{12} \end{bmatrix}$.

    \item \textbf{Degree capped}: The network is observed with censored degrees, where for each node $i$, at most $d_{\text{max}} = 20$ incident edges are available. If a node has more than $20$ edges, edges are removed uniformly at random until the constraint is met. We compute $\Apophat$ as the rank-$k$ truncated singular value decomposition of the resulting censored adjacency matrix $\Anoise$.

    \item \textbf{Edges flipped}: $\Anoise$ is a version of $\A$ when $15\%$ of edges in the network have been flipped. To preserve the total number of edges in the network, this is implemented via random edge swapping, where $15\%$ of edges in the network are swapped with a random edge with different edge value. $\Upop$ and $\Spop$ are estimated via singular value decomposition of $\Anoise$.
\end{enumerate}

The results of these simulations are in Figure~\ref{fig:mse-noisy}, which compares mean squared error for $\betay$ and $\thetay$ when using the estimates of $\Xpop$ and $\Gtilde$ defined above. We see that the network smoothing approach is able to recover both $\betay$ and $\thetay$ in settings where $\A$ is contaminated with Gaussian noise, when $\A$ has missing edges, and when only an aggregated relational data variant of $\A$ is observed. This matches our expectations that consistent estimation of peer effects is possible whenever principal subspace estimation is possible. Under degree censoring and edge flip noise mechanisms, principal subspace estimation is challenging, and the singular value decomposition is a poor estimator, such that $\betay$ and $\thetay$ cannot be recovered accurately.

\begin{figure}[t]
    \centering
    \includegraphics{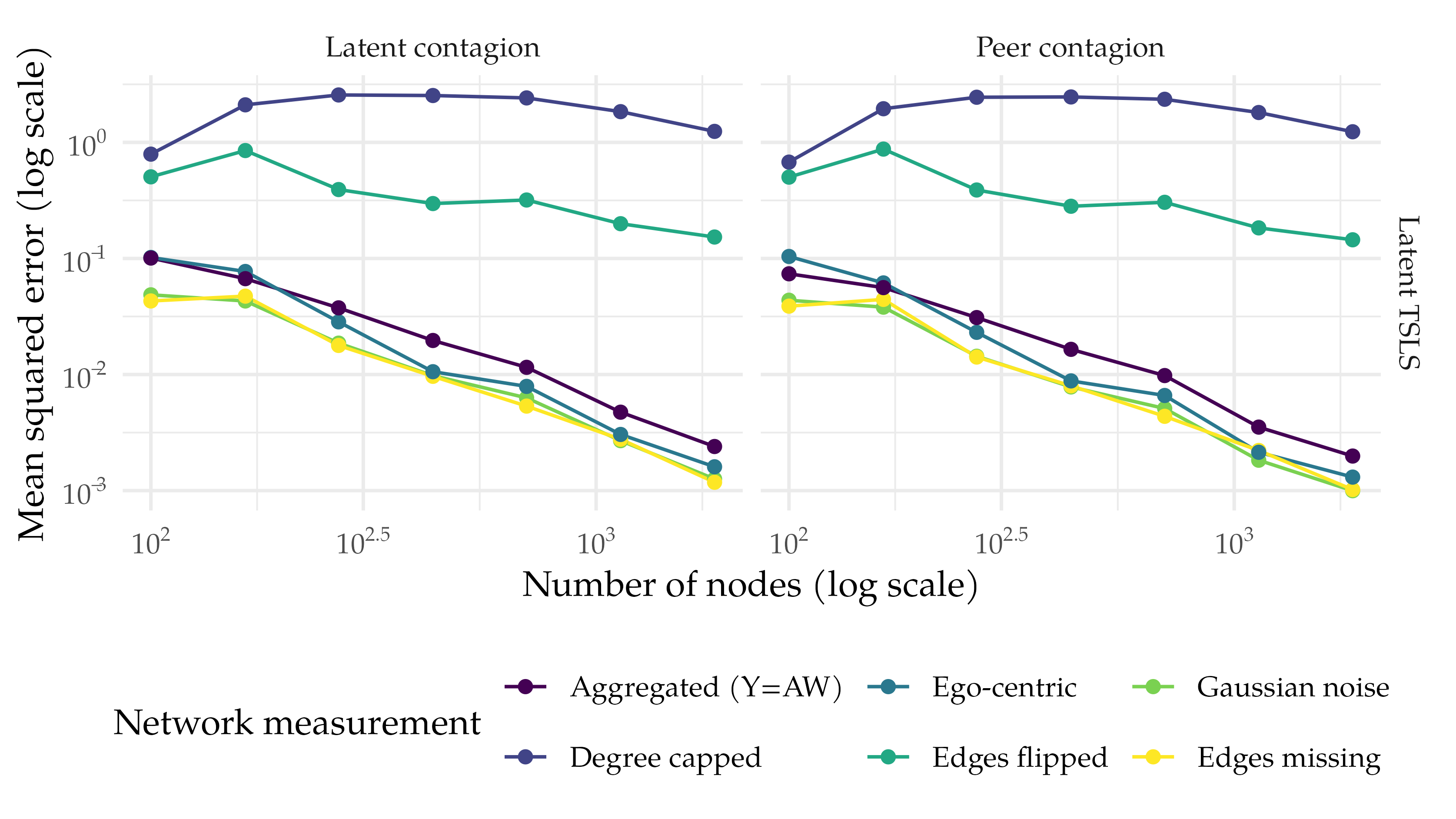}
    \caption{Monte Carlo estimates of mean squared error of $\thetahattsls_y$ under various noise models for the network. The left panel considers the setting where latent contagion is the true data generating process, and the right panel considers the setting where peer contagion is the true data generating process. The x-axis represents the number of nodes in the network on a log scale, and the y-axis represents the mean-squared error of the estimate of the contagion term, $\theta_y$ under latent contagion and $\beta_y$ under peer contagion, and is also on log scale. Each line corresponds to simulations under a different form of measurement error, differentiated by color. The average degree in these simulations is $n^{3/4}$.}
    \label{fig:mse-noisy}
\end{figure}

\subsection{Sparsity}

We additionally investigate how $\thetahattsls$ and $\betahattsls$ perform in sparser, correctly observed networks. In Figure~\ref{fig:mse-n12} we report mean squared error when the average degree is $n^{1/2}$, and in Figure~\ref{fig:mse-n14} we report mean squared error when the average degree is $n^{1/4}$. In both cases, we observed worse performance relative to the denser $n^{3/4}$ average degree case. When the average degree is $n^{1/2}$, mean squared error is more volatile, but Figure~\ref{fig:mse-n12} does still suggest that $\thetahattsls$ and $\betahattsls$ are consistent under peer contagion. Under latent contagion, $\thetahattsls$ appears consistent, but $\betahattsls$ struggle to recover the intercept (coefficient $Z_1$) and the peer effect. We suspect that this is primarily due to finite sample collinearity issues, as degree heterogeneity is necessary to differentiate the intercept from the contagion term, and degree heterogeneity is less pronounced in sparser graphs.
In the sparsest setting, with average degree $n^{1/4}$, the impact of sparsity is more severe. Figure~\ref{fig:mse-n14} suggests that all estimators experience slower rates of convergence and may not even be consistent, under both models, although $\betahattsls$ is potentially consistent under the peer contagion model. We believe that this degradation in performance is primarily attributable to noise in the estimates $\Xhat$ around $\Xpop$.

These simulations show that the sparsity can degrade the performance of the estimators that we have proposed, via two mechanisms: sparsity can reduce degree heterogeneity, leading to collinearity issues, and it may degrade the accuracy of the plug-in estimate $\Xhat$ of $\Xpop$. These simulations suggest that applied analyses using our estimators should assess the quality of estimates $\Xhat$ (namely, assess whether the estimates $\Xhat$ correspond to meaningful structure in the network) and should consider the possibility of variance inflation in regression estimates due to collinearity. In sparse networks with strong block diagonal structure and limited degree heterogenity, the intercept and $\Xpop$ may be collinear, and it may be reasonable to drop a column for collinearity reasons.

\begin{figure}[t]
    \centering
    \includegraphics{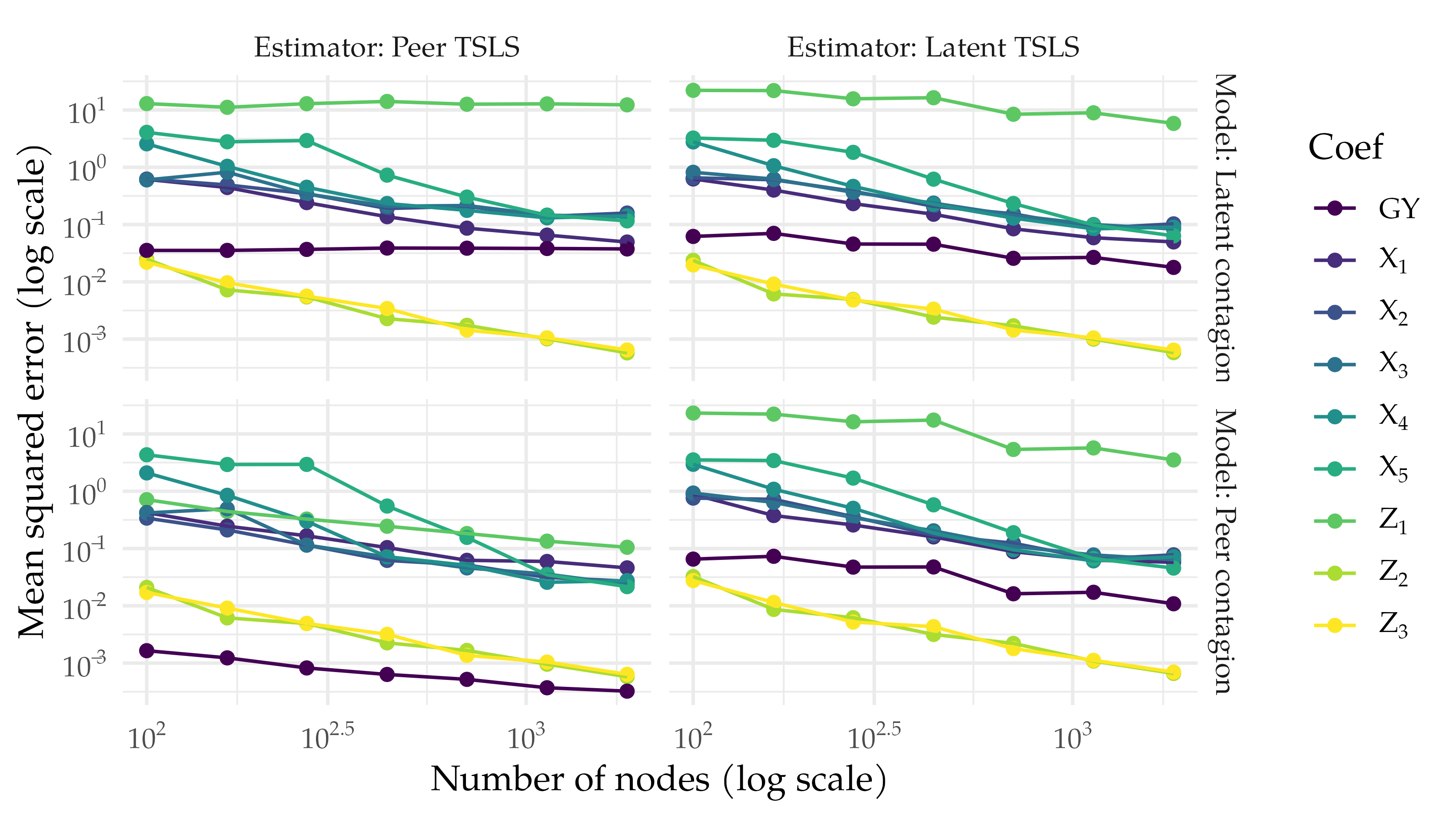}
    \caption{Monte Carlo estimates of mean squared error of $\thetahattsls_y$ under various noise models for the network. The left panel considers the setting where latent contagion is the true data generating process, and the right panel considers the setting where peer contagion is the true data generating process. The x-axis represents the number of nodes in the network on a log scale, and the y-axis represents the mean-squared error of the estimate of the contagion term, $\theta_y$ under latent contagion and $\beta_y$ under peer contagion, and is also on log scale. Each line corresponds to simulations under a different form of measurement error, differentiated by color. The average degree in these simulations is $n^{1/2}$.}
    \label{fig:mse-n12}
\end{figure}

\begin{figure}[t]
    \centering
    \includegraphics{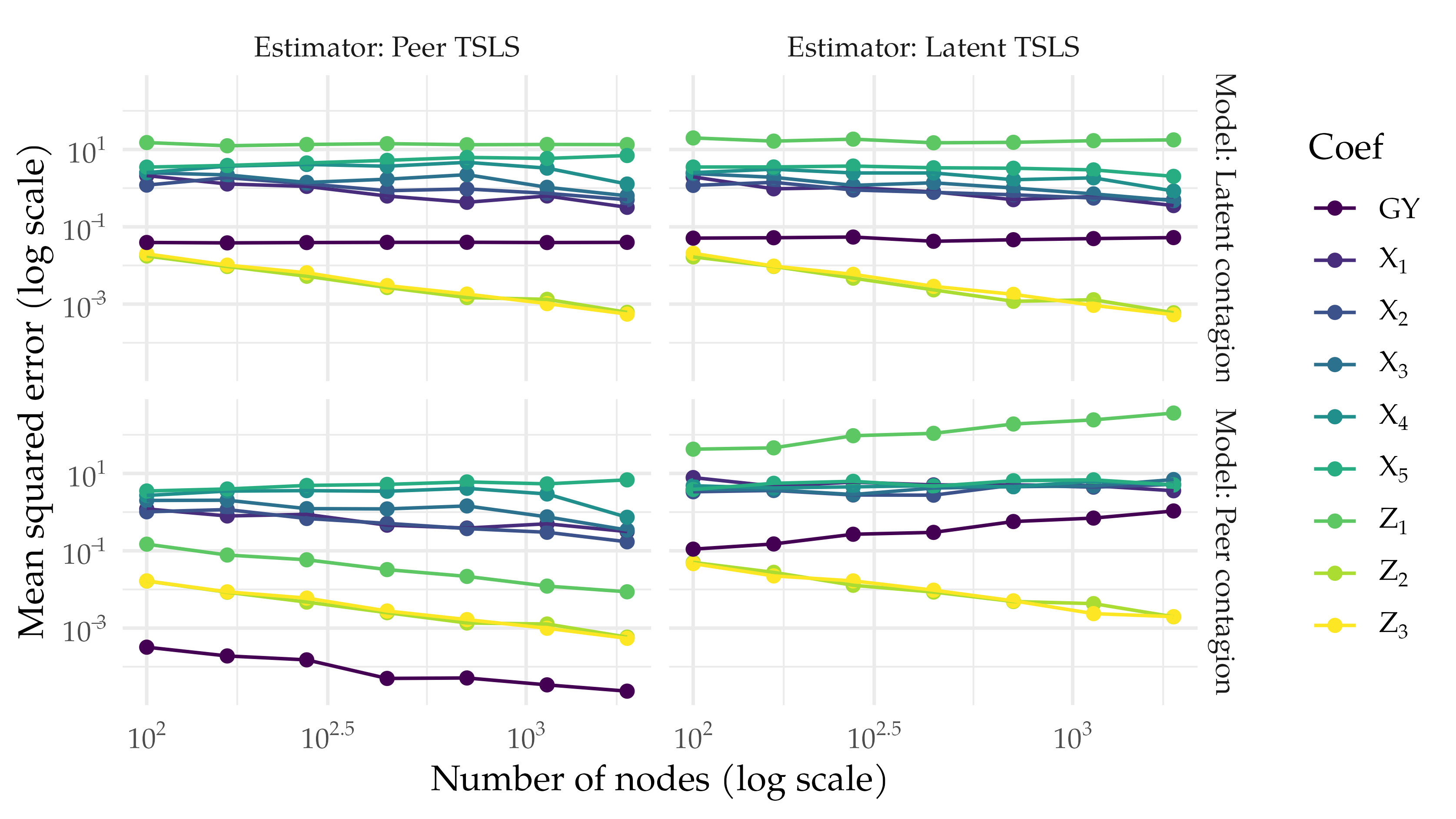}
    \caption{Monte Carlo estimates of mean squared error of $\thetahattsls_y$ under various noise models for the network. The left panel considers the setting where latent contagion is the true data generating process, and the right panel considers the setting where peer contagion is the true data generating process. The x-axis represents the number of nodes in the network on a log scale, and the y-axis represents the mean-squared error of the estimate of the contagion term, $\theta_y$ under latent contagion and $\beta_y$ under peer contagion, and is also on log scale. Each line corresponds to simulations under a different form of measurement error, differentiated by color. The average degree in these simulations is $n^{1/4}$.}
    \label{fig:mse-n14}
\end{figure}

\section{Data Application}
\label{sec:data-application}

To demonstrate our methods, we applied our estimators to network data collected during the \emph{Teenage Friends and Lifestyle Study}, reported in \citet{michell1996}, \citet{michell1997}, \citet{michell1997a}, and \citet{michell2000a}. Recently, \citet{hayes2025} and \citet{dimaria2022a} investigated network-mediation using this data, studying how sex influenced network position, and how network position subsequently influenced smoking behaviors in adolescents. These analyses assumed that there were no peer effects on smoking after accounting for latent network position. We re-analyzed the same data to investigate if smoking exhibits spillovers in addition to being localized within the adolescent social network.

\subsection{Data}

The \emph{Teenage Friends and Lifestyle Study} collected three waves of survey data in a secondary school in Glasgow, beginning in January 1995. Students in the study filled out a questionnaire about their lifestyle and risk-taking behaviors, including alcohol, tobacco and drug use, and additionally were asked to list six of their friends. \citet{michell2000a} found that smoking was mostly concentrated in friend groups composed of popular girls, unpopular students, and trouble-makers: ``risk taking behaviour was heavily polarized within social categories so that, for instance, groups of individuals (and their peripherals) were in general either risk-taking or non-risk-taking''.

The social network was collected by asking students ``who are your best friends'', and allowing adolescents to list up to six responses. We considered only data from the first wave of the survey, which included 153 adolescents. Sex and tobacco use were self-reported as nominal features with levels ``Male'' and ``Female''; and ``Never'', ``Occasional,'' and ``Regular,'' respectively. To match the analyses of \citet{hayes2025} and \citet{dimaria2022a}, for the tobacco use measure, we combined ``Occasional'' and ``Regular'' into a single level, and compared smokers with non-smokers. Also to match the analysis of \citet{hayes2025}, we treated age (continuous) and church attendance (nominal) as possible confounders and thus included these variables as controls.

We computed the adjacency spectral embedding of the social network $\A$. In the Glasgow data, the social network is directed: an edge $i \to j$ indicates that student $i$ listed student $j$ as friend. This directedness means that students have two distinct co-embeddings corresponding to their propensity to send out-edges and receive in-edges. Letting $\widehat{\A} \approx \Uhat \Shat \Vhat^T$ be the truncated singular value decomposition of $\A$, the left co-embedding $\widehat{\mL} = \Uhat \Shat^{1/2}$ described how students in the network send edges (i.e., claim friends), and the right co-embedding $\Xhat \equiv \Vhat \Shat^{1/2}$ described how students receive edges (i.e., are claimed as friends). Our results used the right co-embeddings $\Xhat$. We did not select any particular dimension $d$ for the latent space. Instead, we repeated our analysis for many values of $d$, to investigate the sensitivity of our results to the dimension of the latent space.
Once we obtained embeddings $\Xhat$ via the singular value decomposition, we performed a multiverse analysis, estimating regression coefficients using $\thetahattsls$ and $\betahattsls$. For each estimator, we considered two variants, one including $\Xhat$ as covariates, to adjust estimates for latent positions in the network, and one without including $\Xhat$ as covariates. That is, we used two-stage least squares estimators derived under all four of the following generative models:
\begin{align}
    \label{eq:peer-no-x} \Y  & = \1_n \betanaught + \W \betaw + \G \Y \betay  + \be                        \\
    \label{eq:latent-no-x}\Y & = \1_n \thetanaught + \W \thetaw + \Gtilde \Y \thetay + \be                 \\
    \label{eq:peer-x}\Y      & = \1_n \betanaught + \W \betaw + \Xpop \betax + \G \Y \betay  + \be         \\
    \label{eq:latent-x}\Y    & = \1_n \thetanaught + \W \thetaw + \Xpop \thetax + \Gtilde \Y \thetay + \be
\end{align}
Recall that $\W$ consists of age and church attendance. The estimators $\thetahattsls$ required an estimate of the dimension $d$ of the latent positions $\Xpop$, as does $\betahattsls$ when $\Xpop$ are included as covariates in the model. We repeated our analysis for $2 \le d \le 25$ in order to understand how the dimension of the embedding affected estimates. For each estimate, we reported a 95\% asymptotic confidence interval in Figure~\ref{fig:glasgow-estimates}.

\begin{figure}[t]
    \centering
    \includegraphics{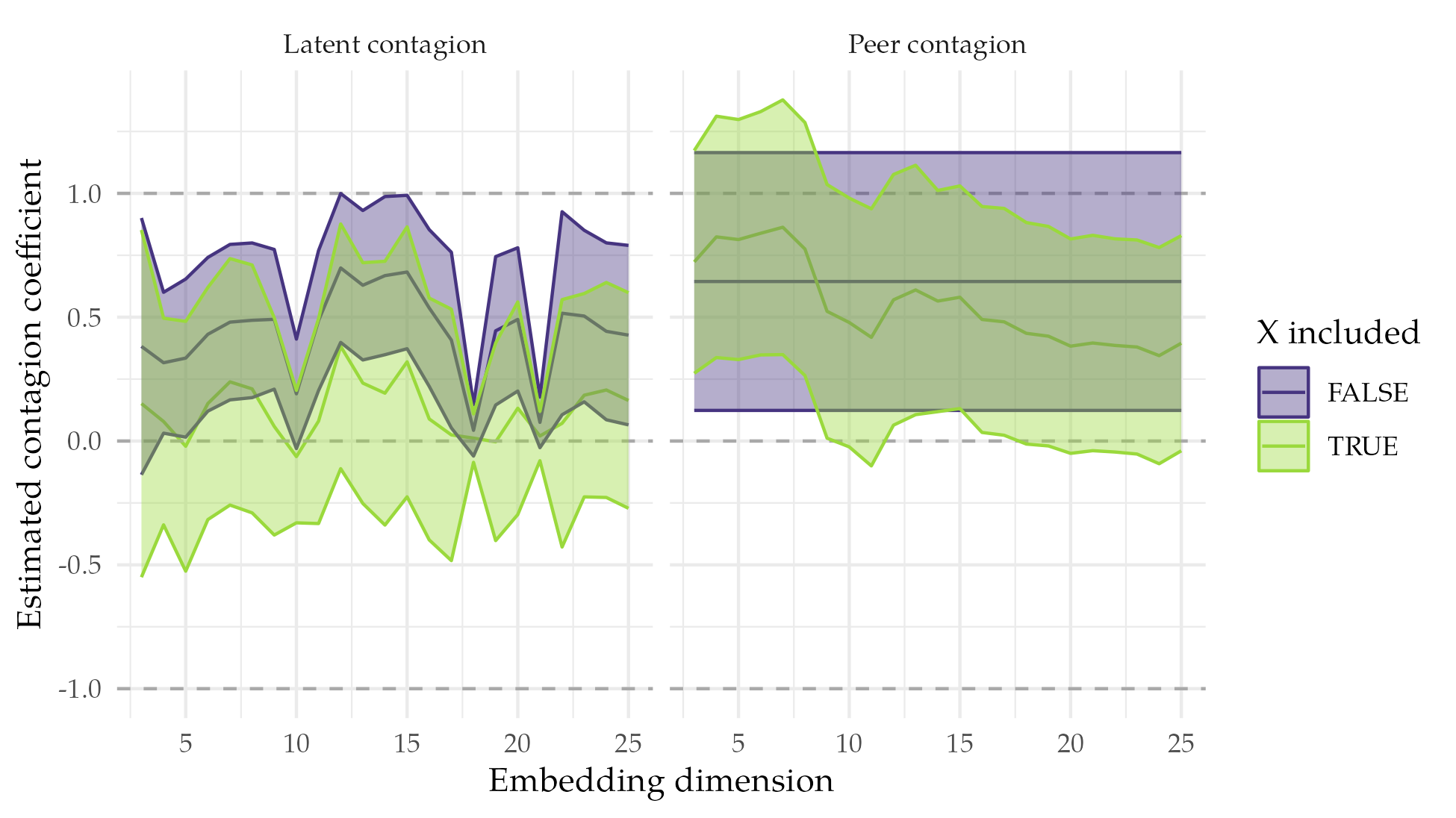}
    \caption{Point estimates and asymptotic 95\% confidence intervals for $\betay$ and $\thetay$ in the Glasgow adolescent social network. $\betay$ and $\thetay$ measure the contagiousness of smoking. The left panel considers estimates based on a latent contagion working model, and the right panel considers estimates based on a peer contagion working model. The x-axis represents the dimension of the latent space used in the adjacency spectral embedding. The y-axis is the estimate of the contagiousness of smoking. The dark blue ribbons correspond to estimates that do not include estimated latent positions $\Xhat$ as covariates, and the lighter green ribbons corresponds to estimates that do adjust for the estimated latent positions $\Xhat$.}
    \label{fig:glasgow-estimates}
\end{figure}

\subsection{Results}

Estimates of the smoking spillovers varied moderately across the multiverse analysis. The simplest and most consistent story emerged among estimates that do not adjust for latent homophily by including $\Xpop$ in the regression specification. These estimates, visualized in blue, were consistently large and statistically differentiated from zero, suggesting that smoking does exhibit spillover effects. However, estimates that adjusted for latent homophily by including $\Xpop$ told a different story. In most cases, including $\Xpop$ in the regression reduced the point estimate of the contagion coefficient, and the associated confidence interval often contained zero. However, these confidence intervals remained wide, indicating substantially uncertainty about the presence or absence of a spillover effect after accounting for localized smoking behavior in the network via the $\Xhat$ terms.

Another set of comparisons, between the latent contagion estimators and the peer contagion estimators, was also informative. In particular, Theorems~\ref{thm:peertruepeerfit},~\ref{thm:lattruelatfit},~\ref{thm:peertruelatfit},~and~\ref{thm:lattruepeerfit} prove that latent and peer contagion estimators are asymptotically equivalent under the random dot product graph, provided that the network is precisely observed. However, we observed some deviation between the latent contagion and peer contagion estimators, with latent contagion estimates indicating lower levels outcome spillover across most embedding dimensions. There numerous possible explanations for the difference between the latent and peer contagion estimates: (1) the peer network might have been observed with noise, (2) a random dot product model may not haven been appropriate for the network (in which case we would prefer the estimates based on the peer contagion model), or (3) the network may have been too small (recall $n = 153$ nodes) for asymptotics results to applicable.

Lastly, we observed that estimates under peer contagion model were fairly stable as a function of the embedding dimension $d$, but estimates under the latent contagion model were somewhat more volatile in the embedding dimension, with estimates of the contagion coefficient collapsing towards for the specific values $d = 10, 18, 21$. We suspect this volatility as a function of embedding dimension was related to the small sample size the limited number of smokers in the network.

Altogether, our analysis indicated that there was substantial uncertainty about the presence and scale of the contagiousness of smoking in the adolescent social network, after accounting for the localized nature of smoking in the network by control from $\Xhat$. This uncertainty mirrored previous results suggesting that peer effects and homophily are challenging to distinguish in social networks \citep{hayes2025b, shalizi2011}. A possible next step to differentiate between these effects in the adolescent social network would be to consider longitudinal models of smoking behavior, such as those proposed in \citep{chang2025b}.

\section{Discussion}
\label{sec:discussion}

We have shown that under low-rank network models, contagion over a true network is asymptotically equivalent to contagion over a smoothed, latent adjacency matrix. This equivalence enables consistent peer effects estimation even when the observed network contains measurement error, provided we can reliably estimate the network's eigenspace. Any method that reliably estimates the eigenspace of $\bbE [\A]$—whether through spectral embedding for sub-gamma noise, matrix completion for missing edges, or debiasing techniques for systematic measurement error—can be combined with our latent contagion framework to estimate peer effects in challenging data settings.

While this paper focuses on parameter estimation rather than causal identification, the equivalence we establish has implications for causal inference under interference. Under appropriate conditional ignorability assumptions, network autoregressive parameters can have causal interpretations \citep{vazquez-bare2023, leung2022, mcfowland2021}. Our results suggest that when peer influence operates through latent structure, causal effects may be more accurately estimated by defining exposures based on latent proximity rather than observed edges.

\subsection*{Acknowledgements}

We thank Hyunseung Kang, Ralph Trane, Karl Rohe, Edward McFowland, Arun Chandrasekhar, Michael Newton, Tyler McCormick, Vincent Boucher, Ben Golub and the attendees of the Network Science in Economics 2026 conference, as well as attendees of the University of Wisconsin--Madison IFDS Ideas Seminar for their helpful comments and suggestions. Support for this research was provided by the University of Wisconsin--Madison, Office of the Vice Chancellor for Research and Graduate Education with funding from the Wisconsin Alumni Research Foundation, as well as NSF grants DMS 2052918 and DMS 2023239. Support for this research was also provided by American Family Insurance through a research partnership with the University of Wisconsin--Madison's American Family Insurance Data Science Institute.

\clearpage
\newpage
\appendix
\section*{Appendix}

Here we collect proofs of our main theoretical results, stated in Theorems~\ref{thm:peertruepeerfit},~\ref{thm:lattruelatfit},~\ref{thm:peertruelatfit} and~\ref{thm:lattruepeerfit}, as well as our projection equivalence result in Theorem~\ref{thm:proj-equivalence}.
In the main text, these are listed as holding under ``suitable regularity conditions''.
We begin by listing the assumptions on the covariates and model parameters that constitute these conditions.
The first of our assumptions ensures that the spectrum of $\Apop = \Xpop \Xpop^\top$ is suitably well-behaved.

\begin{assumption} \label{assum:Apop:spectrum}
  The expected adjacency matrix $\Apop$ satisfies $d = \rank \Apop = \Op{1}$ and
  \begin{equation} \label{eq:assum:spop:growth}
    \spop_d = \Omegap{ 1 } .
  \end{equation}
  Further, the edge-level variance does not grow too quickly compared to the signal in $\spop_d$:
  \begin{equation} \label{eq:assum:spectralConc}
    \spop_d = \omegap{ \sqrt{ \nu + b^2 } \sqrt{n} \log n } .
  \end{equation}
  Letting $\kappa = \spop_1/\spop_d$ be the condition number of $\Apop$, we require
  \begin{equation} \label{eq:assum:XhatX:rate:2}
    \frac{ \kappa (\nu+b^2) n \log^2 n }{ \spop_d^{3/2} }
    = \op{ 1 } .
  \end{equation}
\end{assumption}

We also require assumptions on the behavior of the minimum degree of the weighted adjacency matrix $\Apop$.
Denoting the degree of vertex $i$ in $\Apop$ by
\begin{equation} \label{eq:def:dtilde}
  \dtilde_i = \bbE\left[ d_i \mid \Xpop \right]
  = \left[ \Apop \onevec \right]_i = \sum_{j=1}^n \Xpop_i^\top \Xpop_j ,
\end{equation}
we require the following assumption.

\begin{assumption} \label{assum:degrees}
  The degrees of the expected adjacency matrix $\Apop$ are such that
  \begin{align}
    \min_{i \in [n]} \dtilde_i^2 %
     & = \omega\left( \kappa (\nu+b^2) n^{3/2} \log^2 n \right), \label{eq:deglb:2} \\
    \label{eq:deglb:2:kappa}
    \spop_d \min_{i \in [n]} \dtilde_i
     & = \omega\left( \kappa^{3/2} (\nu+b^2) n^{3/2} \log^2 n \right) ~\text{ and } \\
    \label{eq:assum:Grate:stronger}
    \frac{ \spop_1 \sqrt{ \nu + b^2 } \sqrt{n} \log n }
    { \min_{i \in [n]} \dtilde_i^2 }
     & = \op{ \frac{1}{n^{1/4} } } .
  \end{align}
\end{assumption}

We also require regularity assumptions on the latent positions and other covariates.

\begin{assumption} \label{assum:latentpositions}
  The latent positions $\Xpop \in \R^{n \times d}$ are such that
  \begin{equation} \label{eq:assum:Xtti}
    \left\| \Xpop \right\|_{\tti}
    = \Op{ \kappa \sqrt{\nu + b^2} \log n } .
  \end{equation}
\end{assumption}

Our results for the latent contagion estimator $\thetahattsls$ require a slightly stronger assumption.
We note that these bounds are focused on the setting where the variance of the edges (as encoded by the subgamma parameters $\nu,b$) are decoupled from the growth rate of $\spop_d$.
As such, we expect that these assumptions can be relaxed by extending the bounds in \cite{rubin-delanchy2022} to the case of subgamma edge distributions, but we leave this to future work.
\begin{assumption} \label{assum:stronger4dhat}
  The tail behavior of the edge-level noise is such that
  \begin{equation} \label{eq:assum:degreekappa}
    \frac{ \kappa \sqrt{\nu+b^2} n \log n }{ \spop_d } = \Op{ 1 }.
  \end{equation}
\end{assumption}

\begin{assumption}[Peer oracle assumptions] \label{assum:peer-oracle}
  Under the peer contagion model of Equation~\eqref{eq:lim-peer}, the following conditions hold.
  \begin{enumerate}
    \item All the diagonal elements of $\G$ are zero and $\G$ has uniformly bounded row and column sums, almost surely. %
    \item $\abs*{\betay} < 1$, so that $\paren*{\mI - \betay \G}^{-1} $ is non-singular.
    \item The regressor matrices $\W$ and $\Xpop$ have full column rank (for $n$ large enough), and the elements of $\W$ and $\Xpop$ are uniformly bounded in absolute value almost surely. \label{item:assum:oracle:boundedCovariates}
    \item $\be$ are i.i.d. with zero mean, variance $\bbE \brac*{\be_i^2} = \sigmaeps^2 < \infty$ and all entries of $\be$ have finite fourth moments.  \label{item:assum:oracle:epsilonMoments}
    \item The instrument matrices $\HpeerOracle$ have full column rank almost surely and are composed of a subset of linearly independent columns of
          \begin{equation*}
            \begin{bmatrix}
              \W \; \Xpop \; \G \W \; \G \Xpop \;
              \G^2 \W \; \G^2 \Xpop
            \end{bmatrix},
          \end{equation*}
          where the subset contains $\W$ and $\Xpop$.
    \item The limit
          \begin{equation*}
            \lim_{n \to \infty} \frac{1}{n} \HpeerOracle^\top \HpeerOracle
          \end{equation*}
          is almost surely finite and non-singular. Further,
          \begin{equation*}
            \frac{1}{n} \HpeerOracle^\top \ZpeerOracle
          \end{equation*}
          converges almost surely to a matrix that is finite with full column rank.
  \end{enumerate}

  Assumption~\ref{assum:peer-oracle} is nearly identical to the assumptions for the non-stochastic $\Xpop$ case. We adapt results for non-stochastic $\Xpop$ to the stochastic $\Xpop$ case by conditioning on $\Xpop$ and requiring that the conditions for the non-stochastic $\Xpop$ estimator hold almost surely. In random dot product graphs, where $\Xpop$ is often uniformly bounded, these assumptions are very natural. In settings where $\Xpop$ is not almost surely uniformly bounded, some moment conditions on $\Xpop$ are instead necessary to ensure convergence of the oracle estimator \citep{gupta2019}.
\end{assumption}

Our results for the latent contagion model require similar assumptions.
\begin{assumption}[Latent oracle assumptions]
  \label{assum:latent-oracle}

  Under the model in Equation~\eqref{eq:lim-latent}, the conditions of Assumption~\ref{assum:peer-oracle} hold, with $\Gtilde$ in place of $\G$ and $\Htilde$ in place of $\H$.
\end{assumption}

Under either Assumption~\ref{assum:peer-oracle} or~\ref{assum:latent-oracle}, the matrix of node-level covariates, the design matrix and the instrument matrix are all well-conditioned.
Further, it is straightforward under either of these assumptions, the following growth rates hold.
These will prove convenient to have for reference in our proofs to follow.
\begin{equation} \label{eq:assum:XandW:spectrumUB}
  \max\left\{ \left\| \Xpop \right\|, \left\| \W \right\| \right\}
  = \Op{ \sqrt{n} },
\end{equation}
\begin{equation} \label{eq:assum:Zgrowth}
  \sigmamin( \MlatOracle \ZlatOracle ) = \Omega( \sqrt{n} )
  ~\text{ and }~
  \| \ZlatOracle \| = \Op{ \sqrt{n} } , ~\text{ and }
\end{equation}
\begin{equation} \label{eq:assum:Hgrowth}
  \sigmamin( \HlatOracle ) = \Omega( \sqrt{n} )
  ~\text{ and }~
  \| \HlatOracle \| = \Op{ \sqrt{n} } .
\end{equation}

Our projection equivalence result, Theorem~\ref{thm:proj-equivalence}, does not require the subgamma edge behavior of Definition~\ref{def:subgamma-network}.
Instead, we need only that the edge noise has bounded second moments, along with degree growth conditions that hold in expectation, rather than in probability.

\begin{assumption} \label{assum:equiv:limcov}
  The matrix $\lim_{n \rightarrow \infty} \bbE \ZlatOracle^\top \ZlatOracle/n$
  exists and is invertible.
\end{assumption}

\begin{assumption} \label{assum:equiv:edges}
  The entries of $\mA-\mApop$ are, conditionally on $\Xpop$, mean zero and independent (up to symmetry), and obey
  \begin{equation} \label{eq:assum:edgeVariance}
    \max_{i,j} \bbE \left[ (\mA-\mApop)_{ij}^2 \mid \Xpop \right]
    \le \nu_n .
  \end{equation}
\end{assumption}

\begin{assumption} \label{assum:equiv:LandX}
  The expected degrees $\dtilde_1,\dtilde_2,\dots,\dtilde_n$ are such that
  \begin{equation} \label{eq:assum:Edegrecip}
    \sigmaeps^2 \nu_n
    \sum_{i=1}^n \bbE \frac{ 1 }{ \dtilde_i^2 }
    = o( n^{-1/2} ) ,
  \end{equation}
  where $\nu_n$ is the variance parameter is Equation~\eqref{eq:assum:edgeVariance} above.
\end{assumption}

\section{Technical Results} \label{apx:technical}

Here we collect basic results, largely related to concentration inequalities, which we will use to establish our technical results in the sequel.

\begin{definition} \label{def:subgamma} %
  Let $Z$ be a mean-zero random variable with cumulant generating function $\psi_Z(t) = \log \E{e^{t Z}}$.

  \begin{enumerate}
    \item $Z$ is sub-Gaussian($\nu$) for $\nu > 0$ if $\psi_Z(t) \le t^2 \nu / 2$ for all $t \in \R$.

    \item $Z$ is sub-gamma($\nu, b$) for $\nu, b \ge 0$ if $\psi_Z(t) \le \frac{t^2 \nu}{2 (1 - b t)}$ and $\psi_{-Z}(t) \le \frac{t^2 \nu}{2 (1 - b t)}$ for all $t < 1 / b$.
  \end{enumerate}
\end{definition}

\begin{lemma}[\cite{boucheron2013} Chapter 2] \label{lem:sgbasic}
  Suppose that $Z$ is a $(\nu,b)$-subgamma random variable.
  Then for all $ t > 0$,
  \begin{equation*}
    \Pr\left[ |X| > \sqrt{2\nu t} + bt \right] \le \exp\{ -t \}
  \end{equation*}
\end{lemma}

The following is a basic result concerning subgamma random variables, which we prove for the sake of completeness.

\begin{lemma} \label{lem:sgsum}
  Let $Z_1,Z_2,\dots,Z_n$ be a collection of independent $(\nu,b)$-subgamma random variables and let $\alpha_1,\alpha_2,\dots,\alpha_n \in \R$ be nonnegative.
  Then, defining $S_n = \sum_i \alpha_i Z_i$, for any $t > 0$ and any constant $c > 0$, for suitably-chosen constant $C_0$, it holds with probability at least $1 - C_0n^{-c}$ that
  \begin{equation} \label{eq:firstclaim}
    \left| S_n \right| \le
    C \sqrt{ \nu + b^2}\left( \sum_{i=1}^n \alpha_i^2 \right)^{1/2} \log n
  \end{equation}
  and
  \begin{equation} \label{eq:Sn:inprob}
    \left| S_n \right|
    = \Op{ \sqrt{\nu + b^2} \sqrt{ \sum_{i=1}^n \alpha_i^2 } } .
  \end{equation}
\end{lemma}
\begin{proof}
  By a basic property of subgamma random variables \citep[see][Chapter 2]{boucheron2013}, $\alpha_i Z_i$ is $(\nu_i, b_i)$-subgamma, where $\nu_i = \alpha_i^2 \nu$ and $b_i = \alpha_i b_i$, and the random sum $S_n = \sum_i \alpha_i Z_i$ is a subgamma random variable with parameters
  \begin{equation*} \begin{aligned}
      \nubar & = \sum_{i=1}^n \nu_i = \nu \sum_{i=1}^n \alpha_i^2 \\
      \bbar  & = \max_{i \in [n]} b_i
      = b \max_{i \in [n]} \alpha_i
      \le b \sqrt{ \sum_{i=1}^n \alpha_i^2 }.
    \end{aligned} \end{equation*}
  Thus, applying Lemma~\ref{lem:sgbasic}, for any $t > 0$,
  \begin{equation*}
    \Pr\left[ |S_n| > \sqrt{2\nubar t} + \bbar t \right] \le \exp\{ -t \}.
  \end{equation*}
  Setting $t$ to be any constant yields Equation~\eqref{eq:Sn:inprob}.
  Setting $t= C\log n$ for suitably large $C > 0$ and noting that $\bbar \le C \nubar^{1/2}$ for suitably-chosen constant $C > 0$, it follows that
  \begin{equation*}
    \Pr\left[ |S_n| > C( \nubar^{1/2} \log^{1/2} n + \bbar \log n \right]
    \le 2n^{-c}.
  \end{equation*}
  Observing that
  \begin{equation*}
    \nubar^{1/2} \log^{1/2} n + \bbar \log n
    \le \sqrt{ \nu + b^2 }\left( \sum_{i=1}^n \alpha_i^2 \right)^{1/2} \log n
  \end{equation*}
  establishes Equation~\eqref{eq:firstclaim}.
\end{proof}

\begin{lemma} \label{lem:sgmax}
  Let $Z_1,Z_2,\dots,Z_n$ be a collection of independent $(\nu,b)$-subgamma random variables and let $c > 0$ be a constant.
  Then it holds with probability at least $1-Cn^{-c}$ that
  \begin{equation*}
    \max_{i \in [n]} |Z_i| \le C\sqrt{ \nu + b^2 } \log n,
  \end{equation*}
\end{lemma}
\begin{proof}
  For $t \ge 0$, applying a union bound followed by Lemma~\ref{lem:sgbasic},
  \begin{equation*}
    \Pr\left[ \max_i |Z_i| > \sqrt{2\nu t} + bt \right]
    \le \sum_{i=1}^n \Pr\left[ |Z_i| > \sqrt{2\nu t} + bt \right]
    \le n \exp\{ -t \}.
  \end{equation*}
  Taking $t = C \log n$ for $C>0$ chosen suitably large, it holds that with probability at least $1-n^{-c}$,
  \begin{equation*}
    \max_i |Z_i|
    \le \sqrt{2C \nu \log n} + Cb\log n
    \le C\sqrt{ \nu + b^2 } \log n,
  \end{equation*}
  as we set out to show.
\end{proof}

A similar result to the one below appeared in \cite{hayes2025}.
We restate it here with slightly adapted notation for the sake of completeness.

\begin{lemma} \label{lem:sglinear}
  Suppose that $\be \in \R^n$ is a vector of independent mean-zero random variables with 
  \begin{equation*}
    \max_{i \in [n]} \E{\varepsilon_i^2} \le B
  \end{equation*}
  for some $B > 0$ not depending on $n$.
  Let $H \in \R^{n \times n}$ be a (possibly random) matrix with $\be$ independent of $H$. Then
  \begin{equation*}
    \norm*{H \be} = \Op{ \sqrt{B \trace H^T H } }.
  \end{equation*}
  In particular, taking $H = I$, $\norm*{\be} = \Op{\sqrt{B n}}$.
\end{lemma}
\begin{proof}
  We observe that
  \begin{equation*}
    \E{\norm*{H \be}^2}
    = \E{\be^T H^T H \be}
    \le B \trace H^T H.
  \end{equation*}
  Let $\delta > 0$ be a constant.
  Applying Markov's inequality, for any $t > 0$,
  \begin{equation*}
    \Pr\left[ \frac{ \norm*{ H \be }^2 }{ t } > \delta \right]
    \le \frac{\E{\norm*{ H \be }^2}}{ t \delta }
    \le \frac{B \trace H^T H }{ t \delta }.
  \end{equation*}
  Let $r_n$ be any function of $n$ growing such that $r_n = \omega( B \trace H^TH )$.
  Then taking $t = r_n$,
  \begin{equation*}
    \lim_{n \rightarrow \infty}
    \Pr\left[ \frac{ \norm*{ H \be }^2 }{ r_n } > \delta \right]
    = 0.
  \end{equation*}
  Thus, $\norm*{ H \be }^2 = \op{ r_n }$ for any $r_n = \omega( B \trace H^T H )$, and it follows that
  \begin{equation*}
    \norm*{ H \be }^2
    = \Op{ B \trace H^T H }.
  \end{equation*}
  Taking square roots completes the proof.
\end{proof}

\begin{lemma} \label{lem:epstilde:ctl}
  With notation as above, for $\beta \in (-1,1)$ fixed and $\be$ independent of $\Gtilde$, define
  \begin{equation*}
    \epstilde = \paren*{I - \beta \Gtilde}^{-1} \be
    \in \R^n.
  \end{equation*}
  Then, if the entries of $\be$ have bounded second moments as in Lemma~\ref{lem:sgmax},
  \begin{equation*}
    \| \epstilde \|
    =
    \Op{ \sqrt{ \frac{ Bn }{ (1-\beta)^2 } } } .
  \end{equation*}
\end{lemma}
\begin{proof}
  Observe that by Lemma~\ref{lem:IbgyG:invertible},
  \begin{equation*}
    \trace \paren*{\paren*{I - \beta \Gtilde}^{-1}}^\top \paren*{I - \beta \Gtilde}^{-1}
    \le
    \frac{ n }{ (1-\beta)^2 }.
  \end{equation*}
  Applying Lemma~\ref{lem:sglinear} with $H = \paren*{I - \beta \Gtilde}^{-1}$, it follows that
  \begin{equation*}
    \| \epstilde \|
    =
    \Op{ \sqrt{ \frac{ Bn }{ (1-\beta)^2 } } },
  \end{equation*}
  as we set out to show.
\end{proof}

Our final result in this section concerns concentration of the degrees $d_1,d_2,\dots,d_n$ of the observed network $\A$ about their conditional expectations, defined in Equation~\eqref{eq:def:dtilde}.

\begin{lemma} \label{lem:degconc}
  Suppose that $\A$ follows a sub-gamma model as in Definition~\ref{def:subgamma-network}.  Then
  \begin{equation*}
    \max_{i \in [n]} \left| d_i - \dtilde_i \right|
    = \Op{ \sqrt{\nu + b^2} \sqrt{n} \log n } .
  \end{equation*}
\end{lemma}
\begin{proof}
  Fix $i \in [n]$. We observe that
  \begin{equation*}
    d_i - \dtilde_i
    = \sum_{j \in [n]\setminus\{i\} } (A_{ij} - \rho_n X_i^T X_j)
  \end{equation*}
  is, conditional on $\Xpop$, a sum of $(\nu,b)$-subgamma random variables.
  Applying Lemma~\ref{lem:sgsum} with suitably chosen constants, it holds with probability at least $1-2n^{-3}$ that
  \begin{equation}
    \left| d_i - \dtilde_i \right|
    \le C \sqrt{\nu + b^2} \sqrt{n} \log n.
  \end{equation}
  A union bound over $i \in [n]$ completes the proof.
\end{proof}

\begin{lemma} \label{lem:degrecip:conc}
  Suppose that $\A$ follows a sub-gamma model as in Definition~\ref{def:subgamma-network} and that the model parameters grow in such a way that Equation~\eqref{eq:deglb:2} holds.
  Then with probability at least $1-O(n^{-2})$, it holds uniformly over all $i \in [n]$ that
  \begin{equation*}
    \left| \frac{1}{d_i} - \frac{1}{\dtilde_i} \right|
    \le
    \frac{ C \sqrt{\nu + b^2} \sqrt{n} \log n }{ \dtilde_i^2 }.
  \end{equation*}
\end{lemma}
\begin{proof}
  Defining $\gamma_n = C \sqrt{\nu + b^2} \sqrt{n} \log n$, using the fact that $a^{-1} - b^{-1} = b^{-1} (a - b) a^{-1}$ and applying Lemma~\ref{lem:degconc} twice, it holds with high probability that
  \begin{equation*}
    \left| \frac{1}{d_i} - \frac{1}{\dtilde_i} \right|
    \le \frac{ C \gamma_n  }{ \dtilde_i (\dtilde_i - \gamma_n) }
    \le \frac{ C \sqrt{\nu + b^2} \sqrt{n} \log n }{ \dtilde_i^2 },
  \end{equation*}
  where the last inequality follows from our growth assumption in Equation~\eqref{eq:deglb:2}.
\end{proof}

\section{Estimating the Latent Positions} \label{apx:XhatX}

Here we collect results relating the latent position estimates $\Xhat$ to the true latent positions $\Xpop$.
Many of the results in this section can be found elsewhere in the spectral methods literature \citep[see, for example,][]{lyzinski2017,levin2019,levin2022a,hayes2025}.
We include these results, with notation adjusted to the current setting, for the sake of convenience.

\begin{lemma} \label{lem:E:spectral}
	Suppose that $\A$ follows a sub-gamma model as in Definition~\ref{def:subgamma-network}.
	Then with high probability,
	\begin{equation*}
		\left\| \A - \Apop \right\|
		\le C \sqrt{ \nu + b^2 } \sqrt{n} \log n .
	\end{equation*}
\end{lemma}
\begin{proof}
	This result appears as Lemma 5 in \cite{levin2022a}, setting $N=1$ in the notation of that work.
\end{proof}

\begin{lemma} \label{lem:invsqrteigvals}
	Suppose that $\A$ follows a sub-gamma model as in Definition~\ref{def:subgamma-network} and that Assumption~\ref{assum:Apop:spectrum} holds.
	There exists a constant $C>0$ such that with high probability,
	\begin{equation*}
		\norm*{\Shat^{-1/2}} \le C \spop_d^{-1/2} \quad \text{and} \quad
		\norm*{\Shat^{1/2}}  \le C \spop_1^{1/2} ,
	\end{equation*}
	where $\Shat$ is as in Definition~\ref{def:ASE}.
\end{lemma}
\begin{proof}
	Both of these facts are shown in the course of proving Lemma 4 of \cite{levin2022a}.
	In particular, see Equations (28) and (32) in that work.
\end{proof}

The following two results are standard in the spectral embeddings literature \citep[see, e.g.,][]{lyzinski2017} once we include our assumption that $d$ is order a constant.
For the first, see Lemma 27 in \cite{hayes2025} or Proposisiton 19 in \cite{levin2022a}.
For the second, see Lemma 38 in \cite{hayes2025} or Proposition 20 in \cite{levin2022a}.

\begin{lemma} \label{lem:H-Q}
	Suppose that $\A$ follows a sub-gamma model as in Definition~\ref{def:subgamma-network} and that Assumption~\ref{assum:Apop:spectrum} holds.
	Then there exists a sequence of orthogonal matrices $\Q \in \R^{d \times d}$ such that
	\begin{equation*}
		\left\| \Upop^\top \Uhat - \Q \right\|_F
		\le \frac{ C (\nu+b^2) n \log^2 n }{ \spop_d^2 } .
	\end{equation*}
\end{lemma}

\begin{lemma} \label{lem:Qswap}
	Under the same assumptions as Lemma~\ref{lem:H-Q},
	\begin{equation*}
		\left\| \Uhat - \Upop \Upop^\top \Uhat \right\|_F
		\le \frac{ C \sqrt{\nu + b^2} \sqrt{n} \log n }{ \spop_d } .
	\end{equation*}
	Furthermore, with $\Q$ the matrix guaranteed by Lemma~\ref{lem:H-Q},
	\begin{equation*}
		\left\| \Q \Shat - \Spop \Q \right\|_F
		\le \frac{ C \spop_1 \left( \nu + b^2 \right) n \log^2 n }
		{ \spop_d^2 }
		+ C \sqrt{\nu + b^2} \log n ,
	\end{equation*}
	\begin{equation*}
		\left\| \Q \Shat^{1/2} - \Spop^{1/2} \Q \right\|_F
		\le \frac{ C \spop_1 \left( \nu + b^2 \right) n \log^2 n }{ \spop_d^{5/2} }
		+ \frac{ C \sqrt{\nu + b^2} \log n }{ \spop_d^{1/2} }
	\end{equation*}
	and
	\begin{equation*}
		\left\| \Q \Shat^{-1/2} - \Spop^{-1/2} \Q \right\|_F
		\le \frac{ C \spop_1 \left( \nu + b^2 \right) n \log^2 n }{ \spop_d^{7/2} }
		+ \frac{ C \sqrt{\nu + b^2} \log n }{ \spop_d^{3/2} } .
	\end{equation*}
\end{lemma}

The following result, which generalizes Lemma 40 in \cite{hayes2025}, is central to proving our main results.

\begin{lemma} \label{lem:XhatXTB:ctrl}
	Suppose that $\A$ follows a sub-gamma model as in Definition~\ref{def:subgamma-network} and that Assumption~\ref{assum:Apop:spectrum} holds.
	Let $\mB \in \R^{n \times r}$ be a matrix with $\A-\Apop$ independent of $\mB$ conditional on $\Xpop$.
	Then there exists $\Q \in \R^{d \times d}$ such that
	\begin{equation*}
		\left\| \left( \Xhat \Q^\top - \Xpop \right)^\top \mB \right\|
		\le
		\frac{ C \sqrt{r} \sqrt{ \nu + b^2 } \| \mB \| \log n }{ \sqrt{\spop_d} }
		+
		\frac{ C \kappa (\nu+b^2) \| \mB \| n \log^2 n }{ \spop_d^{3/2} } .
	\end{equation*}
\end{lemma}
\begin{proof}
	Take $\Q \in \R^{d \times d}$ to be the orthogonal matrix guaranteed by Lemma~\ref{lem:H-Q}.
	Applying a standard decomposition for the adjacency spectral embedding \citep[see, for example][]{lyzinski2017,levin2019,hayes2025}, writing $\mE = \A - \Apop$ for ease of notation,
	\begin{equation} \label{eq:XhatXBT:maindecomp} \begin{aligned}
			\left( \Xhat \Q^\top - \Xpop \right)^{\top}  \mB
			 & = \Q \left( \Uhat \Shat^{1/2} - \Upop \Spop^{1/2} \Q \right)^\top \mB \\
			 & = \Q \Spop^{-1/2} \Upop^\top \mE  \mB                                 %
			+ \Q \left( \Q \Shat^{-1/2} - \Spop^{-1/2} \Q \right)^\top \Upop^\top
			\mE \mB                                                                  \\ %
			 & ~~~+ \Q \Shat^{-1/2} \Q^\top \Upop^\top \mE \Upop\Upop^\top \mB       %
			+ \Q\Shat^{1/2}\left(\Upop\Upop^\top\Uhat - \Upop\Q\right)^\top \mB      \\ %
			 & ~~~+ \Q \left( \Q \Shat^{1/2} - \Spop^{1/2} \Q \right)^{\top }  \!
			\Upop^\top \mB %
			+ \Q \Shat^{-1/2} \! \left(\Uhat -\Upop \Q\right)^{\top} \!
			\mE \left(\mI - \Upop\Upop^\top\right)\!\mB . %
		\end{aligned} \end{equation}
	We will bound each of the six right-hand terms, after which the triangle inequality will yield our result.

	For the first right-hand term in Equation~\eqref{eq:XhatXBT:maindecomp}, using submultiplicativity and the fact that $\Q$ is orthogonal,
	\begin{equation*}
		\left\| \Q \Spop^{-1/2} \Upop^\top \mE \mB \right\|
		\le \frac{ \left\| \Upop^\top \mE \mB \right\| }{ \sqrt{\spop_d} } .
	\end{equation*}
	Taking the singular value decomposition of $\mB$, then using submultiplicativity and applying Bernstein's inequality \citep{boucheron2013,vershynin2020}, recalling that $\mE$ is independent of $\mB$ conditional on $\Xpop$,
	\begin{equation} \label{eq:XhatXBT:term1:done}
		\left\| \Q \Spop^{-1/2} \Upop^\top \mE \mB \right\|
		\le \frac{ C \sqrt{ \nu + b^2 } \sqrt{r} \| \mB \| \log n }{ \sqrt{\spop_d} } .
	\end{equation}

	For the second term in Equation~\eqref{eq:XhatXBT:maindecomp}, submultiplicativity followed by Bernstein's inequality and Lemma~\ref{lem:Qswap} yield
	\begin{equation*} \begin{aligned}
			\left\| \Q \left( \Q \Shat^{-1/2} - \Spop^{-1/2} \Q \right)^\top \Upop^\top
			\mE \mB \right\|
			 & \le \left\| \Q \Shat^{-1/2} - \Spop^{-1/2} \Q \right\|
			\left\| \Upop^\top \mE \mB \right\|                            \\
			 & \le C  \left( \frac{ \spop_1 \sqrt{ \nu + b^2 }~ n \log n }
			{ \spop_d^2 }
			+ d  \right)
			\frac{ ( \nu + b^2 ) \sqrt{r} \| \mB \| \log^2 n }{ \spop_d^{3/2} } .
		\end{aligned} \end{equation*}
	Using the growth assumptions in Equations~\eqref{eq:assum:XandW:spectrumUB},~\eqref{eq:assum:spectralConc},~\eqref{eq:assum:XhatX:rate:2} and~\eqref{eq:assum:spop:growth},
	\begin{equation} \label{eq:XhatXBT:term2:done}
		\left\| \Q \left( \Q \Shat^{-1/2} - \Spop^{-1/2} \Q \right)^\top \Upop^\top
		\mE \mB \right\|
		\le \frac{ C \sqrt{ \nu + b^2 } \sqrt{r} \| \mB \| \log n }{ \sqrt{\spop_d} } .
	\end{equation}

	Considering the third right-hand term in Equation~\eqref{eq:XhatXBT:maindecomp}, Bernstein's inequality and Lemma~\ref{lem:invsqrteigvals} yield
	\begin{equation} \label{eq:XhatXBT:term3:done}
		\left\| \Q \Shat^{-1/2} \Q^\top \Upop^\top \mE \Upop\Upop^\top \mB \right\|
		\le \frac{ C \sqrt{ \nu + b^2 } \| \mB \| \log n }{ \sqrt{\spop_d} } .
	\end{equation}

	For the fourth right-hand term in Equation~\eqref{eq:XhatXBT:maindecomp}, Lemmas~\ref{lem:invsqrteigvals} and~\ref{lem:H-Q} yield
	\begin{equation} \label{eq:XhatXBT:term4:done}
		\left\| \Q\Shat^{1/2}\left(\Upop\Upop^\top\Uhat - \Upop\Q\right)^\top \mB
		\right\|
		\le \left\| \Shat^{1/2} \right\| \left\| \Upop^\top\Uhat - \Q \right\|
		\left\| \mB \right\|
		\le \frac{ C (\nu +b^2) \| \mB \| n \log^2 n }{ \spop_d^{3/2} } .
	\end{equation}

	Similarly, for the fifth term in Equation~\eqref{eq:XhatXBT:maindecomp}, submultiplicativity and Lemma~\ref{lem:Qswap} yield
	\begin{equation} \label{eq:XhatXBT:term5:done} \begin{aligned}
			\left\| \Q \left( \Q \Shat^{1/2} - \Spop^{1/2} \Q \right)^{\top }
			\Upop^\top \mB \right\|
			 & \le \left\| \Q \Shat^{1/2} - \Spop^{1/2} \Q \right\| \left\| \mB \right\|     \\
			 & \le \frac{ C \kappa ( \nu+b^2 ) \|\mB\|  n \log^2 n }{ \spop_d^{3/2} }
			+ \frac{ C \sqrt{\nu+b^2} \| \mB \| \log n }{ \sqrt{\spop_d} } .
		\end{aligned} \end{equation}

	Finally, to control the sixth right-hand term in Equation~\eqref{eq:XhatXBT:maindecomp}, adding and subtracting appropriate quantities yields
	\begin{equation} \label{eq:XhatXBT:term6:split} \begin{aligned}
			\Q \Shat^{-1/2} \! \left(\Uhat -\Upop \Q\right)^{\top}
			\mE \left(\mI - \Upop\Upop^\top\right)\!\mB
			 & = \Q \Shat^{-1/2} \left( \Uhat - \Upop\Upop^\top \Uhat \right)^\top
			\mE \left(\mI - \Upop\Upop^\top\right)\!\mB                                       \\
			 & ~~~~~~+ \Q \Shat^{-1/2} \left( \Upop\Upop^\top \Uhat - \Upop \Q\right)^{\top}
			\mE \left(\mI - \Upop\Upop^\top\right)\!\mB .
		\end{aligned} \end{equation}
	By submultiplicativity followed by Lemmas~\ref{lem:E:spectral} and~\ref{lem:Qswap},
	\begin{equation} \label{eq:XhatXBT:term6a} \begin{aligned}
			\left\| \Q \Shat^{-1/2} \left( \Uhat - \Upop\Upop^\top \Uhat \right)^\top
			\mE \left(\mI - \Upop\Upop^\top\right)\!\mB \right\|
			 & \le \| \Shat^{-1/2} \| \left\| \Uhat - \Upop\Upop^\top \Uhat \right\|
			\| \mE \| \| \mB \|                                                      \\
			 & \le
			\frac{ C \left\| \mB \right\| (\nu + b^2) n \log^2 n }{ \spop_d^{3/2} } .
		\end{aligned} \end{equation}
	Similarly, submultiplicativity followed by Lemmas~\ref{lem:E:spectral},~\ref{lem:invsqrteigvals} and~\ref{lem:H-Q},
	\begin{equation} \label{eq:XhatXBT:term6b} \begin{aligned}
			\left\| \Q \Shat^{-1/2} \left( \Upop\Upop^\top \Uhat - \Upop \Q\right)^{\top}
			\mE \left(\mI - \Upop\Upop^\top\right)\!\mB \right\|
			 & \le \frac{ C \left\| \mB \right\| (\nu +b^2)^{3/2} n^{3/2} \log^3 n }
			{ \spop_d^{5/2} }                                                        \\
			 & \le \frac{ C (\nu + b^2) \| \mB \| n \log^2 n }{ \spop_d^{3/2} }  .
		\end{aligned} \end{equation}
	where the second inequality follows from our growth assumption in Equation~\eqref{eq:assum:spectralConc}.
	Applying the triangle inequality to Equation~\eqref{eq:XhatXBT:term6:split} and using Equations~\eqref{eq:XhatXBT:term6a} and~\eqref{eq:XhatXBT:term6b},
	\begin{equation} \label{eq:XhatXBT:term6:done}
		\left\| \Q \Shat^{-1/2} \! \left(\Uhat -\Upop \Q\right)^{\top}
		\mE \left(\mI - \Upop\Upop^\top\right)\!\mB \right\|
		\le
		\frac{ C (\nu + b^2) \| \mB \| n \log^2 n }{ \spop_d^{3/2} }  .
	\end{equation}

	Applying the triangle inequality to Equation~\eqref{eq:XhatXBT:maindecomp}, followed by
	Equations~\eqref{eq:XhatXBT:term1:done},
	~\eqref{eq:XhatXBT:term2:done},
	~\eqref{eq:XhatXBT:term3:done},
	~\eqref{eq:XhatXBT:term4:done},
	~\eqref{eq:XhatXBT:term5:done},
	and~\eqref{eq:XhatXBT:term6:done}, completes the proof.
\end{proof}

Our final result in this section concerns the behavior of
\begin{equation} \label{eq:def:dhat}
	\dhat_i = \sum_{j=1}^n \Xhat_i^\top \Xhat_j ,
\end{equation}
which estimates the ``latent'' degrees $\dtilde_i$ given in Equation~\eqref{eq:def:dtilde}.

\begin{lemma} \label{lem:dhatdtilde}
	Suppose that $\A$ follows a sub-gamma model as in Definition~\ref{def:subgamma-network} and that Assumptions~\ref{assum:Apop:spectrum} through~\ref{assum:stronger4dhat} hold.
	Then, with high probability, it holds uniformly over all $i \in [n]$ that
	\begin{equation} \label{eq:dhatdtilde:close}
		\left| \dhat_i - \dtilde_i \right|
		\le
		C \sqrt{\kappa} \sqrt{\nu+b^2} \sqrt{n} \log n .
	\end{equation}
	Further,
	\begin{equation} \label{eq:dhatdtilde:recip:close}
		\max_{i \in [n]} \left| \frac{ 1 }{ \dhat_i } - \frac{ 1 }{ \dtilde_i } \right|
		\le
		\frac{ C \sqrt{\kappa} \sqrt{\nu + b^2} \sqrt{n} \log n }{ \min_i \delta_i^2 }.
	\end{equation}
\end{lemma}
\begin{proof}
	Fix $i \in [n]$.
	Recalling the definitions of $\dhat_i$ and $\dtilde_i$ from Equations~\eqref{eq:def:dhat} and~\eqref{eq:def:dtilde}, respectively, writing $\ve_i \in \R^n$ for the $i$-th standard basis vector,
	\begin{equation} \label{eq:dhatdtilde:triangle} \begin{aligned}
			\left| \dhat_i - \dtilde_i \right|
			 & = \left| \ve_i^\top \left( \Apophat - \Apop \right) \onevec \right| \\
			 & \le
			\left| \ve_i^\top \left( \Xhat \Q^\top \!\!-\! \Xpop \!\right) \Xpop^\top \onevec \right|
			+
			\left| \ve_i^\top \Xpop  \left( \Xhat \Q^\top \!\!-\! \Xpop \right)^\top \! \onevec \right|
			+
			\left| \ve_i^\top \left( \Xhat \Q^\top \!\!-\! \Xpop \right)
			\left( \Xhat \Q^\top \!\!-\! \Xpop \right)^\top \! \onevec \right| .
		\end{aligned} \end{equation}

	By Cauchy-Schwarz and Lemma~\ref{lem:XhatXTB:ctrl},
	\begin{equation*} \begin{aligned}
			\left| \ve_i^\top \left( \Xhat \Q^\top - \Xpop \right) \Xpop^\top \onevec \right|
			 & \le
			C \sqrt{n} \| \Xpop \|
			\left( \frac{ \sqrt{ \nu + b^2 } \log n }{ \sqrt{\spop_d} }
			+
			\frac{ \kappa (\nu+b^2) n \log^2 n }{ \spop_d^{3/2} }
			\right) \\
			 & \le
			C \sqrt{n}
			\left( \sqrt{\kappa} \sqrt{ \nu + b^2 } \log n
			+ \frac{ \kappa^{3/2} (\nu+b^2) n \log^2 n }{ \spop_d }
			\right).
		\end{aligned} \end{equation*}

	Applying our growth bound in Equation~\eqref{eq:assum:degreekappa},
	\begin{equation*}
		\left| \ve_i^\top \left( \Xhat \Q^\top - \Xpop \right) \Xpop^\top \onevec \right|
		\le
		C \sqrt{\kappa} \sqrt{ \nu+b^2 } \sqrt{n} \log n.
	\end{equation*}
	Similarly, trivially upper bounding $\| \ve_i^\top \Xpop \| \le \| \Xpop \|$,
	\begin{equation*}
		\left| \ve_i^\top \Xpop  \left( \Xhat \Q^\top - \Xpop \right)^\top \onevec \right|
		\le
		C \sqrt{\kappa} \sqrt{ \nu+b^2 } \sqrt{n} \log n,
	\end{equation*}
	Applying Lemma~\ref{lem:XhatXTB:ctrl} twice more, once with $\mB = \ve_i$ and once with $\mB = \onevec$,
	\begin{equation*} \begin{aligned}
			\left| \ve_i^\top \left( \Xhat \Q^\top - \Xpop \right)
			\left( \Xhat \Q^\top - \Xpop \right)^\top \onevec \right|
			 & \le
			C \sqrt{n}
			\left( \frac{ ( \nu + b^2 ) \log^2 n }{ \spop_d }
			+
			\frac{ \kappa^2 (\nu+b^2)^2 n^2 \log^4 n }
				{ \spop_d^{3} }
			\right)                                                           \\
			 & \le C \sqrt{ \kappa } \sqrt{ \nu+b^2 } \sqrt{n} \log n,
		\end{aligned} \end{equation*}
	where we have used our growth assumptions in Equations~\eqref{eq:assum:spop:growth},~\eqref{eq:assum:spectralConc},~\eqref{eq:assum:XandW:spectrumUB},~\eqref{eq:assum:XhatX:rate:2} and~\eqref{eq:assum:degreekappa}.
	Applying the above three bounds to Equation~\eqref{eq:dhatdtilde:triangle} yields Equation~\eqref{eq:dhatdtilde:close} after noting that the right-hand side does not depend on our choice of $i \in [n]$.

	To see Equation~\eqref{eq:dhatdtilde:recip:close}, note that for any $i \in [n]$,
	\begin{equation*}
		\left| \frac{ 1 }{ \dhat_i } - \frac{ 1 }{ \dtilde_i } \right|
		\le \frac{ \left| \dhat_i - \dtilde_i \right| }{ \dhat_i \dtilde_i }.
	\end{equation*}
	Using Equation~\eqref{eq:dhatdtilde:triangle} in combination with our assumption in Equation~\eqref{eq:deglb:2}, it holds uniformly over all $i \in [n]$ that
	\begin{equation*}
		\left| \frac{ 1 }{ \dhat_i } - \frac{ 1 }{ \dtilde_i } \right|
		\le
		\frac{ C \left| \dhat_i - \dtilde_i \right| }{ \dtilde_i^2 } .
	\end{equation*}
	A second application of Equation~\eqref{eq:dhatdtilde:triangle} yields Equation~\eqref{eq:dhatdtilde:recip:close} and completes the proof.
\end{proof}

\section{Controlling the Averaging Operators} \label{apx:Gctrl}

\begin{lemma} \label{lem:GGtilde:spectral}
  Suppose that $\A$ follows a sub-gamma model as in Definition~\ref{def:subgamma-network} and suppose that Assumptions~\ref{assum:Apop:spectrum} through~\ref{assum:latentpositions} hold and either of Assumptions~\ref{assum:peer-oracle} or~\ref{assum:latent-oracle} hold.
  Then with high probability,
  \begin{equation*}
    \left\| \G - \Gtilde \right\|
    \le
    C \left( 1 + \frac{ \spop_1 }{ \min_{i \in [n]} \dtilde_i } \right)
    \frac{ \sqrt{ \nu + b^2} \sqrt{n} \log n }
    { \min_{i \in [n]} \dtilde_i } .
  \end{equation*}
\end{lemma}
\begin{proof}
  Recalling the definitions of $\G$ and $\Gtilde$ from Section~\ref{subsec:Gdefs}, the triangle inequality yields
  \begin{equation} \label{eq:GGtilde:spectral:tristart}
    \left\| \mG - \Gtilde \right\|
    \le \left\| \left( \D^{-1} - \Dtilde^{-1} \right) \Apop \right\|
    + \left\| \D^{-1} \left( \A - \Apop \right) \right\| .
  \end{equation}

  By submultiplicativity of the norm and Lemma~\ref{lem:degconc},
  \begin{equation} \label{eq:GGtilde:spectral:term1:done}
    \left\| \left( \D^{-1} - \Dtilde^{-1} \right) \Apop \right\|
    \le
    \spop_1 \max_{i \in [n]} \frac{ \left| d_i - \dtilde_i \right| }
    { d_i \dtilde_i }
    \le
    \frac{ C \spop_1 \sqrt{\nu + b^2} \sqrt{n} \log n }
    { \min_{i \in [n]} \dtilde_i^2 },
  \end{equation}
  where we have used our growth assumption in Equation~\eqref{eq:deglb:2}.

  Similarly, this time using Lemma~\ref{lem:E:spectral},
  \begin{equation} \label{eq:GGtilde:spectral:term2:done}
    \left\| \D^{-1} \left( \A - \Apop \right) \right\|
    \le
    \frac{ C \sqrt{ \nu + b^2} \sqrt{n} \log n }{ \min_{i \in [n]} \dtilde_i } .
  \end{equation}
  Applying the above two displays to Equation~\eqref{eq:GGtilde:spectral:tristart},
  \begin{equation*}
    \left\| \mG - \Gtilde \right\|
    \le
    C \left( 1+ \frac{ \spop_1 }{ \min_{i \in [n]} \dtilde_i }
    \right)
    \frac{ \sqrt{ \nu + b^2} \sqrt{n} \log n }{ \min_{i \in [n]} \dtilde_i },
  \end{equation*}
  as we set out to show.
\end{proof}

\begin{lemma} \label{lem:GhatGtilde:spectral}
  Suppose that $\A$ follows a sub-gamma model as in Definition~\ref{def:subgamma-network} and suppose that Assumptions~\ref{assum:Apop:spectrum} through~\ref{assum:stronger4dhat} hold.
  Then under either Assumption~\ref{assum:peer-oracle} or~\ref{assum:latent-oracle},
  \begin{equation*}
    \left\| \Ghat - \Gtilde \right\|
    \le
    C \left( 1 + \frac{ \spop_1 }{ \min_{i \in [n]} \dtilde_i } \right)
    \frac{ \sqrt{\kappa(\Apop)} \sqrt{ \nu + b^2 } \sqrt{n} \log n }
    { \min_{i \in [n]} \dtilde_i } .
  \end{equation*}
\end{lemma}
\begin{proof}
  Recalling the definitions of $\Ghat$ and $\Gtilde$ from Equation~\eqref{eq:def:Ghat} and Section~\ref{subsec:Gdefs}, respectively, and applying the triangle inequality,
  \begin{equation} \label{eq:GhatGtilde:triangle}
    \left\| \Ghat - \Gtilde \right\|
    \le \left\| \left( \Dhat^{-1} - \mDtilde^{-1} \right) \Apop \right\|
    + \left\| \Dhat^{-1} \left( \Apophat - \Apop \right) \right\| .
  \end{equation}

  Using submultiplicativity of the norm and Lemma~\ref{lem:dhatdtilde},
  \begin{equation} \label{eq:GhatGtilde:triangle:term1:done}
    \left\| \left( \Dhat^{-1} - \mDtilde^{-1} \right) \Apop \right\|
    \le
    \frac{ C \spop_1 \sqrt{\kappa(\Apop)} \sqrt{\nu+b^2} \sqrt{n} \log n }
    { \min_{i\in [n]} \dtilde_i^2 } .
  \end{equation}

  Since $\Apophat$ is a truncation of $\A$, we have $\| \Apophat - \Apop \| \le 2\| \A - \Apop \|$.
  Submultiplicativity and Lemma~\ref{lem:E:spectral} thus yield
  \begin{equation*}
    \left\| \Dhat^{-1} \left( \Apophat - \Apop \right) \right\|
    \le
    \left\| \Dhat^{-1} \right\| \left\| \Apophat - \Apop \right\|
    \le \frac{ C \sqrt{ \nu + b^2 } \sqrt{n} \log n }{ \min_{i \in [n]} \dhat_i }.
  \end{equation*}
  By Lemma~\ref{lem:dhatdtilde} and our assumption in Equation~\eqref{eq:deglb:2}, $\min_i \dhat \ge C \min_i \dtilde_i$, and thus
  \begin{equation} \label{eq:GhatGtilde:triangle:term2:done}
    \left\| \Dhat^{-1} \left( \Apophat - \Apop \right) \right\|
    \le
    \frac{ C \sqrt{ \nu + b^2 } \sqrt{n} \log n }
    { \min_{i \in [n]} \dtilde_i }.
  \end{equation}
  Applying Equations~\eqref{eq:GhatGtilde:triangle:term1:done} and~\eqref{eq:GhatGtilde:triangle:term2:done} to Equation~\eqref{eq:GhatGtilde:triangle},
  \begin{equation*}
    \left\| \Ghat - \Gtilde \right\|
    \le
    C\left( 1 + \frac{ \spop_1 }{ \min_{i\in [n]} \dtilde_i } \right)
    \frac{ \sqrt{ \nu + b^2 } \sqrt{n} \log n }
    { \min_{i \in [n]} \dtilde_i },
  \end{equation*}
  completing the proof.
\end{proof}

\begin{lemma} \label{lem:GhatG:spectral}
  Suppose that $\A$ follows a sub-gamma model as in Definition~\ref{def:subgamma-network} and suppose that Assumptions~\ref{assum:Apop:spectrum} through~\ref{assum:stronger4dhat} hold.
  Then under either Assumption~\ref{assum:peer-oracle} or~\ref{assum:latent-oracle}, with high probability,
  \begin{equation*}
    \left\| \Ghat - \G \right\|
    \le
    C \left( 1 + \frac{ \spop_1 }{ \min_{i \in [n]} \dtilde_i } \right)
    \frac{ \sqrt{ \nu + b^2 } \sqrt{n} \log n }{ \min_{i \in [n]} \dtilde_i }.
  \end{equation*}
\end{lemma}
\begin{proof}
  By the triangle inequality,
  \begin{equation*}
    \left\| \Ghat - \G \right\|
    \le \left\| \Ghat - \Gtilde \right\| + \left\| \G - \Gtilde \right\| .
  \end{equation*}
  Applying Lemmas~\ref{lem:GhatGtilde:spectral} and~\ref{lem:GGtilde:spectral} completes the proof.
\end{proof}

\begin{lemma} \label{lem:GGtilde:quad}
  Suppose that $\A$ follows a sub-gamma model as in Definition~\ref{def:subgamma-network} and suppose that Assumptions~\ref{assum:Apop:spectrum} through~\ref{assum:latentpositions} hold.
  Let $\vu,\vv \in \R^n$ be such that $\A - \Apop$ is independent of $\vv,\vu$ conditional on $\Xpop$.
  Then, under either Assumption~\ref{assum:peer-oracle} or~\ref{assum:latent-oracle},
  \begin{equation*}
    \left| \vu^\top \left( \G - \Gtilde \right) \vv \right|
    = \op{ \frac{ \| \vu \| \| \vv \| }{ \sqrt{n} } } . %
  \end{equation*}
\end{lemma}
\begin{proof}
  Recalling the definitions of $\G$ and $\Gtilde$ from Section~\ref{subsec:Gdefs}, the triangle inequality yields
  \begin{equation} \label{eq:uGGtildev:tristart}
    \left| \vu^\top \left( \G - \Gtilde \right) \vv \right|
    \le
    \left| \vu^\top \left( \D^{-1} - \Dtilde^{-1} \right) \A \vv \right|
    +
    \left| \vu^\top \Dtilde^{-1} \left( \A - \Apop \right) \vv \right| .
  \end{equation}

  Using the assumption that $\A-\Apop$ is independent of $\vu$ and $\vv$ conditional on $\Xpop$, Bernstein's inequality yields
  \begin{equation*} %
    \left| \vu^\top \Dtilde^{-1} \left( \A - \Apop \right) \vv \right|
    \le C\sqrt{ \sum_{i,j} \frac{ (\nu+b^2) u_i^2 v_j^2 \log n }{ \dtilde_i^2 } }
    \le \frac{C \| \vu \| \| \vv \| \sqrt{\nu+b^2} \log n}{\min_{i \in[n]} \dtilde_i } .
  \end{equation*}
  Using our growth assumption in Equation~\eqref{eq:deglb:2},
  \begin{equation} \label{eq:uGGtildev:Eterm:littleoh}
    \left| \vu^\top \Dtilde^{-1} \left( \A - \Apop \right) \vv \right|
    = \op{ \frac{ \| \vu \| \| \vv \| }{ \sqrt{n} } } .
  \end{equation}

  Applying the triangle inequality,
  \begin{equation} \label{eq:uGGtildev:Dterm:triangle}
    \left| \vu^\top \left( \D^{-1} - \Dtilde^{-1} \right) \A \vv \right|
    \le
    \left| \vu^\top \left( \D^{-1} - \Dtilde^{-1} \right) \Apop \vv \right|
    +
    \left| \vu^\top \left( \D^{-1} - \Dtilde^{-1} \right) \left( \A - \Apop \right) \vv \right| .
  \end{equation}

  Applying submultiplicativity of the norm followed by Lemmas~\ref{lem:E:spectral} and~\ref{lem:degconc} and using our growth assumption in Equation~\eqref{eq:deglb:2},
  \begin{equation*} %
    \left| \vu^\top \left( \D^{-1} - \Dtilde^{-1} \right) \left( \A - \Apop \right) \vv \right|
    \le
    \| \vu \| \| \vv \| \left\| \A - \Apop \right\|
    \max_{i \in [n]} \frac{ | d_i - \dtilde_i | }{ d_i \dtilde_i }
    \le
    \frac{ C \| \vu \| \| \vv \| (\nu+b^2) n \log^2 n }
    { \min_{i \in [n]} \dtilde_i^2 } .
  \end{equation*}
  Applying our growth assumption in Equation~\eqref{eq:deglb:2},
  \begin{equation} \label{eq:uGGtildev:Dterm:tri:term2:littleoh}
    \left| \vu^\top \left( \D^{-1} - \Dtilde^{-1} \right) \left( \A - \Apop \right) \vv \right|
    = \op{ \frac{ \| \vu \| \| \vv \| }{ \sqrt{n} } } .
  \end{equation}

  Recalling $\Gtilde = \Dtilde^{-1} \Apop$, factoring appropriately and applying the triangle inequality,
  \begin{equation} \label{eq:uGGtildev:Dterm:tri:term1:expand}
    \left| \vu^\top \!\left( \D^{-1} \! - \! \Dtilde^{-1} \right) \Apop \vv \right|
    \le
    \left| \vu^\top \Dtilde^{-1} \! \left( \D \!-\! \Dtilde \right) \Gtilde \vv \right|
    +
    \left| \vu^\top \! \left(\D^{-1} \! - \!\Dtilde^{-1} \right)
    \left( \D \!-\! \Dtilde \right) \Gtilde \vv \right| .
  \end{equation}
  By submultiplicativity of the norm and the Cauchy-Schwarz inequality,
  \begin{equation*} %
    \left| \vu^\top \! \left(\D^{-1} \!-\! \Dtilde^{-1} \right) \!
    \left( \D \!-\! \Dtilde \right) \Gtilde \vv \right|
    \le
    \| \vu \| \| \vv \| \left\| \D^{-1} - \Dtilde^{-1} \right\|
    \left\| \D-\Dtilde \right\| ,
  \end{equation*}
  where we have used the fact that $\| \Gtilde \| = 1$.
  Applying Lemmas~\ref{lem:degconc} and~\ref{lem:degrecip:conc} and using our growth assumption in Equation~\eqref{eq:deglb:2},
  \begin{equation}  \label{eq:uGGtildev:Dterm:tri:term1:b:littleoh}
    \left| \vu^\top \left(\D^{-1} - \Dtilde^{-1} \right)
    \left( \D - \Dtilde \right) \Gtilde \vv \right|
    \le
    \frac{ C \| \vu \| \| \vv \| (\nu + b^2) n \log^2 n }{ \min_{i \in [n]} \dtilde_i^2 }
    = \op{ \frac{ \| \vu \| \| \vv \| }{ \sqrt{n} } } .
  \end{equation}

  Expanding the matrix-vector products and rearranging,
  \begin{equation} \label{eq:uGGtildev:DinvDgapDinv:quad}
    \vu^\top \Dtilde^{-1} \left( \D - \Dtilde \right) \Gtilde \vv
    = \sum_{i=1}^n \sum_{j=1}^n
    \frac{ \left( d_i - \dtilde_i \right) u_i }{ \dtilde_i^2 } \Xpop_i^\top \Xpop_j v_j
  \end{equation}
  Define the matrix $\mR \in \R^{n \times n}$ according to
  \begin{equation} \label{eq:def:R}
    R_{ij} = \frac{ u_i \Xpop_i^\top \Xpop_j }{ \dtilde_i^2 }.
  \end{equation}
  Recalling that the degrees are the row sums of the adjacency matrices, we can rewrite Equation~\eqref{eq:uGGtildev:DinvDgapDinv:quad} as
  \begin{equation*}
    \vu^\top \Dtilde^{-1} \left( \D - \Dtilde \right) \Gtilde \vv
    = \left[ \left( \A - \Apop \right) \onevec \right]^\top \mR \vv
    = \onevec^\top \left( \A - \Apop \right) \mR \vv.
  \end{equation*}
  Thus, expanding out the product,
  \begin{equation*}
    \vu^\top \Dtilde^{-1} \left( \D - \Dtilde \right) \Gtilde \vv
    = \sum_{i=1}^n \sum_{j=1}^n \left( \A - \Apop \right)_{ij}
    \left(\mR \vv\right)_j
    = \sum_{i<j} \left( \A - \Apop \right)_{ij}
    \left[ \left(\mR \vv\right)_j + \left(\mR \vv\right)_i \right],
  \end{equation*}
  where we have used the fact that $\A-\Apop$ is symmetric with zero diagonal by assumption (though note that the case of non-zero diagonal can be handled straightforwardly). Since the entries of $\A - \Apop$ are $(\nu,b)$-subgamma random variables, conditionally independent (up to symmetry) given $\vv$ and $\mR$, standard concentration inequalities imply that with high probability,
  \begin{equation*}
    \left| \vu^\top \Dtilde^{-1} \left( \D - \Dtilde \right) \Gtilde \vv \right|
    \le
    C \sqrt{ (\nu\!+\!b^2) \sum_{i<j} \left[ \left(\mR \vv\right)_j + \left(\mR \vv\right)_i \right]^2 }
    \le C \sqrt{ (\nu\!+\!b^2) \sum_{i=1}^n \left(\mR \vv\right)_i^2 n \log n },
  \end{equation*}
  where we have used the inequality $(x+y)^2 \le 2(x^2+y^2)$. It follows that, with high probability,
  \begin{equation*}
    \left| \vu^\top \Dtilde^{-1} \! \left( \D \!-\! \Dtilde \right)
    \Gtilde \vv \right|
    \le
    C \sqrt{\nu\!+\!b^2} \left\| \mR \vv \right\| \!\sqrt{n} \log n
    \le
    C \| \vv \| \sqrt{\nu\!+\!b^2}
    \sqrt{ \sum_{i,j} \frac{ u_i^2 (\Xpop_i^\top \Xpop_j)^2 }{ \dtilde_i^4 } } \sqrt{n} \log n .
  \end{equation*}
  Using the Cauchy-Schwarz inequality and the fact that all summands inside the square root are non-negative,
  \begin{equation}  \label{eq:uGGtildev:Dterm:tri:term1:a:littleoh}
    \left| \vu^\top \Dtilde^{-1} \left( \D - \Dtilde \right) \Dtilde^{-1} \Apop \vv \right|
    \le
    \frac{ C \| \vu \| \| \vv \| \sqrt{\nu+b^2} \| \Xpop \|_F \| \Xpop \|_{\tti} \sqrt{n} \log n }
    { \min_{i \in [n]} \dtilde_i^2 }
    = \op{ \frac{ \|\vu\| \|\vv\| }{ \sqrt{n} } } ,
  \end{equation}
  where the equality follows from our growth assumptions in Equations~\eqref{eq:deglb:2},~\eqref{eq:assum:Xtti} and~\eqref{eq:assum:XandW:spectrumUB}.

  Applying Equations~\eqref{eq:uGGtildev:Dterm:tri:term1:b:littleoh} and~\eqref{eq:uGGtildev:Dterm:tri:term1:a:littleoh} to Equation~\eqref{eq:uGGtildev:Dterm:tri:term1:expand},
  \begin{equation} \label{eq:uGGtildev:Dterm:tri:term1:littleoh}
    \left| \vu^\top \left( \D^{-1} - \Dtilde^{-1} \right) \Apop \vv \right|
    = \op{ \frac{ \|\vu\| \|\vv\| }{ \sqrt{n} } } .
  \end{equation}

  Applying Equations~\eqref{eq:uGGtildev:Dterm:tri:term2:littleoh}
  and~\eqref{eq:uGGtildev:Dterm:tri:term1:littleoh}
  to Equation~\eqref{eq:uGGtildev:Dterm:triangle},
  \begin{equation*} %
    \left| \vu^\top \left( \D^{-1} - \Dtilde^{-1} \right) \A \vv \right|
    = \op{ \frac{ \| \vu \| \| \vv \| }{ \sqrt{n} } } ,
  \end{equation*}
  and applying this and Equation~\eqref{eq:uGGtildev:Eterm:littleoh} to Equation~\eqref{eq:uGGtildev:tristart} completes the proof.
\end{proof}

\begin{lemma} \label{lem:GhatGtilde:quad}
  Suppose that $\A$ follows a sub-gamma model as in Definition~\ref{def:subgamma-network} and suppose that Assumptions~\ref{assum:Apop:spectrum} through~\ref{assum:stronger4dhat} hold.
  Let $\vu,\vv \in \R^n$ be such that $\A - \Apop$ is independent of $\vv,\vu$ conditional on $\Xpop$.
  Then under either Assumption~\ref{assum:peer-oracle} or~\ref{assum:latent-oracle},
  \begin{equation*}
    \left| \vu^\top \left( \Ghat - \Gtilde \right) \vv \right|
    = \op{ \frac{ \| \vu \| \| \vv \| }{ \sqrt{n} } } .
  \end{equation*}
\end{lemma}
\begin{proof}
  Recalling $\Ghat$ and $\Gtilde$ as given in Equation~\eqref{eq:def:Ghat} and Section~\ref{subsec:Gdefs}, respectively,
  \begin{equation} \label{eq:GhatGtilde:tri}
    \left| \vu^\top \left( \Ghat - \Gtilde \right) \vv \right|
    \le
    \left| \vu^\top \left(\Dhat^{-1} - \mDtilde^{-1} \right) \Apophat \vv \right|
    + \left| \vu^\top \Dtilde^{-1} \left( \Apophat - \Apop \right) \vv \right| .
  \end{equation}

  By definition of $\Apophat$ and $\Apop$, the triangle inequality yields
  \begin{equation} \label{eq:uDinvAdiffv:tri}
    \begin{aligned}
      \left| \vu^\top \Dtilde^{-1} \left( \Apophat - \Apop \right) \vv \right|
       & \le
      \left| \vu^\top \Dtilde^{-1} \left( \Q \Xhat - \Xpop \right) \Xpop^\top \vv
      \right|
      +
      \left| \vu^\top \Dtilde^{-1} \Xpop \left( \Q \Xhat - \Xpop \right)^\top \vv
      \right|       \\
       & ~~~~~~~~~+
      \left| \vu^\top \Dtilde^{-1} \left( \Q \Xhat - \Xpop \right)
      \left( \Q \Xhat - \Xpop \right)^\top \vv \right| .
    \end{aligned} \end{equation}
  By submultiplicativity and Lemma~\ref{lem:XhatXTB:ctrl},
  \begin{equation*}
    \left| \vu^\top \Dtilde^{-1} \left( \Q \Xhat - \Xpop \right) \Xpop^\top \vv
    \right|
    \le
    \frac{ C \| \vu \| \| \vv \| \| \Xpop \| }{ \min_i \dtilde_i }
    \left( \frac{ \sqrt{ \nu + b^2 } \log n }{ \sqrt{\spop_d} }
    +
    \frac{ \kappa(\Apop) (\nu+b^2) n \log^2 n }{ \spop_d^{3/2} }
    \right),
  \end{equation*}
  and our growth assumptions in Equations~\eqref{eq:deglb:2} and~\eqref{eq:deglb:2:kappa}
  yield
  \begin{equation*}
    \left| \vu^\top \Dtilde^{-1} \left( \Q \Xhat - \Xpop \right) \Xpop^\top \vv
    \right|
    = \op{ \frac{ \| \vu \| \| \vv \| }{ \sqrt{n} } }.
  \end{equation*}
  A near-identical argument yields
  \begin{equation*}
    \left| \vu^\top \Dtilde^{-1} \Xpop \left( \Xhat \Q - \Xpop \right)^\top \vv
    \right|
    = \op{ \frac{ \| \vu \| \| \vv \| }{ \sqrt{n} } }.
  \end{equation*}
  Applying the above two bounds to Equation~\eqref{eq:uDinvAdiffv:tri},
  \begin{equation} \label{eq:uDinvAdiffv:step2}
    \left| \vu^\top \Dtilde^{-1} \left( \Apophat - \Apop \right) \vv \right|
    \le
    \left| \vu^\top \Dtilde^{-1} \left( \Q \Xhat - \Xpop \right)
    \left( \Xhat \Q - \Xpop \right)^\top \vv \right|
    + \op{ \frac{ \| \vu \| \| \vv \| }{ \sqrt{n} } }.
  \end{equation}

  By Cauchy-Schwarz and two applications of Lemma~\ref{lem:XhatXTB:ctrl},
  \begin{equation} \label{eq:GhatGtilde:halfway:preliminary}
    \left| \vu^\top \Dtilde^{-1} \! \left( \!\Xhat \Q \!-\! \Xpop \!\right)
    \! \left( \! \Q \Xhat \!-\! \Xpop \!\right)^\top \!\! \vv \right|
    \le
    \frac{ C \| \vu \| \| \vv \| }{ \min_i \dtilde_i } \!
    \left( \frac{ ( \nu \!+\! b^2 ) \log^2 \!n }{ \spop_d }
    +
    \frac{ \kappa^2(\Apop) (\nu\!+\!b^2)^2 n^2 \log^4 \!n }
      { \spop_d^3 } \!
    \right) .
  \end{equation}

  Since the largest eigenvalue of the adjacency matrix is an upper bound on the minimum degree, %
  we have
  \begin{equation*}
    \frac{ \kappa^2(\Apop) (\nu + b^2)^2 n^2 \log^4 \!n }
    { \spop_d^3 \min_i \dtilde_i }
    =
    \frac{ \min_i \dtilde_i }{ \spop_d }
    \frac{ \kappa^2(\Apop) (\nu + b^2)^2 n^2 \log^4 \!n }
    { \spop_d^2 \min_i \dtilde_i^2 }
    \le
    \frac{ \kappa^3(\Apop) (\nu + b^2)^2 n^2 \log^4 \!n }
    { \spop_d^2 \min_i \dtilde_i^2 },
  \end{equation*}
  and our assumption in Equation~\eqref{eq:deglb:2:kappa} implies
  \begin{equation} \label{eq:deglb:2:implied}
    \frac{ \kappa^2(\Apop) (\nu + b^2)^2 n^2 \log^4 \!n }
    { \spop_d^3 \min_i \dtilde_i }
    = \op{ \frac{1}{n} } .
  \end{equation} %
  Applying this to Equation~\eqref{eq:GhatGtilde:halfway:preliminary},
  along with our growth assumption in Equation~\eqref{eq:deglb:2:kappa}
  \begin{equation*}
    \left| \vu^\top \Dtilde^{-1}  \left( \Xhat \Q - \Xpop \right)
    \left(  \Q \Xhat - \Xpop \right)^\top \vv \right|
    = \op{ \frac{ \| \vu \| \| \vv \| }{ \sqrt{n} } }.
  \end{equation*}
  Applying this to Equation~\eqref{eq:uDinvAdiffv:step2},
  \begin{equation*} %
    \left| \vu^\top \Dtilde^{-1} \left( \Apophat - \Apop \right) \vv \right|
    = \op{ \frac{ \| \vu \| \| \vv \| }{ \sqrt{n} } }.
  \end{equation*}
  Applying this to Equation~\eqref{eq:GhatGtilde:tri} in turn,
  \begin{equation} \label{eq:GhatGtilde:halfway}
    \left| \vu^\top \left( \Ghat - \Gtilde \right) \vv \right|
    \le
    \left| \vu^\top \left( \Dhat^{-1} - \mDtilde^{-1} \right) \Apophat \vv \right|
    + \op{ \frac{ \| \vu \| \| \vv \| }{ \sqrt{n} } }.
  \end{equation}

  By the triangle inequality,
  \begin{equation} \label{eq:dhatdtilde:otherHalf}
    \left| \vu^\top \left( \Dhat^{-1} - \mDtilde^{-1} \right) \Apophat \vv \right|
    \le
    \left| \vu^\top \left( \Dhat^{-1} - \mDtilde^{-1} \right)
    \left( \Apophat - \Apop \right) \vv \right|
    + \left| \vu^\top \left( \Dhat^{-1} - \mDtilde^{-1} \right)
    \Apop \vv \right| .
  \end{equation}
  By submultiplicativity,
  \begin{equation*}
    \left| \vu^\top \left( \Dhat^{-1} - \mDtilde^{-1} \right )
    \left( \Apophat - \Apop \right) \vv \right|
    \le
    \| \vu \| \| \vv \| \left\| \Dhat^{-1} - \mDtilde^{-1} \right\|
    \left\| \Apophat - \Apop \right\| .
  \end{equation*}
  Upper bounding $\| \Apophat - \Apop \| \le 2\|\A-\Apop\|$, since $\Apophat$ is a truncation of $\A$, Lemma~\ref{lem:E:spectral} yields
  \begin{equation*} \begin{aligned}
      \left| \vu^\top \left( \Dhat^{-1} - \mDtilde^{-1} \right )
      \left( \Apophat - \Apop \right) \vv \right|
       & \le C\| \vu \| \| \vv \|
      \left( \max_i \frac{ |\dhat_i - \dtilde_i| }{ \dtilde_i^2 } \right)
      \sqrt{ \nu + b^2 } \sqrt{n} \log n .
    \end{aligned} \end{equation*}
  Applying Lemma~\ref{lem:dhatdtilde} and using the growth assumption in Equation~\eqref{eq:deglb:2},
  \begin{equation*}
    \left| \vu^\top \left( \Dhat^{-1} - \mDtilde^{-1} \right )
    \left( \Apophat - \Apop \right) \vv \right|
    = \op{ \frac{ \| \vu \| \| \vv \| }{ \sqrt{n} } }.
  \end{equation*}
  Applying this to Equation~\eqref{eq:dhatdtilde:otherHalf} and applying the result to Equation~\eqref{eq:GhatGtilde:halfway},
  \begin{equation} \label{eq:GhatGtilde:quad:threeQuarters}
    \left| \vu^\top \left( \Ghat - \Gtilde \right) \vv \right|
    \le
    \left| \vu^\top \left( \Dhat^{-1} - \mDtilde^{-1} \right) \Apop \vv \right|
    + \op{ \frac{ \| \vu \| \| \vv \| }{ \sqrt{n} } }.
  \end{equation}

  Recalling $\Gtilde = \Dtilde^{-1} \Apop$, factoring appropriately and applying the triangle inequality yields
  \begin{equation} \label{eq:GhatGtilde:quad:threeQuarters:triangle}
    \left| \vu^\top \left( \Dhat^{-1} - \mDtilde^{-1} \right) \Apop \vv \right|
    \le
    \left| \vu^\top \Dtilde^{-1} \left( \Dhat - \mDtilde \right)
    \Gtilde \vv \right|
    +
    \left| \vu^\top \left( \Dhat^{-1} - \Dtilde^{-1} \right)
    \left( \Dhat - \mDtilde \right) \Gtilde \vv \right| .
  \end{equation}

  By submultiplicativity and using our growth assumption in Equation~\eqref{eq:deglb:2},
  \begin{equation*}
    \left| \vu^\top \left( \Dhat^{-1} - \Dtilde^{-1} \right)
    \left( \Dhat - \mDtilde \right) \Gtilde \vv \right|
    \le \frac{ C \| \vu \| \| \vv \| \left\| \Dhat - \Dtilde \right\|^2 }
    { \min_i \dtilde_i^2 } .
  \end{equation*}
  where we have used the fact that $\| \Gtilde \| = 1$.
  Applying Lemma~\ref{lem:dhatdtilde} and using the growth assumption in Equation~\eqref{eq:deglb:2}, it follows that
  \begin{equation} \label{eq:GhatGtilde:quad:threeQuarters:tri:easyTerm}
    \left| \vu^\top \left( \Dhat^{-1} - \Dtilde^{-1} \right)
    \left( \Dhat - \mDtilde \right) \Gtilde \vv \right|
    = \op{ \frac{ \| \vu \| \| \vv \| }{ \sqrt{n} } } .
  \end{equation}
  Defining $\mR \in \R^{n \times n}$ as in Equation~\eqref{eq:def:R}, recalling that the degrees $\dhat_i$ and $\dtilde_i$ are given by row-sums of $\Apophat$ and $\Apop$, respectively,
  \begin{equation*}
    \vu^\top \Dtilde^{-1} \left( \Dhat - \mDtilde \right) \Gtilde \vv
    = \onevec^\top \left( \Apophat - \Apop \right) \mR \vv .
  \end{equation*}
  By the triangle inequality,
  \begin{equation} \label{eq:GhatGtilde:quad:annoying}
    \begin{aligned}
      \left| \vu^\top \Dtilde^{-1} \left( \Apophat - \Apop \right) \vv \right|
       & \le
      \left| \vu^\top \Dtilde^{-1} \left( \Xhat \Q - \Xpop \right)
      \Xpop^\top \vv \right|
      +
      \left| \vu^\top \Dtilde^{-1} \Xpop
      \left( \Xhat \Q - \Xpop \right)^\top \vv \right| \\
       & ~~~~~~+
      \left| \vu^\top \Dtilde^{-1} \left( \Xhat \Q - \Xpop \right)
      \left( \Xhat \Q - \Xpop \right)^\top \vv \right| .
    \end{aligned} \end{equation}

  Applying submultiplicativity and Lemma~\ref{lem:XhatXTB:ctrl} with $\mB = \Dtilde^{-1} \vu$,
  \begin{equation*}
    \left| \vu^\top \Dtilde^{-1} \! \left( \Xhat \Q - \Xpop \right)
    \! \Xpop^\top \vv \right|
    \le \frac{ C\| \vu \| \| \vv \| }{ \min_i \dtilde_i }
    \left( \!
    \sqrt{ \kappa } \sqrt{ \nu + b^2 } \log n
    +
    \frac{ \kappa^{3/2} (\nu+b^2) n \log^2 n }{ \spop_d }
    \! \right)
    = \op{ \frac{ \| \vu \| \| \vv \| }{ \sqrt{n} } },
  \end{equation*}
  where we have used our growth assumptions in Equations~\eqref{eq:deglb:2} and~\eqref{eq:deglb:2:kappa}.
  A near-identical argument yields
  \begin{equation*}
    \left| \vu^\top \Dtilde^{-1} \Xpop
    \left( \Xhat \Q - \Xpop \right)^\top \vv \right|
    = \op{ \frac{ \| \vu \| \| \vv \| }{ \sqrt{n} } }.
  \end{equation*}
  Applying the above two displays to Equation~\eqref{eq:GhatGtilde:quad:annoying},
  \begin{equation} \label{eq:GhatGtilde:quad:annoying:2}
    \left| \vu^\top \Dtilde^{-1} \left( \Apophat - \Apop \right) \vv \right|
    \le
    \left| \vu^\top \Dtilde^{-1} \left( \Xhat \Q - \Xpop \right)
    \left( \Xhat \Q - \Xpop \right)^\top \vv \right|
    + \op{ \frac{ \| \vu \| \| \vv \| }{ \sqrt{n} } } .
  \end{equation}

  Applying Cauchy-Schwarz and invoking Lemma~\ref{lem:XhatXTB:ctrl} two more times,
  \begin{equation*} \begin{aligned}
    \left| \vu^\top \Dtilde^{-1} \left( \Xhat \Q - \Xpop \right)
    \left( \Xhat \Q - \Xpop \right)^\top \vv \right|
    &\le
    \frac{ C \| \vu \| \| \vv \| }{ \min_i \dtilde_i }
    \left( \!
    \frac{ ( \nu + b^2 ) \log^2 n }{ \spop_d }
    +
    \frac{ \kappa^2 (\nu+b^2)^2 n^2 \log^4 n }{ \spop_d^{3} }
    \! \right) \\
    &= \op{ \frac{ \| \vu \| \| \vv \| }{ \sqrt{n} } },
  \end{aligned} \end{equation*}
  where we have used our growth assumption in
  Equation~\eqref{eq:deglb:2:kappa} %
  and the bound in Equation~\eqref{eq:deglb:2:implied}. %
  Applying this to Equation~\eqref{eq:GhatGtilde:quad:annoying:2},
  \begin{equation*}
    \left| \vu^\top \Dtilde^{-1} \left( \Apophat - \Apop \right) \vv \right|
    \le
    \left| \vu^\top \Dtilde^{-1} \left( \Xhat \Q - \Xpop \right)
    \left( \Xhat \Q - \Xpop \right)^\top \vv \right|
    + \op{ \frac{ \| \vu \| \| \vv \| }{ \sqrt{n} } } .
  \end{equation*}
  Applying this to Equation~\eqref{eq:GhatGtilde:quad:threeQuarters} in turn completes the proof.
\end{proof}

\begin{lemma} \label{lem:GhatG:quad}
  Suppose that $\A$ follows a sub-gamma model as in Definition~\ref{def:subgamma-network} and suppose that Assumptions~\ref{assum:Apop:spectrum} through~\ref{assum:stronger4dhat} hold.
  Let $\vu,\vv \in \R^n$ be such that $\A - \Apop$ is independent of $\vv,\vu$ conditional on $\Xpop$.
  Then under either Assumption~\ref{assum:peer-oracle} or~\ref{assum:latent-oracle},
  \begin{equation*}
    \left| \vu^\top \left( \Ghat - \G \right) \vv \right|
    = \op{ \frac{ \| \vu \| \| \vv \| }{ \sqrt{n} } } .
  \end{equation*}
\end{lemma}
\begin{proof}
  By the triangle inequality,
  \begin{equation*}
    \left| \vu^\top \left( \Ghat - \G \right) \vv \right|
    \le \left| \vu^\top \left( \Ghat - \Gtilde \right) \vv \right|
    + \left| \vu^\top \left( \G - \Gtilde \right) \vv \right|.
  \end{equation*}
  Applying Lemmas~\ref{lem:GhatGtilde:quad} and~\ref{lem:GGtilde:quad} completes the proof.
\end{proof}

\section{Controlling the Responses} \label{apx:Yctrl}

\begin{lemma} \label{lem:Y:latcon}
	Under the latent contagion model in Equation~\eqref{eq:lim-latent-red},
	suppose that Assumptions~\ref{assum:Apop:spectrum} through~\ref{assum:latentpositions} and Assumption~\ref{assum:latent-oracle} hold.
	Then
	\begin{equation*}
		\left\| \Y \right\| = \Op{ \sqrt{n} } .
	\end{equation*}
\end{lemma}
\begin{proof}
	Recalling Equation~\eqref{eq:lim-latent-red} and applying the triangle inequality and submultiplicativity,
	\begin{equation*}
		\left\| \Y \right\|
		\le
		\left\| \paren*{\mI - \thetay \Gtilde}^{-1} \right\|
		\left\| \1_n \thetanaught + \W \thetaw + \Xpop \thetax \right\|
		+
		\left\| \paren*{\mI - \thetay \Gtilde}^{-1} \be \right\| .
	\end{equation*}
	Controlling the first term using Lemma~\ref{lem:IbgyG:invertible} and the triangle inequality, and controlling the second term using Lemma~\ref{lem:epstilde:ctl}, %
	\begin{equation*}
		\left\| \Y \right\|
		\le
		C\left( \frac{\thetanaught \sqrt{n} }{|1-\thetay|} + \| \W \thetaw \| + \| \Xpop \| \| \thetax \| \right)
	\end{equation*}
	The proof is complete after applying the growth rate in Equation~\eqref{eq:assum:XandW:spectrumUB} and using the fact that the model parameters are constant in $n$.
\end{proof}

\begin{lemma} \label{lem:Y:peercon}
	Under the peer contagion model in Equation~\eqref{eq:lim-peer-red},
	suppose that Assumptions~\ref{assum:Apop:spectrum} through~\ref{assum:latentpositions} and Assumption~\ref{assum:peer-oracle} hold.
	Then
	\begin{equation*}
		\Y = \Ytilde + \zeta_n,
	\end{equation*}
	where $\A-\Apop$ and $\Ytilde$ are independent given $\Xpop$,
	$\| \Ytilde \| = \Op{ \sqrt{n} }$, and $\| \zeta_n \| = \op{ \sqrt{n} }$.
	Further,
	\begin{equation*}
		\left\| \Y \right\| = \Op{ \sqrt{n} } .
	\end{equation*}
\end{lemma}
\begin{proof}
	Recalling the definition of $\Gtilde$ from Section~\ref{subsec:Gdefs}, define
	\begin{equation*}
		\Ytilde = \paren*{\mI - \betay \Gtilde }^{-1} \paren*{\1_n \betanaught + \W \betaw + \Xpop \betax  + \be}
	\end{equation*}
	and
	\begin{equation} \label{eq:def:zetan}
		\zeta_n
		= \left[ \paren*{\mI - \betay \G }^{-1} - \paren*{\mI - \betay \Gtilde }^{-1} \right]
		\paren*{\1_n \betanaught + \W \betaw + \Xpop \betax  + \be} .
	\end{equation}
	Then adding and subtracting appropriate quantities in Equation~\eqref{eq:lim-peer-red},
	\begin{equation*}
		\Y %
		= \Ytilde + \zeta_n .
	\end{equation*}
	An argument parallel to that given in the proof of Lemma~\ref{lem:Y:latcon} above yields
	\begin{equation*}
		\left\| \Ytilde \right\| = \Op{ \sqrt{n} } ,
	\end{equation*}
	so our proof will be complete once we control $\zeta_n$.

	For ease of notation, define
	\begin{equation} \label{eq:def:L}
		\mL = \paren*{\1_n \betanaught + \W \betaw + \Xpop \betax  + \be} .
	\end{equation}
	Applying the Neumann expansion and the triangle inequality,
	\begin{equation*}
		\left\| \zeta_n \right\|
		\le \sum_{q=0}^\infty |\betay|^q \left\| \G^q - \Gtilde^q \right\| \left\| \mL \right\|
		= \sum_{q=1}^\infty |\betay|^q \left\| \G^q - \Gtilde^q \right\| \left\| \mL \right\| .
	\end{equation*}
	Expanding $\G^q - \Gtilde^q$ and using the triangle inequality and submultiplicativity,
	\begin{equation}
		\label{eq:Gq-Glatq}
		\norm*{\G^q - \Gtilde^q}
		\le \sum_{k=0}^{q-1} \norm{\G^{q-1-k}} \norm*{\G - \Gtilde} \norm*{\Gtilde^k}
		\le q \norm*{\G - \Gtilde}.
	\end{equation}
	Plugging this into the Neumann expansion above and evaluating the series,
	\begin{equation*}
		\left\| \zeta_n \right\|
		\le \norm*{\mL} \norm*{\G - \Gtilde} \sum_{q=1}^\infty |\betay|^q q
		= \norm*{\mL} \norm*{\G - \Gtilde} \frac{\abs{\betay}}{(1 - \betay)^2} .
	\end{equation*}
	Lemma~\ref{lem:GGtilde:spectral} then yields
	\begin{equation*}
		\left\| \zeta_n \right\|
		\le
		C \| \mL \|
		\left( 1+ \frac{ \spop_1 }{ \min_{i \in [n]} \dtilde_i } \right)
		\frac{ \sqrt{ \nu + b^2} \sqrt{n} \log n }{ \min_{i \in [n]} \dtilde_i }.
	\end{equation*}
	Applying our growth assumption in Equation~\eqref{eq:assum:Grate:stronger} yields
	\begin{equation} \label{eq:zetan:littleoh}
		\| \zeta_n \| = \op{ \| \mL \| }.
	\end{equation}
	Applying the triangle inequality,
	\begin{equation*}
		\| \mL \|
		\le |\betanaught| \| \1_n \| + \|\betaw\| \| \W \|
		+ \| \betax\| \| \Xpop \| + \| \be \|,
	\end{equation*}
	from which Lemma~\ref{lem:sglinear} and our growth rate in Equation~\eqref{eq:assum:XandW:spectrumUB} implies
	\begin{equation} \label{eq:L:upperBound}
		\| \mL \| = \Op{ \sqrt{n} } .
	\end{equation}
	Applying this to Equation~\eqref{eq:zetan:littleoh} completes the proof.
\end{proof}

\begin{lemma} \label{lem:Y:peerquad}
	Suppose that Assumptions~\ref{assum:Apop:spectrum} through~\ref{assum:stronger4dhat} hold and that $\vu \in \R^n$ is such that $\A - \Apop$ is independent of $\vu$ given $\Xpop$.
	Then under either of the models in Equations~\eqref{eq:lim-latent-red} and~\eqref{eq:lim-peer-red},
	\begin{equation*}
		\left| \vu^\top \left( \Ghat - \G \right) \Y \right|
		= \op{ \| \vu \| }. %
	\end{equation*}
\end{lemma}
\begin{proof}
	We note that this bound is trivial under latent contagion given our growth assumptions: simply use Lemma~\ref{lem:GhatG:quad} and the fact that $\| \Y \| = \Op{\sqrt{n}}$.
	Under peer contagion as in Equation~\eqref{eq:lim-peer-red}, the result requires more careful analysis.

	Using Lemma~\ref{lem:Y:peercon}, write $\Y = \Ytilde + \zeta_n$, where $\Ytilde$ and $\zeta_n$ obey the growth rates in Equation~\eqref{eq:Ytilde} and $\Ytilde$ is independent of $\A - \Apop$ conditional on $\Xpop$.
	The triangle inequality then yields
	\begin{equation} \label{eq:Y:peerquad:start}
		\left| \vu^\top \left( \Ghat - \G \right) \Y \right|
		\le
		\left| \vu^\top \left( \Ghat - \G \right) \zeta_n \right|
		+
		\left| \vu^\top \left( \Ghat - \G \right) \Ytilde \right|
		\le
		\left| \vu^\top \left( \Ghat - \G \right) \zeta_n \right|
		+ \op{ \| \vu \| } ,
	\end{equation}
	where the second bound follows from Lemma~\ref{lem:GhatG:quad} and the fact that $\| \Y \| = \Op{\sqrt{n}}$ by construction.
	By the definition of $\zeta_n$ in Equation~\eqref{eq:def:zetan},
	\begin{equation} \label{eq:Y:peerquad:neumann}
		\vu^\top \left( \Ghat - \G \right) \zeta_n
		=
		\sum_{q=0}^\infty
		\betay^q \vu^\top \left( \Ghat - \G \right)
		\left( \G^q - \Gtilde^q \right) \mL,
	\end{equation}
	where $\mL$ is as in Equation~\eqref{eq:def:L}.
	By Lemma~\ref{lem:GhatG:quad} and Equation~\eqref{eq:L:upperBound},
	\begin{equation*}
		\left| \vu^\top \left( \Ghat - \G \right) \mL \right|
		= \op{ \| \vu \| } .
	\end{equation*}
	By submultiplicativity followed by Lemmas~\ref{lem:GhatG:spectral} and~\ref{lem:GGtilde:spectral},
	\begin{equation*}
		\left| \vu^\top \left( \Ghat - \G \right)
		\left( \G^q - \Gtilde^q \right) \mL \right|
		\le
		C \| \vu \| \| \mL \|
		\left( 1 + \frac{ \spop_1 }{ \min_{i \in [n]} \dtilde_i } \right)^2
		\frac{ (\nu + b^2 ) n \log^2 n }{ \min_{i \in [n]} \dtilde_i^2 }
		= \op{ \| \vu \| },
	\end{equation*}
	where the equality follows from Equation~\eqref{eq:L:upperBound} and our growth assumption in Equation~\eqref{eq:assum:Grate:stronger}.
	Applying the triangle inequality in Equation~\eqref{eq:Y:peerquad:neumann} and using the above two bounds,
	\begin{equation} \label{eq:Y:peerquad:neumann:step2}
		\left| \vu^\top \left( \Ghat - \G \right) \zeta_n  \right|
		\le
		\sum_{q=2}^\infty
		\left| \betay^q \vu^\top \left( \Ghat - \G \right)
		\left( \G^q - \Gtilde^q \right) \mL \right|
		+ \op{ \| \vu \| } .
	\end{equation}

	For any $q \ge 2$, the triangle inequality yields
	\begin{equation*} \begin{aligned}
			\left| \vu^\top \left( \Ghat - \G \right)
			\left( \G^q - \Gtilde^q \right) \mL \right|
			\le
			\left| \vu^\top \left( \Ghat - \G \right)
			\Gtilde \left( \G^{q-1} - \Gtilde^{q-1} \right) \mL \right|
			+
			\left| \vu^\top \left( \Ghat - \G \right)
			\left( \G - \Gtilde \right) \G^{q-1} \mL \right| .
		\end{aligned} \end{equation*}
	Recursing, it holds for any $q \ge 2$ that
	\begin{equation*}
		\left| \vu^\top \left( \Ghat - \G \right)
		\left( \G^q - \Gtilde^q \right) \mL \right|
		\le q \| \vu \| \| \mL \| \left\| \Ghat - \G \right\|
		\left\| \G - \Gtilde \right| .
	\end{equation*}
	Applying Lemmas~\ref{lem:GhatG:spectral} and~\ref{lem:GGtilde:spectral} along with the bound in Equation~\eqref{eq:L:upperBound},
	\begin{equation*}
		\left| \vu^\top \left( \Ghat - \G \right)
		\left( \G^q - \Gtilde^q \right) \mL \right|
		\le
		C q \| \vu \| \sqrt{n}
		\left( 1 + \frac{ \spop_1 }{ \min_{i \in [n]} \dtilde_i } \right)^2
		\frac{ (\nu + b^2 ) n \log^2 n }{ \min_{i \in [n]} \dtilde_i^2 } .
	\end{equation*}
	Applying this bound to each of the terms in Equation~\eqref{eq:Y:peerquad:neumann:step2},
	\begin{equation*}
		\left| \vu^\top \left( \Ghat - \G \right) \zeta_n  \right|
		\le
		C \| \vu \| \sqrt{n}
		\left( 1 + \frac{ \spop_1 }{ \min_{i \in [n]} \dtilde_i } \right)^2
		\frac{ (\nu + b^2 ) n \log^2 n }{ \min_{i \in [n]} \dtilde_i^2 }
		\sum_{q=2}^\infty
		q \left| \betay \right|^q
		+ \op{ \| \vu \| } .
	\end{equation*}
	Applying our growth assumption in Equation~\eqref{eq:assum:Grate:stronger},
	\begin{equation*}
		\left| \vu^\top \left( \Ghat - \G \right) \zeta_n  \right|
		= \op{ \| \vu \| } .
	\end{equation*}
	Applying this to Equation~\eqref{eq:Y:peerquad:start} completes the proof.
\end{proof}

\section{Controlling the Instruments} \label{apx:instruments}

We now turn to describing the behavior of our instruments.
As an oracle counterpart to $\Hhat$ in Definition~\ref{def:feasibleEstimators:peer}, we define
\begin{equation} \label{eq:def:HpeerOracle}
	\HpeerOracle = \begin{bmatrix} \W \; \Xpop \; \G \W \; \G \Xpop \;
		\G^2 \W \; \G^2 \Xpop\end{bmatrix}
	\in \R^{n \times (3 p + 3 d)} .
\end{equation}

Similarly, as an oracle counterpart to the instrument $\Hcheck$ in definition~\ref{def:feasibleEstimators:latent}, we define
\begin{equation} \label{eq:def:HlatOracle}
	\HlatOracle = \begin{bmatrix} \W \; \Xpop \; \Gtilde \W \; \Gtilde \Xpop \;
		\Gtilde^2 \W \; \Gtilde^2 \Xpop\end{bmatrix}
	\in \R^{n \times (3 p + 3 d)} .
\end{equation}

The technical results below relate these four different versions of the instruments to one another.
To account for the rotational non-identifiability of $\Xpop$, we must consider appropriate orthogonal rotations of $\Hhat$ and $\Hcheck$, which are given by
\begin{equation} \label{eq:def:Qinst}
\Qinst
= 
\begin{bmatrix} 
\mI_{p \times p} & \zerovec & \zerovec & \zerovec & \zerovec & \zerovec \\
\zerovec 	& \Q 	   & \zerovec & \zerovec & \zerovec & \zerovec \\
\zerovec 	& \zerovec & \mI_{p\times p} & \zerovec & \zerovec & \zerovec \\
\zerovec 	& \zerovec & \zerovec 	    & \Q  & \zerovec & \zerovec \\
\zerovec 	& \zerovec & \zerovec 	    & \zerovec & \mI_{p\times p} & \zerovec \\
\zerovec 	& \zerovec & \zerovec 	    & \zerovec & \zerovec & \Q
\end{bmatrix} \in \R^{3(p+d)\times 3(p+d)},
\end{equation}
where $\Q \in \R^{d \times d}$ is the orthogonal matrix guaranteed by Lemma~\ref{lem:XhatXTB:ctrl}.

\begin{lemma} \label{lem:HhatHlat:spectral}
	Suppose that Assumptions~\ref{assum:Apop:spectrum} through~\ref{assum:stronger4dhat} hold.
	Then under either the peer contagion model in Equation~\eqref{eq:lim-peer-red} with Assumption~\ref{assum:peer-oracle} or the latent contagion model in Equation~\eqref{eq:lim-latent-red} with Assumption~\ref{assum:latent-oracle},
	with $\Qinst \in \R^{(3 p + 3 d) \times (3 p + 3 d)}$ as given in Equation~\eqref{eq:def:Qinst},
	\begin{equation*}
		\left\| \Hhat \Qinsttop - \HlatOracle \right\|
		= \op{ \sqrt{ \sigmamin(\HlatOracle) } }.
	\end{equation*} %
\end{lemma}
\begin{proof}
	Recalling the definitions of $\Hhat$ and $\HlatOracle$ from Definition~\ref{def:feasibleEstimators:peer} and Equation~\eqref{eq:def:HlatOracle}, respectively, and applying the triangle inequality,
	\begin{equation} \label{eq:HhatHlat:triangle} \begin{aligned}
	\left\| \Hhat \Qinsttop - \HlatOracle \right\|
	& \le \left\| \Xhat \Q^\top - \Xpop \right\|
	+ \left\| \left( \Ghat - \Gtilde \right) \W \right\|
	+ \left\| \Ghat \Xhat \Q^\top - \Gtilde \Xpop \right\|                                 \\
	 & ~~~~~~~~~~~~~~~~~~~+ \left\| \left( \Ghat^2 - \Gtilde^2 \right) \W \right\|
	+ \left\| \Ghat^2 \Xhat \Q^\top - \Gtilde^2 \Xpop \right\| .
	\end{aligned} \end{equation}

	Applying Lemma~\ref{lem:XhatXTB:ctrl} with $\mB = \mI$ followed by our growth assumptions in Equations~\eqref{eq:assum:spectralConc},~\eqref{eq:assum:XandW:spectrumUB} and~\eqref{eq:assum:XhatX:rate:2},
	\begin{equation*}
		\left\| \Xhat \Q^\top - \Xpop \right\|
		\le
		\frac{ C \sqrt{ \nu + b^2 } \sqrt{n} \log n }{ \sqrt{\spop_d} }
		+ \frac{ C\kappa(\Apop) (\nu+b^2) n \log^2 n }
		{ \spop_d^{3/2} }
		= \op{ n^{1/4} } .
	\end{equation*}
	Applying our assumption in Equation~\eqref{eq:assum:Hgrowth},
	\begin{equation} \label{eq:HhatHlat:spectral:XhatX}
		\left\| \Xhat\Q^\top - \Xpop \right\|
		= \op{ \sqrt{ \sigmamin(\HlatOracle) } } .
	\end{equation}
	Applying submultiplicativity of the norm and Lemma~\ref{lem:GhatGtilde:spectral},
	\begin{equation*}
	\left\| \left( \Ghat - \Gtilde \right) \W \right\|
	\le
	C \left( 1+ \frac{ \spop_1 }{ \min_{i \in [n]} \dtilde_i } \right)
	\frac{ \sqrt{ \nu + b^2} \sqrt{n} \log n }{ \min_{i \in [n]} \dtilde_i }
	\left\| \W \right\| .
	\end{equation*}
	Applying our growth assumptions in Equations~\eqref{eq:assum:XandW:spectrumUB} and~\eqref{eq:assum:Grate:stronger},
	\begin{equation*}
		\left\| \left( \Ghat - \Gtilde \right) \W \right\|
		= \op{ n^{1/4} } .
	\end{equation*}
	Applying the growth assumption in Equation~\eqref{eq:assum:Hgrowth},
	\begin{equation} \label{eq:GhatGtildeW:spectral}
		\left\| \left( \Ghat - \Gtilde \right) \W \right\|
		= \op{ \sqrt{ \sigmamin(\HlatOracle) } } .
	\end{equation}
	Similarly, by the triangle inequality, submultiplicativity of the norm and the fact that $\Ghat$ and $\Gtilde$ both have unit spectral norm,
	\begin{equation*}
		\left\| \left( \Ghat^2 - \Gtilde^2 \right) \W \right\|
		\le
		\left\| \Ghat \left( \Ghat - \Gtilde \right) \W \right\|
		+ \left\| \left( \Ghat - \Gtilde \right) \Gtilde \W \right\|
		\le
		2 \left\|  \Ghat - \Gtilde \right\| \left\| \W \right\| .
	\end{equation*}
	The same argument as that leading to Equation~\eqref{eq:GhatGtildeW:spectral} then yields
	\begin{equation} \label{eq:Ghat2Gtilde2W:spectral}
		\left\| \left( \Ghat^2 - \Gtilde^2 \right) \W \right\|
		= \op{ \sqrt{ \sigmamin(\HlatOracle) } } .
	\end{equation}

	By the triangle inequality, again using submultiplicativity of the norm and the fact that $\Ghat$ has unit spectral norm,
	\begin{equation*}
		\left\| \Ghat \Xhat \Q^\top - \Gtilde \Xpop \right\|
		\le
		\left\| \Ghat - \Gtilde \right\| \left\| \Xpop \right\|
		+ \left\| \Xhat \Q^\top - \Xpop \right\| ,
	\end{equation*}
	whence Lemma~\ref{lem:GhatGtilde:spectral} and Equation~\eqref{eq:HhatHlat:spectral:XhatX} yield
	\begin{equation*}
		\left\| \Ghat \Xhat \Q^\top - \Gtilde \Xpop \right\|
		\le
		C \sqrt{ \spop_d }
		\left( 1+ \frac{ \spop_1 }{ \min_{i \in [n]} \dtilde_i } \right)
		\frac{ \sqrt{ \nu + b^2} \sqrt{n} \log n }{ \min_{i \in [n]} \dtilde_i }
		+ \op{ \sqrt{ \sigmamin( \HlatOracle ) } } .
	\end{equation*}
	Once more applying the same argument as that leading to Equation~\eqref{eq:GhatGtildeW:spectral}, this time using the growth assumption in Equation~\eqref{eq:assum:XandW:spectrumUB},
	\begin{equation} \label{eq:GhatXhatGtildeXpop:spectral}
		\left\| \Ghat \Xhat \Q^\top - \Gtilde \Xpop \right\|
		= \op{ \sqrt{ \sigmamin(\HlatOracle) } }.
	\end{equation}

	By an argument parallel to that yielding Equation~\eqref{eq:Ghat2Gtilde2W:spectral},
	\begin{equation} \label{eq:GhatGtildesquaresX}
		\left\| \left( \Ghat^2 - \Gtilde^2 \right) \Xpop \right\|
		= \op{ \sqrt{ \sigmamin(\HlatOracle) } }.
	\end{equation}

	Applying Equations~\eqref{eq:HhatHlat:spectral:XhatX},
	~\eqref{eq:GhatGtildeW:spectral},
	~\eqref{eq:Ghat2Gtilde2W:spectral},
	~\eqref{eq:GhatXhatGtildeXpop:spectral},
	and~\eqref{eq:GhatGtildesquaresX}
	to Equation~\eqref{eq:HhatHlat:triangle} completes the proof.
\end{proof}

\begin{lemma} \label{lem:HcheckHlat:spectral}
	Suppose that Assumptions~\ref{assum:Apop:spectrum} through~\ref{assum:latentpositions} hold.
	Then under either the peer contagion model in Equation~\eqref{eq:lim-peer-red} with Assumption~\ref{assum:peer-oracle} or the latent contagion model in Equation~\eqref{eq:lim-latent-red} with Assumption~\ref{assum:latent-oracle},
        with $\Qinst \in \R^{(3 p + 3 d) \times (3 p + 3 d)}$ as given in Equation~\eqref{eq:def:Qinst},
	\begin{equation*}
		\left\| \Hcheck \Qinsttop - \HlatOracle \right\|
		= \op{ \sqrt{ \sigmamin(\HlatOracle) } } .
	\end{equation*} %
\end{lemma}
\begin{proof}
	Recalling the definitions of $\Hcheck$, $\HlatOracle$ and $\Qinst$, the triangle inequality yields
	\begin{equation*} %
	\begin{aligned}
		\left\| \Hcheck \Qinsttop - \HlatOracle \right\|
		 & \le \left\| \Xhat \Q^\top - \Xpop \right\|
		+ \left\| \left( \G - \Gtilde \right) \W \right\|
		+ \left\| \G \Xhat \Q^\top - \Gtilde \Xpop \right\|                                 \\
		 & ~~~~~~~~~~~~~~~~~~~+ \left\| \left( \G^2 - \Gtilde^2 \right) \W \right\|
		+ \left\| \G^2 \Xhat \Q^\top - \Gtilde^2 \Xpop \right\| .
	\end{aligned} \end{equation*}
	The result then follows from the same argument as in Lemma~\ref{lem:HhatHlat:spectral}, using Lemma~\ref{lem:GGtilde:spectral} in place of Lemma~\ref{lem:GhatGtilde:spectral}.
\end{proof}

\begin{lemma} \label{lem:HcheckHpeer:spectral}
	Suppose that Assumptions~\ref{assum:Apop:spectrum} through~\ref{assum:latentpositions} hold.
	Then under either the peer contagion model in Equation~\eqref{eq:lim-peer-red} with Assumption~\ref{assum:peer-oracle} or the latent contagion model in Equation~\eqref{eq:lim-latent-red} with Assumption~\ref{assum:latent-oracle},
	with $\Hcheck$ and $\HpeerOracle$ as given in Definition~\ref{def:feasibleEstimators:peer} and Equation~\eqref{eq:def:HpeerOracle}, respectively, and with $\Qinst \in \R^{3(p+d)\times 3(p+d)}$ as given in Equation~\eqref{eq:def:Qinst},
	\begin{equation*}
	\left\| \Hcheck \Qinsttop - \HpeerOracle \right\|
	= \op{ \sqrt{ \sigmamin(\HlatOracle) } }.
	\end{equation*} %
\end{lemma}
\begin{proof}
	Recalling the definitions of $\Hcheck$ and $\HpeerOracle$ and applying the triangle inequality,
	\begin{equation*} \begin{aligned}
		\left\| \Hcheck \Qinsttop - \HpeerOracle \right\|
		 & \le \left\| \Xhat \Q^\top - \Xpop \right\|
		+ \left\| \G \left( \Xhat \Q^\top - \Xpop \right) \right\|
		+ \left\| \G^2 \left( \Xhat \Q^\top - \Xpop \right) \right\|
		\le 3 \left\| \Xhat \Q^\top - \Xpop \right\|,
	\end{aligned} \end{equation*}
	where the second inequality follows from submultiplicativity and the fact that $\| \G \| = 1$.
	Applying Lemma~\ref{lem:XhatXTB:ctrl} with $\mB = \mI$ followed by our assumption in Equations~\eqref{eq:assum:spectralConc},~\eqref{eq:assum:XandW:spectrumUB} and~\eqref{eq:assum:XhatX:rate:2},
	\begin{equation*}
		\left\| \Hcheck \Qinsttop - \HpeerOracle \right\|
		\le
		\frac{ C \sqrt{ \nu + b^2 } \sqrt{n} \log n }{ \sqrt{\spop_d} }
		+
		\frac{ C \kappa(\Apop) (\nu+b^2) n \log^2 n }{ \spop_d^{3/2} }
		= \op{ \sqrt{n} }.
	\end{equation*}
	Applying our growth assumption in Equation~\eqref{eq:assum:Hgrowth} completes the proof.
\end{proof}

\begin{lemma} \label{lem:HpeerHlat:spectral}
	Suppose that Assumptions~\ref{assum:Apop:spectrum} through~\ref{assum:latentpositions} hold.
	Then under either the peer contagion model in Equation~\eqref{eq:lim-peer-red} with Assumption~\ref{assum:peer-oracle} or the latent contagion model in Equation~\eqref{eq:lim-latent-red} with Assumption~\ref{assum:latent-oracle},
	with $\HpeerOracle$ and $\HlatOracle$ as given in Definitions~\ref{def:feasibleEstimators:peer} and~\ref{def:feasibleEstimators:latent}, 
	\begin{equation*}
		\left\| \HpeerOracle - \HlatOracle \right\|
		= \op{ \sqrt{ \sigmamin(\HlatOracle) } }.
	\end{equation*} %
\end{lemma}
\begin{proof}
	By the triangle inequality, letting $\Qinst \in \R^{(3 p + 3 d) \times (3 p + 3 d)}$ be as given in Equation~\eqref{eq:def:Qinst},
	\begin{equation*}
		\left\| \HpeerOracle - \HlatOracle \right\|
		\le \left\| \Hcheck \Qinsttop- \HpeerOracle \right\|
		+ \left\| \Hcheck \Qinsttop - \HlatOracle \right\|,
	\end{equation*}
	and Lemmas~\ref{lem:HcheckHpeer:spectral} and~\ref{lem:HcheckHlat:spectral} complete the proof.
\end{proof}

\begin{lemma} \label{lem:HhatHlat:innerprod}
	Let $\vu \in \R^n$ be independent of $\A-\Apop$ given $\Xpop$ and suppose that Assumptions~\ref{assum:Apop:spectrum} through~\ref{assum:stronger4dhat} hold.
	Then under either the peer contagion model in Equation~\eqref{eq:lim-peer-red} with Assumption~\ref{assum:peer-oracle} or the latent contagion model in Equation~\eqref{eq:lim-latent-red} with Assumption~\ref{assum:latent-oracle},
        with $\Qinst \in \R^{(3 p + 3 d) \times (3 p + 3 d)}$ as given in Equation~\eqref{eq:def:Qinst},
	\begin{equation*}
		\left\| \vu^\top \left(\Hhat \Qinsttop - \HlatOracle\right) \right\|
		= \op{ \| \vu \| } .
	\end{equation*}
\end{lemma} %
\begin{proof}
	Recalling the definitions of $\Hhat$ and $\HlatOracle$ and applying the triangle inequality,
	\begin{equation} \label{eq:HhatHlat:inner:triangle} \begin{aligned}
	\left\| \vu^\top \left( \Hhat \Qinsttop - \HlatOracle \right) \right\|
	 & \le \left\| \vu^\top \left( \Xhat \Q^\top- \Xpop \right) \right\|
	+ \left\| \vu^\top \left( \Ghat - \Gtilde \right) \W \right\|
	+ \left\| \vu^\top \left(\Ghat \Xhat \Q^\top - \Gtilde \Xpop\right) \right\| \\
	 & ~~~~~~~~~~~~~~~~~~~+
	\left\| \vu^\top \left( \Ghat^2 - \Gtilde^2 \right) \W \right\|
	+ \left\|\vu^\top \left(\Ghat^2 \Xhat \Q^\top -\Gtilde^2 \Xpop\right) \right\| .
\end{aligned} \end{equation}
	Applying Lemma~\ref{lem:XhatXTB:ctrl} with $\mB = \vu$, followed by our growth assumptions in Equations~\eqref{eq:assum:spectralConc},~\eqref{eq:assum:XandW:spectrumUB} and~\eqref{eq:assum:XhatX:rate:2},
	\begin{equation} \label{eq:HhatHlat:inner:XhatX}
		\left\| \vu^\top \left( \Xhat \Q^\top - \Xpop \right) \right\|
		\le
		\frac{ C \| \vu \| \sqrt{ \nu + b^2 } \log n }{ \sqrt{\spop_d} }
		+ \frac{ C \kappa(\Apop) \| \vu \| (\nu+b^2) n \log^2 n }{ \spop_d^{3/2} }
		= \op{ \| \vu \| } .
	\end{equation}

	Applying the SVD to $\W$ and using Lemma~\ref{lem:GhatGtilde:quad},
	\begin{equation} \label{eq:HhatHlat:vec:term2}
		\left\| \vu^\top \left( \Ghat - \Gtilde \right) \W \right\|
		= \op{ \frac{ \| \W \| \| \vu \| }{ \sqrt{n} } }
		= \op{ \| \vu \| },
	\end{equation}
	where the second equality follows from the growth assumption in Equation~\eqref{eq:assum:XandW:spectrumUB}.

	Similarly, by the triangle inequal)ity,
	\begin{equation*} \begin{aligned}
			\left\| \vu^\top \! \left( \Ghat^2 \!-\! \Gtilde^2 \right) \W \right\|
			 & \le
			\left\| \vu^\top \Ghat \left( \Ghat \!-\! \Gtilde \right) \W \right\|
			+ \left\| \vu^\top \!\left( \Ghat \!-\! \Gtilde \right) \Gtilde \W \right\| \\
			 & \le
			\left\| \vu^\top \Gtilde \left( \Ghat \!-\! \Gtilde \right) \W \right\|
			+ \left\| \vu^\top \! \left( \Ghat \!-\! \Gtilde \right) \left( \Ghat \!-\! \Gtilde \right) \W \right\|
			+ \left\| \vu^\top \! \left( \Ghat \!-\! \Gtilde \right) \Gtilde \W \right\| .
		\end{aligned} \end{equation*}
	Applying Lemma~\ref{lem:GhatGtilde:quad} to the first and third terms, recalling that $\|\Gtilde\|=1$,
	\begin{equation*}
		\left\| \vu^\top \left( \Ghat^2 - \Gtilde^2 \right) \W \right\|
		\le
		\left\| \vu^\top \left( \Ghat - \Gtilde \right) \left( \Ghat - \Gtilde \right) \W \right\|
		+ \op{ \frac{ \| \W \| \| \vu \| }{ \sqrt{n} } }.
	\end{equation*}
	By submultiplicativity,
	\begin{equation*} \begin{aligned}
			\left\| \vu^\top \left( \Ghat^2 - \Gtilde^2 \right) \W \right\|
			 & \le \| \vu \| \| \W \| \left\| \Ghat - \Gtilde \right\|^2
			+ \op{ \frac{ \| \W \| \| \vu \| }{ \sqrt{n} } }             \\
			 & \le
			C\| \vu \| \| \W \|
			\left( 1 + \frac{ \spop_1 }{ \min_{i\in [n]} \dtilde_i } \right)^2
			\frac{ (\nu+b^2) n \log^2 n }{ \min_{i \in [n]} \dtilde_i^2 }
			+ \op{ \frac{ \| \W \| \| \vu \| }{ \sqrt{n} } } ,
		\end{aligned} \end{equation*}
	where the second inequality follows from Lemma~\ref{lem:GhatGtilde:spectral}.
	Applying our growth assumptions in Equations~\eqref{eq:assum:XandW:spectrumUB} and~\eqref{eq:assum:Grate:stronger},
	\begin{equation} \label{eq:HhatHlat:vec:term4}
		\left\| \vu^\top \left( \Ghat^2 - \Gtilde^2 \right) \W \right\|
		= \op{ \| \vu \| } .
	\end{equation}

	Again using the triangle inequality and submultiplicativy of the norm,
	\begin{equation*}
	\left\| \vu^\top \left( \Ghat \Xhat \Q^\top- \Gtilde \Xpop \right) \right\|
	\le
	\left\| \vu^\top \left( \Ghat - \Gtilde \right) \Xpop \right\|
	+ \left\| \vu^\top \left( \Ghat - \Gtilde \right) \left( \Xhat \Q^\top - \Xpop \right) \right\|
	+ \left\| \vu^\top \Gtilde \left(\Xhat \Q^\top-\Xpop\right) \right\| .
	\end{equation*}
	Applying Lemma~\ref{lem:GhatGtilde:quad} and recalling that $\|\Xpop\|=\Op{\sqrt{n}}$ by Equation~\eqref{eq:assum:XandW:spectrumUB},
	\begin{equation*}
	\left\| \vu^\top \left( \Ghat \Xhat \Q^\top - \Gtilde \Xpop \right) \right\|
	\le
	\left\| \vu^\top \left( \Ghat - \Gtilde \right) \left( \Xhat \Q^\top - \Xpop \right) \right\|
	+ \left\| \vu^\top \Gtilde \left( \Xhat \Q^\top- \Xpop \right) \right\|
	+ \op{ \| \vu \| } .
	\end{equation*}
	By Lemma~\ref{lem:XhatXTB:ctrl} with $\mB = \Gtilde^\top \vu$ and our growth assumptions in Equations~\eqref{eq:assum:spectralConc},~\eqref{eq:assum:XandW:spectrumUB} and~\eqref{eq:assum:XhatX:rate:2} and recalling that $\|\Gtilde\|=1$,
	\begin{equation} \label{eq:GhatXhatGtildeXpop:prelim}
	\left\| \vu^\top \left( \Ghat \Xhat \Q^\top - \Gtilde \Xpop \right) \right\|
	\le
	\left\| \vu^\top \left( \Ghat - \Gtilde \right) \left( \Xhat \Q^\top - \Xpop \right) \right\|
	+ \op{ \frac{ \| \Xpop \| \| \vu \| }{ \sqrt{n} } } .
	\end{equation}
	By submultiplicativity and Lemma~\ref{lem:GhatGtilde:spectral},
	\begin{equation*}
	\left\| \vu^\top \left( \Ghat - \Gtilde \right) \left( \Xhat \Q^\top - \Xpop \right) \right\|
	\le
	C \| \vu \|
	\left( 1 + \frac{ \spop_1 }{ \min_{i\in [n]} \dtilde_i } \right)
	\frac{ \sqrt{ \nu + b^2 } \sqrt{n} \log n }
	{ \min_{i \in [n]} \dtilde_i }
	\left\| \Xhat \Q^\top - \Xpop \right\| .
	\end{equation*}

	Applying Lemma~\ref{lem:XhatXTB:ctrl} with $\mB = \mI$ and collecting terms,
	\begin{equation*} %
	\begin{aligned}
	 & \left\| \vu^\top \left( \Ghat - \Gtilde \right) \left( \Xhat \Q^\top - \Xpop \right) \right\| \\
	 & ~~~\le C \| \vu \|
	\left( 1 + \frac{ \spop_1 }{ \min_{i\in [n]} \dtilde_i } \right)
	\left[
		\frac{ (\nu+b^2) n \log^2 n }{ \sqrt{\spop_d} \min_{i \in [n]} \dtilde_i }
		+
		\frac{ \kappa(\Apop) (\nu+b^2)^{3/2} n^{3/2} \log^3 n }
		{ \spop_d^{3/2} \min_{i \in [n]} \dtilde_i }
	\right]                                                                                  \\
	 & ~~~\le C \| \vu \|
	\frac{ \spop_1 }{ \min_{i\in [n]} \dtilde_i }
	\left[
		\frac{ (\nu+b^2) n \log^2 n }{ \sqrt{\spop_d} \min_{i \in [n]} \dtilde_i }
		+
		\frac{ \kappa(\Apop) (\nu+b^2)^{3/2} n^{3/2} \log^3 n }
		{ \spop_d^{3/2} \min_{i \in [n]} \dtilde_i }
		\right] ,
	\end{aligned} \end{equation*}
	where the second inequality follows from the fact that $\spop_1$ is an upper bound on the minimum degree. %
	Applying our growth assumptions in Equations~\eqref{eq:deglb:2}, ~\eqref{eq:assum:XandW:spectrumUB}, ~\eqref{eq:deglb:2} and~\eqref{eq:deglb:2:kappa},
	\begin{equation*}
		\left\| \vu^\top \left( \Ghat - \Gtilde \right) \left( \Xhat \Q^\top - \Xpop \right) \right\|
		= \op{ \| \vu \| } .
	\end{equation*}
	Applying this to Equation~\eqref{eq:GhatXhatGtildeXpop:prelim},
	\begin{equation} \label{eq:HhatHlat:vec:term3}
		\left\| \vu^\top \left( \Ghat \Xhat \Q^\top - \Gtilde \Xpop \right) \right\|
		= \op{ \| \vu \| } .
	\end{equation}

	By an argument parallel to that yielding Equation~\eqref{eq:HhatHlat:vec:term4}, this time using the growth assumption in Equation~\eqref{eq:assum:XandW:spectrumUB},
	\begin{equation} \label{eq:HhatHlat:vec:term5}
		\left\| \vu^\top \left( \Ghat^2 - \Gtilde^2 \right) \Xpop \right\|
		= \op{ \| \vu \| } .
	\end{equation}

	Applying Equations~\eqref{eq:HhatHlat:inner:XhatX},~\eqref{eq:HhatHlat:vec:term2}, ~\eqref{eq:HhatHlat:vec:term4}, ~\eqref{eq:HhatHlat:vec:term3} and~\eqref{eq:HhatHlat:vec:term5} to Equation~\eqref{eq:HhatHlat:inner:triangle} completes the proof.
\end{proof}

\begin{lemma} \label{lem:HcheckHlat:innerprod}
	Suppose that Assumptions~\ref{assum:Apop:spectrum} through~\ref{assum:latentpositions} hold.
	Let $\Hcheck$ and $\HlatOracle$ be as given in Definition~\ref{def:feasibleEstimators:peer} and Equation~\eqref{eq:def:HlatOracle}, respectively,
	and let $\vu \in \R^n$ be independent of $\A-\Apop$ given $\Xpop$.
	Then under either the peer contagion model in Equation~\eqref{eq:lim-peer-red} with Assumption~\ref{assum:peer-oracle} or the latent contagion model in Equation~\eqref{eq:lim-latent-red} with Assumption~\ref{assum:latent-oracle},
        with $\Qinst \in \R^{(3 p + 3 d) \times (3 p + 3 d)}$ as given in Equation~\eqref{eq:def:Qinst},
	\begin{equation*}
	\left\| \vu^\top \left(\Hcheck \Qinsttop - \HlatOracle\right) \right\|
	= \op{ \| \vu \| } .
	\end{equation*}
\end{lemma} %
\begin{proof}
	Recalling the definitions of $\Hcheck$ and $\HlatOracle$ and applying the triangle inequality,
	\begin{equation*} %
	\begin{aligned}
	\left\| \vu^\top \left( \Hcheck \Qinsttop - \HlatOracle \right) \right\|
	 & \le \left\| \vu^\top \left( \Xhat \Q^\top - \Xpop \right) \right\|
	+ \left\| \vu^\top \left( \G - \Gtilde \right) \W \right\|
	+ \left\| \vu^\top \left( \G \Xhat \Q^\top - \Gtilde \Xpop \right) 
		\right\| \\
	 & ~~~~~~~~~~~~~~~~~~~+ \left\| \vu^\top \left( \G^2 - \Gtilde^2 \right) \W \right\|
	+ \left\| \vu^\top \left( \G^2 \Xhat \Q^\top - \Gtilde^2 \Xpop \right) \right\| .
	\end{aligned} \end{equation*}
	The result then follows from an argument analogous to that in the proof of Lemma~\ref{lem:HhatHlat:innerprod}, using
	Lemma~\ref{lem:GGtilde:spectral} in place of Lemma~\ref{lem:GhatGtilde:spectral}
	and
	Lemma~\ref{lem:GGtilde:quad} in place of Lemma~\ref{lem:GhatGtilde:quad}.
\end{proof}

\begin{lemma} \label{lem:HcheckHpeer:innerprod}
	Suppose that Assumptions~\ref{assum:Apop:spectrum} through~\ref{assum:latentpositions} hold.
	Let $\Hcheck$ and $\HpeerOracle$ be as given in Definition~\ref{def:feasibleEstimators:peer} and Equation~\eqref{eq:def:HpeerOracle}
	and let $\vu \in \R^n$ be independent of $\A-\Apop$ given $\Xpop$.
	Then under either the peer contagion model in Equation~\eqref{eq:lim-peer-red} with Assumption~\ref{assum:peer-oracle} or the latent contagion model in Equation~\eqref{eq:lim-latent-red} with Assumption~\ref{assum:latent-oracle},
        with $\Qinst \in \R^{(3 p + 3 d) \times (3 p + 3 d)}$ as given in Equation~\eqref{eq:def:Qinst},
	\begin{equation*}
		\left\| \vu^\top \left(\Hcheck \Qinsttop - \HpeerOracle\right) \right\|
		= \op{ \| \vu \| } .
	\end{equation*}
\end{lemma}
\begin{proof}
	Recalling the definitions of $\Hcheck$ and $\HpeerOracle$ from Definition~\ref{def:feasibleEstimators:peer} and Equation~\eqref{eq:def:HpeerOracle}, respectively, and applying the triangle inequality,
	\begin{equation*}
		\left\| \vu^\top \left( \Hcheck \Qinsttop - \HpeerOracle \right) \right\|
		\le \left\| \vu^\top \left( \Xhat \Q^\top - \Xpop \right) \right\|
		+ \left\| \vu^\top \G \left( \Xhat \Q^\top- \Xpop \right) \right\|
		+ \left\| \vu^\top \G^2 \left( \Xhat \Q^\top - \Xpop \right) \right\| .
	\end{equation*}
	Three applications of Lemma~\ref{lem:XhatXTB:ctrl} with $\mB = \vu$, $\mB = \G^\top \vu$ and $\mB = (\G^2)^\top \vu$ yields, after recalling that
	$\| (\G^2)^\top \vu \| \le  \| \G^\top \vu \| \le \| \vu \|$ since $\| \G \| = 1$,
	\begin{equation*}
		\left\| \vu^\top \left( \Hcheck \Qinsttop- \HpeerOracle \right) \right\|
		\le
		\frac{ C \| \vu \| \sqrt{ \nu + b^2 } \log n }{ \sqrt{\spop_d} }
		+
		\frac{ C \kappa(\Apop) \| \vu \| (\nu+b^2) n \log^2 n }{ \spop_d^{3/2} } .
	\end{equation*}
	Applying the growth assumptions in Equations~\eqref{eq:assum:spectralConc},~\eqref{eq:assum:XandW:spectrumUB} and~\eqref{eq:assum:XhatX:rate:2} completes the proof.
\end{proof}

\begin{lemma} \label{lem:HpeerHlat:innerprod}
	Let $\vu \in \R^n$ be independent of $\A-\Apop$ given $\Xpop$ and suppose that Assumptions~\ref{assum:Apop:spectrum} through~\ref{assum:latentpositions} hold.
	Then under either the peer contagion model in Equation~\eqref{eq:lim-peer-red} with Assumption~\ref{assum:peer-oracle} or the latent contagion model in Equation~\eqref{eq:lim-latent-red} with Assumption~\ref{assum:latent-oracle},
	\begin{equation*}
		\left\| \vu^\top \left( \HpeerOracle - \HlatOracle \right) \right\|
		= \op{ \sqrt{ \| \vu \| } } .
	\end{equation*}
\end{lemma}
\begin{proof} 
Letting $\Qinst \in \R^{(3 p + 3 d) \times (3 p + 3 d)}$ be as in Equation~\eqref{eq:def:Qinst}, the triangle inequality yields
	\begin{equation*}
		\left\| \vu^\top\left( \HpeerOracle - \HlatOracle \right) \right\|
		\le \left\| \vu^\top\left( \Hcheck \Qinsttop - \HpeerOracle \right) \right\|
		+ \left\| \vu^\top \left( \Hcheck \Qinsttop - \HlatOracle \right) \right\|,
	\end{equation*}
	and Lemmas~\ref{lem:HcheckHpeer:innerprod} and~\ref{lem:HcheckHlat:innerprod} complete the proof.
\end{proof}

\section{Controlling Projections} \label{apx:projections}

Here we collect results controlling the projection matrices used in our two-stage least squares estimators.
Our primary task is to relate the projection $\Mhat$, from Definition~\ref{def:feasibleEstimators:latent} and the projection $\Mcheck$, from Definition~\ref{def:feasibleEstimators:latent}, to their oracle versions,
\begin{equation} \label{eq:def:MlatOracle}
	\MlatOracle = \HlatOracle \left( \HlatOracle^\top \HlatOracle \right) \HlatOracle^\top,
\end{equation}
and
\begin{equation} \label{eq:def:MpeerOracle}
	\MpeerOracle
	= \HpeerOracle (\HpeerOracle^\top \HpeerOracle)^{-1} \HpeerOracle^\top ,
\end{equation}
where $\HlatOracle \in \R^{n \times (3p+3d)}$ is as defined in Equation~\eqref{eq:def:HlatOracle} and $\HpeerOracle$ is as defined in Equation~\eqref{eq:def:HpeerOracle}.

\begin{lemma} \label{lem:MhatMlat:spectral}
	Suppose that Assumptions~\ref{assum:Apop:spectrum},~\ref{assum:degrees},~\ref{assum:latentpositions} and~\ref{assum:stronger4dhat} hold.
	With $\Mhat$ as given in Definition~\ref{def:feasibleEstimators:latent} and $\MlatOracle$ as defined in Equation~\eqref{eq:def:MlatOracle},
	under either the latent contagion model in Equation~\eqref{eq:lim-latent-red} with Assumption~\ref{assum:latent-oracle} or the peer contagion model in Equation~\eqref{eq:lim-peer-red} with Assumption~\ref{assum:peer-oracle},
	\begin{equation*}
		\left\| \Mhat - \MlatOracle \right\|
		= \op{ \frac{ 1 }{\sigmamin(\HlatOracle)} } .
	\end{equation*}
\end{lemma}
\begin{proof}
Let and $\symbf{U}_{\Hhat}$ and $\symbf{U}_{\HlatOracle}$ be the left singular vectors of $\Hhat$ and $\HpeerOracle$, respectively. Since $\Mhat$ and $\MlatOracle$ are projection matrices, we may write
	\begin{align*}
	\Mhat = \symbf{U}_{\Hhat} \symbf{U}_{\Hhat}^\top
	~~~\text{ and }~~~
	\MlatOracle = \symbf{U}_{\HlatOracle} \symbf{U}_{\HlatOracle}^\top,
	\end{align*}
	where $\symbf{U}_{\Hhat}$ and $\symbf{U}_{\HlatOracle}$ are the eigenvectors of $\Hhat \Hhat^\top$  and $\HlatOracle \HlatOracle^\top$, respectively.
	By the Davis-Kahan $\sin \Theta$ theorem \citep{bhatia1997,yu2015},
	\begin{equation} \label{eq:MhatMlat:DK}
		\left\| \Mhat - \MlatOracle \right\|
		= \norm*{\symbf{U}_{\Hhat} \symbf{U}_{\Hhat}^\top - \symbf{U}_{\HlatOracle} \symbf{U}_{\HlatOracle}^\top}
		\le \frac{ C \left\| \Hhat \Hhat^\top - \HlatOracle \HlatOracle^\top\right\| }
		{ \sigmamin^2\left( \HlatOracle \right) } ,
	\end{equation}
	where $\HlatOracle$ is as defined in Equation~\eqref{eq:def:HlatOracle} and $\Hhat$ is as in Definition~\ref{def:feasibleEstimators:latent}.
	
	Let $\Qinst \in \R^{3(p+d)\times 3(p+d)}$ be as in Equation~\eqref{eq:def:Qinst}, and recall that $\Qinst$ is orthogonal by construction.
	By the triangle inequality,
	\begin{equation*}
	\left\| \Hhat \Hhat^\top - \HlatOracle \HlatOracle^\top\right\|
	= \left\| \Hhat \Qinsttop \Qinst \Hhat^\top - \HlatOracle \HlatOracle^\top\right\|
	\le 2 \left\| \left( \Hhat \Qinsttop - \HlatOracle \right) \HlatOracle^\top \right\|
	+ \left\| \Hhat \Qinsttop - \HlatOracle \right\|^2.
	\end{equation*}
	Controlling the first term with Lemma~\ref{lem:HhatHlat:innerprod} and the second term with Lemma~\ref{lem:HhatHlat:spectral},
	\begin{equation*}
		\left\| \Hhat \Hhat^\top - \HlatOracle \HlatOracle^\top\right\|
		\le \op{ \| \HlatOracle \| }
		+ \op{ \sigmamin(\HlatOracle) }
		= \op{ \| \HlatOracle \| },
	\end{equation*}
	where we have used the growth rate in Equation~\eqref{eq:assum:XandW:spectrumUB}.
	Applying this to Equation~\eqref{eq:MhatMlat:DK},
	\begin{equation*}
		\left\| \Mhat - \MlatOracle \right\|
		= \op{ \frac{ \kappa( \HlatOracle ) }{\sigmamin(\HlatOracle)} } .
	\end{equation*}
	Our growth assumption in Equation~\eqref{eq:assum:Hgrowth} completes the proof.
\end{proof}

\begin{lemma} \label{lem:McheckMlat:spectral}
	Suppose that Assumptions~\ref{assum:Apop:spectrum},~\ref{assum:degrees} and~\ref{assum:latentpositions} hold.
	With $\Mcheck$ as given in Definition~\ref{def:feasibleEstimators:peer} and $\MlatOracle$ as defined in Equation~\eqref{eq:def:MlatOracle},
	under either the latent contagion model in Equation~\eqref{eq:lim-latent-red} with Assumption~\ref{assum:latent-oracle} or the peer contagion model in Equation~\eqref{eq:lim-peer-red} with Assumption~\ref{assum:peer-oracle},
	\begin{equation*}
		\left\| \Mcheck - \MlatOracle \right\|
		= \op{ \frac{1}{\sigmamin(\HlatOracle)} }.
	\end{equation*}
\end{lemma}
\begin{proof}
	The proof follows by the same argument as in the proof of Lemma~\ref{lem:MhatMlat:spectral}, using Lemmas~\ref{lem:HcheckHlat:innerprod} and~\ref{lem:HcheckHlat:spectral} in place of, respectively, Lemmas~\ref{lem:HhatHlat:innerprod} and~\ref{lem:HhatHlat:spectral}.
	Details are omitted.
\end{proof}

\begin{lemma} \label{lem:McheckMpeer:spectral}
	Suppose that Assumptions~\ref{assum:Apop:spectrum},~\ref{assum:degrees} and~\ref{assum:latentpositions} hold.
	With $\Mcheck$ as given in Definition~\ref{def:feasibleEstimators:peer} and $\MpeerOracle$ as defined in Equation~\eqref{eq:def:MpeerOracle},
	under either the latent contagion model in Equation~\eqref{eq:lim-latent-red} with Assumption~\ref{assum:latent-oracle} or the peer contagion model in Equation~\eqref{eq:lim-peer-red} with Assumption~\ref{assum:peer-oracle},
	\begin{equation*}
		\left\| \Mcheck - \MpeerOracle \right\|
		= \op{ \frac{1}{\sigmamin(\HlatOracle)} }.
	\end{equation*}
\end{lemma}
\begin{proof}
	The proof follows from the same argument as Lemma~\ref{lem:MhatMlat:spectral}, using Lemmas~\ref{lem:HcheckHpeer:innerprod} and~\ref{lem:HcheckHpeer:spectral} in place of, respectively, Lemmas~\ref{lem:HhatHlat:innerprod} and~\ref{lem:HhatHlat:spectral}.
	Details are omitted.
\end{proof}

\begin{lemma} \label{lem:MpeerMlat:spectral}
	Suppose that Assumptions~\ref{assum:Apop:spectrum},~\ref{assum:degrees} and~\ref{assum:latentpositions} hold.
	With $\MpeerOracle$ as defined in Equation~\eqref{eq:def:MpeerOracle}
	and $\MlatOracle$ as defined in Equation~\eqref{eq:def:MlatOracle},
	under either the latent contagion model in Equation~\eqref{eq:lim-latent-red} with Assumption~\ref{assum:latent-oracle} or the peer contagion model in Equation~\eqref{eq:lim-peer-red} with Assumption~\ref{assum:peer-oracle},
	\begin{equation*}
		\left\| \MpeerOracle - \MlatOracle \right\|
		= \op{ \frac{1}{\sigmamin(\HlatOracle)} }.
	\end{equation*}
\end{lemma}
\begin{proof}
	The proof follows from the same argument as Lemma~\ref{lem:MhatMlat:spectral}, using Lemmas~\ref{lem:HpeerHlat:innerprod} and~\ref{lem:HpeerHlat:spectral} in place of, respectively, Lemmas~\ref{lem:HhatHlat:innerprod} and~\ref{lem:HhatHlat:spectral}.
	Details are omitted.
\end{proof}

\begin{lemma} \label{lem:MhatMpeer:spectral}
	Suppose that Assumptions~\ref{assum:Apop:spectrum},~\ref{assum:degrees},~\ref{assum:latentpositions} and~\ref{assum:stronger4dhat} hold.
	With $\Mhat$ as given by Definition~\ref{def:feasibleEstimators:peer}
	and $\MpeerOracle$ as defined in Equation~\eqref{eq:def:MpeerOracle},
	under either the latent contagion model in Equation~\eqref{eq:lim-latent-red} with Assumption~\ref{assum:latent-oracle} or the peer contagion model in Equation~\eqref{eq:lim-peer-red} with Assumption~\ref{assum:peer-oracle},
	\begin{equation*}
		\left\| \Mhat - \MpeerOracle \right\|
		= \op{ \frac{1}{\sigmamin(\HlatOracle)} }.
	\end{equation*}
\end{lemma}
\begin{proof}
	Applying the triangle inequality,
	\begin{equation*}
		\left\| \Mhat - \MpeerOracle \right\|
		\le \left\| \Mhat - \MlatOracle \right\|
		+ \left\| \MpeerOracle - \MlatOracle \right\|,
	\end{equation*}
	and Lemmas~\ref{lem:MhatMlat:spectral} and~\ref{lem:MpeerMlat:spectral} complete the proof.
\end{proof}

\section{Controlling the Design Matrices} \label{apx:design}

The following results are aimed toward controlling the behavior of our design matrices $\Zhat, \Zcheck, \ZpeerOracle$ and $\ZlatOracle$,
as well as their interactions with the projections onto the span of the instrument (i.e., $\Mhat,\Mcheck,\MpeerOracle$ and $\MlatOracle$).
As in our results on latent position estimation in Appendix~\ref{apx:XhatX}, we must account for the rotational non-identifiability of the latent positions.
This is achieved via an orthogonal matrix $\Qdes \in \R^{(2+p+d) \times (2+p+d)}$, defined according to
\begin{equation} \label{eq:def:Qdes}
\Qdes = 
\begin{bmatrix} 
1                     &\zerovec_{1\times p} &\zerovec_{1\times d} & 0 \\
\zerovec_{p\times 1}  &\mI_{p \times p}      &\zerovec_{p\times d} &\zerovec \\
\zerovec_{d \times 1} &\zerovec_{d\times p} &\Q 		&\zerovec_{d\times 1} \\
0                     &\zerovec_{1\times p} &\zerovec_{1\times d} & 1
\end{bmatrix} ,
\end{equation}
where $\Q \in \R^{d \times d}$ is the orthogonal matrix guaranteed by Lemma~\ref{lem:XhatXTB:ctrl}.

\begin{lemma} \label{lem:ZhatZlat:spectral}
Suppose that Assumptions~\ref{assum:Apop:spectrum},~\ref{assum:degrees},~\ref{assum:latentpositions} and~\ref{assum:stronger4dhat} hold, 
        Then under either the latent contagion model in Equation~\eqref{eq:lim-latent-red} with Assumption~\ref{assum:latent-oracle} or the peer contagion model in Equation~\eqref{eq:lim-peer-red} with Assumption~\ref{assum:peer-oracle},
	with $\Qdes \in \R^{(2+p+d) \times (2+p+d)}$ as given by Equation~\eqref{eq:def:Qdes},
	\begin{equation*}
		\left\| \Zhat \Qdestop - \ZlatOracle \right\| = \op{ \sigmamin( \MlatOracle \ZlatOracle ) } .
	\end{equation*} %
\end{lemma}
\begin{proof}
	Recalling the definitions of $\Zhat$ and $\ZlatOracle$ from Definition~\eqref{def:feasibleEstimators:latent} and Equation~\eqref{eq:def:ZlatOracle}, respectively, the triangle inequality and construction of $\Qdes$ imply
	\begin{equation} \label{eq:ZhatZlat:triangle}
		\left\| \Zhat \Qdestop - \ZlatOracle \right\|
		\le \left\| \Xhat \Q^\top - \Xpop \right\|
		+ \left\| \left( \Ghat - \Gtilde \right) \Y \right\| .
	\end{equation}
	Applying Lemma~\ref{lem:XhatXTB:ctrl} with $\mB = \mI$,
	\begin{equation*}
		\left\| \Xhat \Q^\top - \Xpop \right\|
		\le
		\frac{ C \sqrt{n} \sqrt{ \nu + b^2 } \log n }{ \sqrt{\spop_d} }
		+
		\frac{ C \kappa(\Apop) (\nu+b^2) n \log^2 n }{ \spop_d^{3/2} } ,
	\end{equation*}
	and our growth assumptions in Equations~\eqref{eq:assum:spectralConc}~\eqref{eq:assum:XandW:spectrumUB} and~\eqref{eq:assum:XhatX:rate:2} imply
	\begin{equation} \label{eq:ZhatZlat:spectral:Xterm}
		\left\| \Xhat \Q^\top - \Xpop \right\|
		= \op{ \sqrt{n} }
		= \op{ \sigmamin( \MlatOracle \ZlatOracle ) } ,
	\end{equation}
	where the second equality follows from our assumption in Equation~\eqref{eq:assum:Zgrowth}
	z
	Submultiplicativity and Lemma~\ref{lem:GhatGtilde:spectral} yield
	\begin{equation*}
		\left\| ( \Ghat-\Gtilde) \Y \right\|
		\le
		C \left( 1 + \frac{ \spop_1 }{ \min_{i\in [n]} \dtilde_i } \right)
		\frac{ \| \Y \| \sqrt{ \nu + b^2 } \sqrt{n} \log n }{ \min_{i \in [n]} \dtilde_i }
		= \op{ \| \Y \| },
	\end{equation*}
	where the equality follows from our growth assumption in Equation~\eqref{eq:assum:Grate:stronger}.
	Using Lemma~\ref{lem:Y:latcon} or Lemma~\ref{lem:Y:peercon} to control $\| \Y \| = \Op{\sqrt{n}}$ depending on whether latent contagion or peer contagion, respectively, is assumed,
	\begin{equation*}
		\left\| ( \Ghat-\Gtilde) \Y \right\|
		= \op{ \sqrt{n} }
		= \op{ \sigmamin( \MlatOracle \ZlatOracle ) },
	\end{equation*}
	where the second equality follows from our growth assumption in Equation~\eqref{eq:assum:Zgrowth}.
	Applying this and Equation~\eqref{eq:ZhatZlat:spectral:Xterm} to Equation~\eqref{eq:ZhatZlat:triangle} completes the proof.
\end{proof}

\begin{lemma} \label{lem:ZcheckZlat:spectral}
	Suppose that Assumptions~\ref{assum:Apop:spectrum},~\ref{assum:degrees} and~\ref{assum:latentpositions} hold.
	Then under either the latent contagion model in Equation~\eqref{eq:lim-latent-red} with Assumption~\ref{assum:latent-oracle} or the peer contagion model in Equation~\eqref{eq:lim-peer-red} with Assumption~\ref{assum:peer-oracle},
	with $\Qdes \in \R^{(2+p+d) \times (2+p+d)}$ as in Equation~\eqref{eq:def:Qdes},
	\begin{equation*}
	\left\| \Zcheck \Qdestop - \ZlatOracle \right\| 
	= \op{ \sigmamin( \MlatOracle \ZlatOracle ) } .
	\end{equation*}
\end{lemma} %
\begin{proof}
	Recalling the definitions of $\Zcheck$ and $\ZlatOracle$ from Definition~\eqref{def:feasibleEstimators:peer} and Equation~\eqref{eq:def:ZlatOracle}, the triangle inequality implies
	\begin{equation*} %
		\left\| \Zcheck \Qdestop - \ZlatOracle \right\|
		\le \left\| \Xhat \Q^\top - \Xpop \right\|
		+ \left\| \left( \G - \Gtilde \right) \Y \right\| .
	\end{equation*}
	The proof then follows by the same argument as in Lemma~\ref{lem:ZhatZlat:spectral}, using Lemma~\ref{lem:GGtilde:spectral} in place of Lemma~\ref{lem:GhatGtilde:spectral}.
\end{proof}

\begin{lemma} \label{lem:ZcheckZpeer:spectral}
	Suppose that Assumptions~\ref{assum:Apop:spectrum},~\ref{assum:degrees} and~\ref{assum:latentpositions} hold.
	Then under either the latent contagion model in Equation~\eqref{eq:lim-latent-red} with Assumption~\ref{assum:latent-oracle} or the peer contagion model in Equation~\eqref{eq:lim-peer-red} with Assumption~\ref{assum:peer-oracle},
        with $\Qdes \in \R^{(2+p+d) \times (2+p+d)}$ as in Equation~\eqref{eq:def:Qdes},
	\begin{equation*}
	\left\| \Zcheck \Qdestop - \ZpeerOracle \right\| = \op{ \sigmamin( \MlatOracle \ZlatOracle ) } .
	\end{equation*}
\end{lemma} %
\begin{proof}
	Recalling the definitions of $\Zcheck$ and $\ZpeerOracle$ from Definition~\eqref{def:feasibleEstimators:peer} and Equation~\eqref{eq:def:ZpeerOracle} and using the definition of $\Qdes$, the triangle inequality implies
	\begin{equation} \label{eq:ZcheckZpeer:triangle}
		\left\| \Zcheck \Qdestop - \ZpeerOracle \right\|
		\le \left\| \Xhat \Q^\top - \Xpop \right\|
		+ \left\| \left( \G - \G \right) \Y \right\|
		= \left\| \Xhat \Q^\top - \Xpop \right\| .
	\end{equation}
	Applying Lemma~\ref{lem:XhatXTB:ctrl} with $\mB = \mI$,
	\begin{equation*}
		\left\| \Zcheck \Qdestop - \ZpeerOracle \right\|
		\le
		\frac{ C \sqrt{n} \sqrt{ \nu + b^2 } \log n }{ \sqrt{\spop_d} }
		+
		\frac{ C \kappa(\Apop) (\nu+b^2) n \log^2 n }{ \spop_d^{3/2} }.
	\end{equation*}
	Applying our growth assumptions in Equations~\eqref{eq:assum:spectralConc},~\eqref{eq:assum:XandW:spectrumUB} and~\eqref{eq:assum:XhatX:rate:2},
	\begin{equation*}
	\left\| \Zcheck \Qdestop - \ZpeerOracle \right\| 
	= \op{ \sqrt{n} },
	\end{equation*}
	and our growth assumption in Equation~\eqref{eq:assum:Zgrowth} completes the proof.
\end{proof}

\begin{lemma} \label{lem:ZpeerZlat:spectral}
	Suppose that Assumptions~\ref{assum:Apop:spectrum},~\ref{assum:degrees} and~\ref{assum:latentpositions} hold.
	Then under either the latent contagion model in Equation~\eqref{eq:lim-latent-red} with Assumption~\ref{assum:latent-oracle} or the peer contagion model in Equation~\eqref{eq:lim-peer-red} with Assumption~\ref{assum:peer-oracle},
	\begin{equation*}
		\left\| \ZpeerOracle - \ZlatOracle \right\| = \op{ \| \MlatOracle \ZlatOracle \| } .
	\end{equation*}
\end{lemma}
\begin{proof}
	Recalling the definitions of $\ZpeerOracle$ and $\ZlatOracle$ from Equations~\eqref{eq:def:ZpeerOracle} and~\eqref{eq:def:ZlatOracle}, respectively, the triangle inequality implies
	\begin{equation*} %
		\left\| \ZpeerOracle - \ZlatOracle \right\|
		\le \left\| \left( \G - \Gtilde \right) \Y \right\| .
	\end{equation*}
	The proof then follows by the same argument as given in the second half of the proof of Lemma~\ref{lem:ZhatZlat:spectral}, using Lemma~\ref{lem:GGtilde:spectral} in place of Lemma~\ref{lem:GhatGtilde:spectral}.
\end{proof}

\begin{lemma} \label{lem:ZhatZpeer:spectral}
	Suppose that Assumptions~\ref{assum:Apop:spectrum},~\ref{assum:degrees},~\ref{assum:latentpositions} and~\ref{assum:stronger4dhat} hold.
	Then under either the latent contagion model in Equation~\eqref{eq:lim-latent-red} with Assumption~\ref{assum:latent-oracle} or the peer contagion model in Equation~\eqref{eq:lim-peer-red} with Assumption~\ref{assum:peer-oracle},
        with $\Qdes \in \R^{(2+p+d) \times (2+p+d)}$ as in Equation~\eqref{eq:def:Qdes},
	\begin{equation*}
	\left\| \Zhat \Qdestop- \ZpeerOracle \right\| 
	= \op{ \| \MlatOracle \ZlatOracle \| } .
	\end{equation*}
\end{lemma}
\begin{proof}
	By the triangle inequality,
	\begin{equation*}
		\left\| \Zhat \Qdestop - \ZpeerOracle \right\|
		\le \left\| \Zhat \Qdestop- \ZlatOracle \right\|
		+ \left\| \ZpeerOracle - \ZlatOracle \right\|,
	\end{equation*}
	and Lemmas~\ref{lem:ZhatZlat:spectral} and~\ref{lem:ZpeerZlat:spectral} complete the proof.
\end{proof}

\begin{lemma} \label{lem:MhatZhat:forWeyl}
	Suppose that Assumptions~\ref{assum:Apop:spectrum},~\ref{assum:degrees},~\ref{assum:latentpositions} and~\ref{assum:stronger4dhat} hold.
	Then under either the latent contagion model in Equation~\eqref{eq:lim-latent-red} with Assumption~\ref{assum:latent-oracle} or the peer contagion model in Equation~\eqref{eq:lim-peer-red} with Assumption~\ref{assum:peer-oracle},
	with $\Qdes \in \R^{(2+p+d) \times (2+p+d)}$ as given by Equation~\eqref{eq:def:Qdes},
	\begin{equation*}
		\left\| \Mhat \Zhat \Qdestop - \MlatOracle \ZlatOracle \right\| = \op{ \sigmamin(\MlatOracle \ZlatOracle) } .
	\end{equation*}
\end{lemma} %
\begin{proof}
	By the triangle inequality,
	\begin{equation} \label{eq:MZ:triangle}
	\left\| \Mhat \Zhat \Qdestop - \MlatOracle \ZlatOracle \right\|
	\le
	\left\| (\Mhat-\MlatOracle) \ZlatOracle \right\|
	+ \left\| \Mhat \left( \Zhat \Qdestop - \ZlatOracle \right) \right\|
	\end{equation}
	Since $\Mhat$ is a projection, submultiplicativity and Lemma~\ref{lem:ZhatZlat:spectral} yields
	\begin{equation*}
	\left\| \Mhat\left( \Zhat\Qdestop - \ZlatOracle \right) \right\|
	= \op{ \sigmamin( \MlatOracle \ZlatOracle ) } .
	\end{equation*}
	Similarly, submultiplicativity and Lemma~\ref{lem:MhatMlat:spectral} yield
	\begin{equation*}
	\left\| (\Mhat-\MlatOracle) \ZlatOracle \right\|
	= \op{ \frac{ \| \ZlatOracle \| }{ \sigmamin(\HlatOracle) } }
	= \op{ \sigmamin( \MlatOracle \ZlatOracle ) },
	\end{equation*}
	where the second equality follows from our growth assumptions in Equations~\eqref{eq:assum:Hgrowth} and~\eqref{eq:assum:Zgrowth}.
	Applying the above two bounds to Equation~\eqref{eq:MZ:triangle} yields the desired result.
\end{proof}

\begin{lemma} \label{lem:McheckZcheck:forWeyl}
	Suppose that Assumptions~\ref{assum:Apop:spectrum},~\ref{assum:degrees} and~\ref{assum:latentpositions} hold.
	Then under either the latent contagion model in Equation~\eqref{eq:lim-latent-red} with Assumption~\ref{assum:latent-oracle} or the peer contagion model in Equation~\eqref{eq:lim-peer-red} with Assumption~\ref{assum:peer-oracle},
	with $\Qdes \in \R^{(2+p+d) \times (2+p+d)}$ as given by Equation~\eqref{eq:def:Qdes},
	\begin{equation*}
		\left\| \Mcheck \Zcheck \Qdestop - \MlatOracle \ZlatOracle \right\| = \op{ \sigmamin(\MlatOracle \ZlatOracle) } .
	\end{equation*}
\end{lemma} %
\begin{proof}
	The proof follows by the same argument as Lemma~\ref{lem:MhatZhat:forWeyl}, using Lemma~\ref{lem:ZcheckZlat:spectral} in place of Lemma~\ref{lem:ZhatZlat:spectral} and Lemma~\ref{lem:McheckMlat:spectral} in place of Lemma~\ref{lem:MhatMlat:spectral}.
\end{proof}

\begin{lemma} \label{lem:MpeerZpeer:forWeyl}
	Suppose that Assumptions~\ref{assum:Apop:spectrum},~\ref{assum:degrees} and~\ref{assum:latentpositions} hold.
	Then under either the latent contagion model in Equation~\eqref{eq:lim-latent-red} with Assumption~\ref{assum:latent-oracle} or the peer contagion model in Equation~\eqref{eq:lim-peer-red} with Assumption~\ref{assum:peer-oracle},
	with $\Qdes \in \R^{(2+p+d) \times (2+p+d)}$ as given by Equation~\eqref{eq:def:Qdes},
	\begin{equation*}
		\left\| \MpeerOracle \ZpeerOracle - \MlatOracle \ZlatOracle \right\|
		= \op{ \sigmamin(\MlatOracle \ZlatOracle) } .
	\end{equation*}
\end{lemma}
\begin{proof}
	The proof follows the same argument as the proof of Lemma~\ref{lem:MhatZhat:forWeyl}, using Lemma~\ref{lem:ZpeerZlat:spectral} in place of Lemma~\ref{lem:ZhatZlat:spectral} and Lemma~\ref{lem:MpeerMlat:spectral} in place of Lemma~\ref{lem:MhatMlat:spectral}.
\end{proof}

\section{Convergence of Oracle Estimators} \label{apx:oracle-convergence}

Our proofs of Theorems~\ref{thm:peertruepeerfit},~\ref{thm:lattruelatfit},~\ref{thm:peertruelatfit} and~\ref{thm:lattruepeerfit} rely on showing that the estimates $\betahattsls$ and $\thetahattsls$ are close to ``oracle'' estimates based on using the true latent positions rather than estimates thereof (see Sections~\ref{apx:latcon} and~\ref{apx:peercon} below). The following theorem implies that these oracle estimates are asymptotically normal about their population targets, so that we need only to show that the four estimates listed above are close to the oracle estimates. This latter argument is carried out in Sections~\ref{apx:latcon} and~\ref{apx:peercon} below.

\begin{theorem}[Theorem 3 of \citealt{kelejian1998}]
    \label{thm:kelejian1998}
    Under the peer contagion model of Equation~\eqref{eq:lim-peer}, suppose that Assumption~\ref{assum:peer-oracle} holds, so that
    \begin{equation*}
        \Y  = \paren*{\mI - \betay \G}^{-1} \paren*{\1_n \betanaught + \W \betaw + \Xpop \betax  + \be}.
    \end{equation*}
    Then
    \begin{equation*}
        \sqrt{n} \paren*{\betatildetsls - \symbf \beta} \to \mathcal{N} \paren*{0, \sigmaeps^2 \paren*{\ZpeerOracle^\top \MpeerOracle \ZpeerOracle}^{-1}} .
    \end{equation*}
    An analogous result holds, replacing $\G, \ZpeerOracle, \MpeerOracle$ and $\betatildetsls$ with $\Gtilde, \ZlatOracle, \MlatOracle$ and  $\thetatildetsls$, respectively, under Assumption~\ref{assum:latent-oracle} instead of Assumption~\ref{assum:peer-oracle}.
\end{theorem}

\section{Latent Contagion Data Model} \label{apx:latcon}

Here we provide proof details of Theorem~\ref{thm:lattruelatfit} and~\ref{thm:lattruepeerfit}, which describe the behavior of our estimators when the responses $\Y$ are driven by contagion on the graph structure encoded by the latent positions $\Xpop$, i.e., what we have termed {\em latent contagion}, as given in Equation~\eqref{eq:lim-latent-red}.
If we had access to this latent structure, we could construct oracle analogues of the estimator,
where $\MlatOracle$ is as defined in Equation~\eqref{eq:def:MlatOracle} and, recalling the definition of $\ZlatOracle$ from Equation~\eqref{eq:def:Zmxs},
\begin{equation} \label{eq:def:ZlatOracle}
  \ZlatOracle = \begin{bmatrix} \1_n \; \W \; \Xpop \; \Gtilde \Y \end{bmatrix}
  \in \R^{n \times (p + d + 2)} .
\end{equation}

We begin by proving convergence of the latent contagion estimators under the latent contagion data model.
We give analogous proofs for the peer contagion estimators in Section~\ref{apx:lattruepeerfit} below.

\subsection{Latent Contagion Estimators}

Our proof of Theorem~\ref{thm:lattruelatfit} relies on showing that when the model in Equation~\eqref{eq:lim-latent-red} holds, our estimator in Equation~\eqref{eq:def:thetahattsls} is suitably close to an ``oracle'' version of the estimator that uses $\ZlatOracle$,
\begin{equation} \label{eq:def:thetaOracleTSLS}
\thetaOracleTSLS = (\ZlatOracle^\top \MlatOracle \ZlatOracle)^{-1}
      	\left( \MlatOracle \ZlatOracle \right)^\top \Y .
\end{equation}
This is established in Lemma~\ref{lem:thetahatTSLS2thetaOracleTSLS} for the TSLS estimator.

\begin{lemma}\label{lem:thetahatTSLS2thetaOracleTSLS}
  Under the latent contagion model in Equation~\eqref{eq:lim-latent-red}, suppose that Assumptions~\ref{assum:Apop:spectrum} through~\ref{assum:stronger4dhat} hold.
  Then there exists a sequence of orthogonal matrices $\Qdes \in \R^{(2+p+d)\times(2+p+d)}$ such that
  \begin{equation*}
    \sqrt{n} \paren*{ \Qdes \thetahattsls - \thetaOracleTSLS } = \op{ 1 } .
  \end{equation*}
\end{lemma}
\begin{proof}
Take $\Qdes$ to be as in Equation~\eqref{eq:def:Qdes}.
  Recalling the definitions in Equations~\eqref{eq:def:thetahattsls} and~\eqref{eq:def:thetaOracleTSLS},
  \begin{equation*}
    \Qdes \thetahattsls - \thetaOracleTSLS
    =
    \left[ \Qdes (\Zhat^\top \Mhat \Zhat)^{-1} \left( \Mhat \Zhat \right)^\top
      - (\ZlatOracle^\top \MlatOracle \ZlatOracle)^{-1}
      \left( \MlatOracle \ZlatOracle \right)^\top \right] \Y .
  \end{equation*}

  Using the fact that $\Qdes$ is orthogonal,
  the triangle inequality yields
  \begin{equation} \label{eq:thetahat2thetaoracle}
  \begin{aligned}
    \left\| \Qdes \thetahattsls - \thetaOracleTSLS \right\|
    & \le
    \left\| \Qdes (\Zhat^\top \Mhat \Zhat)^{-1} \Qdestop
    \left(\Mhat\Zhat\Qdestop-\MlatOracle\ZlatOracle\right)^\top \Y\right\| \\
    & ~~~~~~+
    \left\| \left[ \Qdes (\Zhat^\top \Mhat \Zhat)^{-1} \Qdestop
      - (\ZlatOracle^\top \MlatOracle \ZlatOracle)^{-1} \right]
    \left( \MlatOracle \ZlatOracle \right)^\top \Y \right\| .
  \end{aligned} \end{equation}

  We will control each of these right-hand terms separately.

  \paragraph{Controlling the first term in Equation~\eqref{eq:thetahat2thetaoracle}}
  \enspace

  By submultiplicativity and unitary invariance of the norm,
  \begin{equation*}
    \left\| \Qdes (\Zhat^\top \Mhat \Zhat)^{-1} \Qdestop \!
    \left( \Mhat \Zhat \Qdestop \!-\! \MlatOracle \ZlatOracle \right)^\top 
		\! \Y \right\|
    \le
    \frac{ \left\| \left( \Mhat \Zhat \Qdestop
		\!-\! \MlatOracle \ZlatOracle \right)^{\!\top}\! \Y \right\|}
    {\sigmamin^2( \Mhat \Zhat ) }
    \le
    \frac{ C \left\| \left( \Mhat \Zhat \Qdestop \!-\! \MlatOracle \ZlatOracle 
		\right)^{\!\top} \!\Y \right\|}
    	{\sigmamin^2( \MlatOracle \ZlatOracle ) },
  \end{equation*}
  where the second inequality follows from Weyl's inequality and Lemma~\ref{lem:MhatZhat:forWeyl}.
  Using our assumption in Equation~\eqref{eq:assum:Zgrowth}, it follows that
  \begin{equation} \label{eq:latcon:MZdiffs:Weyl}
    \left\| \Qdes (\Zhat^\top \Mhat \Zhat)^{-1} \Qdestop
    \left( \Mhat \Zhat \Qdestop - \MlatOracle \ZlatOracle \right)^\top \Y
    \right\|
    \le
    \frac{ C \left\| \left( \Mhat \Zhat \Qdestop 
		- \MlatOracle \ZlatOracle \right)^{\!\top} \Y \right\|}
    { n } .
  \end{equation}

  By the triangle inequality,
  \begin{equation} \label{eq:latcon:MZdiffs:numer}
    \left\| \left( \Mhat \Zhat \Qdestop
		- \MlatOracle \ZlatOracle \right)^{\!\top} \Y \right\|
    \le
    \left\| \Qdes \Zhat^\top \left(\Mhat-\MlatOracle\right) \Y \right\|
    + \left\| \left( \Zhat \Qdestop -\ZlatOracle \right)^\top 
		\MlatOracle \Y \right\|.
  \end{equation}
  Again using the triangle inequality,
  \begin{equation} \label{eq:ZhatMdiff:Y:triangle}
    \left\| \Qdes \Zhat^\top \left(\Mhat-\MlatOracle\right) \Y \right\|
    \le \left\| \left(\Zhat\Qdestop-\ZlatOracle\right)^\top 
			\left(\Mhat-\MlatOracle\right) \Y \right\|
    + \left\| \ZlatOracle^\top \left(\Mhat-\MlatOracle\right) \Y \right\|.
  \end{equation}
  By submultiplicativity, Lemma~\ref{lem:MhatMlat:spectral} and the assumptions in Equations~\eqref{eq:assum:Hgrowth} and~\eqref{eq:assum:Zgrowth},
  \begin{equation*}
    \left\| \ZlatOracle^\top \left(\Mhat-\MlatOracle\right) \Y \right\|
    = \op{ \frac{ \| \ZlatOracle \| \| \Y \| }{ \sigmamin(\HlatOracle) } }
    = \op{ \sqrt{n} } ,
  \end{equation*}
  where we have used Lemma~\ref{lem:Y:latcon} to ensure $\|\Y\| = \Op{ \sqrt{n} }$,
  Using submultiplicativity followed by Lemmas~\ref{lem:MhatMlat:spectral} and~\ref{lem:ZhatZlat:spectral},
  \begin{equation*}
    \left\| \left(\Zhat\Qdestop -\ZlatOracle \right)^\top 
		\left(\Mhat-\MlatOracle\right) \Y \right\|
    \le \left\| \Zhat \Qdestop -\ZlatOracle \right\|
    \left\| \Mhat-\MlatOracle \right\| \left\| \Y \right\|
    =
    \op{ \frac{ \sigmamin( \MlatOracle \ZlatOracle ) \| \Y \| }{ \sigmamin(\HlatOracle)} }
    = \op{ \sqrt{n} },
  \end{equation*}
  where the last equality follows from the growth assumption in Equation~\eqref{eq:assum:Hgrowth} and using Lemma~\ref{lem:Y:latcon} to ensure $\|\Y\| = \Op{ \sqrt{n} }$.

  Applying the above two displays to Equation~\eqref{eq:ZhatMdiff:Y:triangle},
  \begin{equation*}
    \left\| \Qdes \Zhat^\top \left(\Mhat-\MlatOracle\right) \Y \right\|
    = \op{ \sqrt{n} }.
  \end{equation*}
  Applying this to Equation~\eqref{eq:latcon:MZdiffs:numer},
  \begin{equation} \label{eq:latcon:MZdiffs:numer:intermezzo}
    \left\| \left( \Mhat \Zhat \Qdestop - \MlatOracle \ZlatOracle \right)^{\!\top} \Y \right\|
    \le
    \left\| \left( \Zhat\Qdestop-\ZlatOracle \right)^\top \MlatOracle \Y \right\|
    + \op{ \sqrt{n} } .
  \end{equation}

  Recalling the structure of $\Zhat$ and $\ZlatOracle$ from Definition~\ref{def:feasibleEstimators:latent} and Equation~\eqref{eq:def:ZlatOracle}, respectively, and applying the triangle inequality,
  \begin{equation} \label{eq:Zdiffs:triangle}
    \left\| (\Zhat\Qdestop - \ZlatOracle )^\top \MlatOracle \Y \right\|
    \le
    \left\| (\Xhat\Q^\top -\Xpop)^\top \MlatOracle \Y \right\|
    + \left| \Y^\top \MlatOracle (\Ghat - \Gtilde) \Y \right| .
  \end{equation}
  Applying Lemma~\ref{lem:XhatXTB:ctrl} with $\mB = \MlatOracle \Y$, noting that $\| \MlatOracle \Y \| \le \| \Y \|$ since $\MlatOracle$ is a projection,
  \begin{equation*}
    \| (\Xhat\Q^\top-\Xpop)^\top \MlatOracle \Y \|
    \le
    C \| \Y \|
    \left(
    \frac{ \sqrt{ \nu + b^2 } \log n }{ \sqrt{\spop_d} }
    +
    \frac{ \kappa(\Apop) (\nu+b^2) n \log^2 n }{ \spop_d^{3/2} }
    \right)
    = \op{ \sqrt{n} },
  \end{equation*}
  where the second equality follows from the growth rates in Equations~\eqref{eq:assum:spectralConc},~\eqref{eq:assum:XandW:spectrumUB} and~\eqref{eq:assum:XhatX:rate:2} and using Lemma~\ref{lem:Y:latcon} to ensure $\|\Y\| = \Op{ \sqrt{n} }$.

  Noting that $\Y$ is independent of $\A-\Apop$ given $\Xpop$ under the latent contagion model,
  Lemma~\ref{lem:GhatGtilde:quad} yields
  \begin{equation*}
    \left| \Y^\top \MlatOracle ( \Ghat-\Gtilde ) \Y \right|
    = \op{ \frac{ \| \Y \|^2 }{ \sqrt{n} } }
    = \op{ \sqrt{n} },
  \end{equation*}
  where we have used the fact that $\| \MlatOracle \Y \| \le \| \Y \|$ trivially and again used Lemma~\ref{lem:Y:latcon} to ensure $\|\Y\| = \Op{ \sqrt{n} }$.
  Applying the above two displays to Equation~\eqref{eq:Zdiffs:triangle},
  \begin{equation*}
    \left\| (\Zhat \Qdestop- \ZlatOracle )^\top \MlatOracle \Y \right\|
    = \op{ \sqrt{n} }.
  \end{equation*}
  Applying this to Equation~\eqref{eq:latcon:MZdiffs:numer:intermezzo},
  and applying the resulting bound to Equation~\eqref{eq:latcon:MZdiffs:Weyl},
  \begin{equation} \label{eq:latcon:MZdiffs:done}
    \left\| \Qdes (\Zhat^\top \Mhat \Zhat)^{-1} \Qdestop
    \left( \Mhat \Zhat \Qdestop - \MlatOracle \ZlatOracle \right)^\top \Y
    \right\|
    = \op{ \frac{1}{\sqrt{n}} } .
  \end{equation}

  \paragraph{Controlling the second term in Equation~\eqref{eq:thetahat2thetaoracle}}
  \enspace

  By Lemma~\ref{lem:MhatZhat:forWeyl} and our growth assumption in Equation~\eqref{eq:assum:Zgrowth}, both $\Zhat^\top \Mhat \Zhat$ and $\ZlatOracle^\top \MlatOracle \ZlatOracle$ are invertible with high probability for all $n$ suitably large.
  Factoring appropriately, applying submultiplicativity and recalling that $\Mhat$ and $\MlatOracle$ are projections,
  \begin{equation*}
    \left\| \left[ \Qdes (\Zhat^\top \Mhat \Zhat)^{-1} \Qdestop
      - (\ZlatOracle^\top \MlatOracle \ZlatOracle)^{-1} \right]
    \left( \MlatOracle \ZlatOracle \right)^\top \Y \right\|
    \le
    \frac{ \left\| \Qdes \Zhat^\top \Mhat \Zhat \Qdestop
      - \ZlatOracle^\top \MlatOracle \ZlatOracle \right\|
      \left\| \ZlatOracle \right\| \left\| \Y \right\| }
    { \sigmamin^2( \Mhat \Zhat ) \sigmamin^2( \MlatOracle \ZlatOracle ) } .
  \end{equation*}
  Using Lemma~\ref{lem:MhatZhat:forWeyl} again along with Weyl's inequality, we can ensure that $\sigmamin( \Mhat \Zhat ) = \Omegap{ \sigmamin( \MlatOracle \ZlatOracle ) }$.
  Our growth assumptions in Equations~\eqref{eq:assum:Zgrowth} and~\eqref{eq:assum:Zgrowth}, as well as using Lemma~\ref{lem:Y:latcon} to ensure $\|\Y\| = \Op{ \sqrt{n} }$, then yield
  \begin{equation} \label{eq:Zcovs:start}
    \left\| \Qdes \left[ (\Zhat^\top \Mhat \Zhat)^{-1} \Qdestop
      - (\ZlatOracle^\top \MlatOracle \ZlatOracle)^{-1} \right]
    \left( \MlatOracle \ZlatOracle \right)^\top \Y \right\|
    \le
    \frac{ C \left\|\Qdes \Zhat^\top \Mhat \Zhat \Qdestop
      - \ZlatOracle^\top \MlatOracle \ZlatOracle \right\| }
    { n } .
  \end{equation}

  Adding and subtracting appropriate quantities and applying the triangle inequality,
  \begin{equation} \label{eq:Zcovs:numertri} \begin{aligned}
    \left\| \Qdes \Zhat^\top \Mhat \Zhat \Qdestop
    - \ZlatOracle^\top \MlatOracle \ZlatOracle \right\|
    &\le
    2\left\| (\Zhat \Qdestop -\ZlatOracle)^\top \MlatOracle \ZlatOracle \right\|
    +
    \left\| \Qdes \Zhat^\top ( \Mhat-\MlatOracle ) \Zhat \Qdestop \right\| \\
    &~~~~~~~~~+
    \left\| (\Zhat\Qdestop-\ZlatOracle)^\top \MlatOracle (\Zhat \Qdestop-\ZlatOracle) \right\| .
  \end{aligned} \end{equation}

  Decomposing $\Zhat \Qdestop-\ZlatOracle$ as in Equation~\eqref{eq:Zdiffs:triangle},
  \begin{equation} \label{eq:Zcovs:linearTerm}
    \left\| (\Zhat \Qdestop-\ZlatOracle)^\top \MlatOracle \ZlatOracle \right\|
    \le
    \left\| (\Xhat \Q^\top-\Xpop)^\top \MlatOracle \ZlatOracle \right\|
    + \left| \ZlatOracle^\top \MlatOracle (\Ghat-\Gtilde) \Y \right| .
  \end{equation}
  Applying Lemma~\ref{lem:XhatXTB:ctrl} with $\mB = \MlatOracle \ZlatOracle$ and using the growth rates in Equations~\eqref{eq:assum:spectralConc},~\eqref{eq:assum:XandW:spectrumUB} and~\eqref{eq:assum:XhatX:rate:2},
  \begin{equation*}
    \left\| (\Xhat \Q^\top-\Xpop)^\top \MlatOracle \ZlatOracle \right\|
    = \op{ \| \MlatOracle \ZlatOracle \| }
    = \op{ \sqrt{n} },
  \end{equation*}
  where the second equality follows from the fact that $\MlatOracle$ is a projection  and the growth assumption in Equation~\eqref{eq:assum:Zgrowth}.
  Using similar growth assumptions, this time with Lemma~\ref{lem:GhatGtilde:quad} and using Lemma~\ref{lem:Y:latcon} to ensure $\|\Y\| = \Op{ \sqrt{n} }$,
  \begin{equation*}
    \left| \ZlatOracle^\top \MlatOracle (\Ghat-\Gtilde) \Y \right|
    = \op{ \frac{ \| \Y \| \| \MlatOracle \ZlatOracle \|  }{ \sqrt{n} } }
    = \op{ \sqrt{n} } .
  \end{equation*}
  Applying the above two displays to Equation~\eqref{eq:Zcovs:linearTerm},
  \begin{equation} \label{eq:Zcovs:linearTerm:done}
    \left\| (\Zhat \Qdestop-\ZlatOracle)^\top \MlatOracle \ZlatOracle \right\|
    = \op{ \sqrt{n} } .
  \end{equation}

  By submultiplicativity,
  \begin{equation*}
    \left\| \Qdes \Zhat^\top ( \Mhat-\MlatOracle ) \Zhat \Qdestop \right\|
    \le
    \left\| \Zhat \Qdestop \right\|^2 \left\| \Mhat-\MlatOracle \right\|
    \le C \left\| \ZlatOracle \right\|^2 \left\| \Mhat-\MlatOracle \right\|,
  \end{equation*}
  where the second inequality follows from Lemma~\ref{lem:ZhatZlat:spectral}, Weyl's inequality, and the fact that $\sigmamin(\MlatOracle \ZlatOracle) \le \sigmamin( \ZlatOracle )$, since $\MlatOracle$ is a projection.
  Applying Lemma~\ref{lem:MhatMlat:spectral} and our growth assumptions in Equations~\eqref{eq:assum:Hgrowth} and~\eqref{eq:assum:Zgrowth},
  \begin{equation} \label{eq:Zcovs:Mdiffs:done}
    \left\| \Qdes \Zhat^\top ( \Mhat-\MlatOracle ) \Zhat \Qdestop \right\|
    = \op{ \sqrt{n} } .
  \end{equation}

  Using idempotence of $\MlatOracle$,
  \begin{equation} \label{eq:Zcovs:quad:start}
    \left\| (\Zhat\Qdestop-\ZlatOracle)^\top \MlatOracle (\Zhat \Qdestop-\ZlatOracle) \right\|
    = \left\| \MlatOracle (\Zhat\Qdestop-\ZlatOracle) \right\|^2 .
  \end{equation}
  Recalling the definition of $\MlatOracle$ from Equation~\eqref{eq:def:MlatOracle} and applying submultiplicativity,
  \begin{equation*}
    \left\| \MlatOracle (\Zhat\Qdestop-\ZlatOracle) \right\|
    \le
    \frac{ \left\| \HlatOracle^\top (\Zhat\Qdestop-\ZlatOracle) \right\| }
    { \sigmamin( \HlatOracle ) } .
  \end{equation*}
  By an argument analogous to those controlling Equations~\eqref{eq:Zdiffs:triangle} and~\eqref{eq:Zcovs:linearTerm} above,
  \begin{equation*}
    \left\| \MlatOracle (\Zhat\Qdestop-\ZlatOracle) \right\|
    = \op{ \frac{ \| \HlatOracle \| }{ \sigmamin( \HlatOracle ) } } .
  \end{equation*}
  Applying this to Equation~\eqref{eq:Zcovs:quad:start},
  \begin{equation*}
    \left\| (\Zhat\Qdestop-\ZlatOracle)^\top \MlatOracle (\Zhat\Qdestop-\ZlatOracle) \right\|
    =
    \op{ \kappa^2( \HlatOracle ) } ,
  \end{equation*}
  and our growth assumption in Equation~\eqref{eq:assum:Hgrowth} yields
  \begin{equation} \label{eq:Zcovs:quad:done}
    \left\| (\Zhat\Qdestop-\ZlatOracle)^\top \MlatOracle (\Zhat\Qdestop-\ZlatOracle) \right\|
    = \op{ \sqrt{n} } .
  \end{equation}

  Applying Equations~\eqref{eq:Zcovs:linearTerm:done},
  ~\eqref{eq:Zcovs:Mdiffs:done}
  and~\eqref{eq:Zcovs:quad:done}
  to Equation~\eqref{eq:Zcovs:numertri},
  \begin{equation*}
    \left\| \Qdes \Zhat^\top \Mhat \Zhat \Qdestop
    - \ZlatOracle^\top \MlatOracle \ZlatOracle \right\|
    = \op{ \sqrt{n} }.
  \end{equation*}
  Applying this to Equation~\eqref{eq:Zcovs:start},
  \begin{equation} \label{eq:Zcovs:hat:done}
    \left\| \left[ \Qdes (\Zhat^\top \Mhat \Zhat)^{-1} \Qdestop
      - (\ZlatOracle^\top \MlatOracle \ZlatOracle)^{-1} \right]
    \left( \MlatOracle \ZlatOracle \right)^\top \Y \right\|
    = \op{ \frac{1}{\sqrt{n}} }.
  \end{equation}
  Applying this and Equation~\eqref{eq:latcon:MZdiffs:done} to Equation~\eqref{eq:thetahat2thetaoracle} and multiplying through by $\sqrt{n}$ completes the proof.
\end{proof}

\begin{proof}[Proof of Theorem~\ref{thm:lattruelatfit}]
  By Theorem~\ref{thm:kelejian1998},
  \begin{equation*} \begin{aligned}
      \sqrt{n} \paren*{\thetaOracleTSLS - \symbf \theta} \to \mathcal{N} \paren*{\0, \Sigma \paren*{\thetahattsls}}
    \end{aligned} \end{equation*}
  Thus, it will suffice for us to show that
  \begin{equation*}
    \sqrt{n} \paren*{ \Qdes \thetahattsls - \thetaOracleTSLS } = \op{ 1 } .
  \end{equation*}
  This is precisely the content of Lemmas~\ref{lem:thetahatTSLS2thetaOracleTSLS}.
\end{proof}

\subsection{Peer Contagion Estimators} \label{apx:lattruepeerfit}

\begin{lemma}\label{lem:betahatTSLS2thetaOracleTSLS}
  Under the latent contagion model in Equation~\eqref{eq:lim-latent-red},
  suppose that Assumptions~\ref{assum:Apop:spectrum} through~\ref{assum:stronger4dhat} hold.
  Then there exists a sequence of orthogonal matrices $\Qdes \in \R^{(2+p+d)\times(2+p+d)}$ such that
  \begin{equation*}
    \sqrt{n} \paren*{ \Qdes \betahattsls - \thetaOracleTSLS } = \op{ 1 } .
  \end{equation*}
\end{lemma}
\begin{proof}
Take $\Qdes$ to be as in Equation~\eqref{eq:def:Qdes}.
Recalling the definitions from Equations~\eqref{eq:def:betahattsls} and~\eqref{eq:def:thetaOracleTSLS},
  \begin{equation*}
    \Qdes \betahattsls - \thetaOracleTSLS
    = \Qdes
	\left(\Zcheck^\top \Mcheck \Zcheck \right)^{-1} \Zcheck^\top \Mcheck \y
    - \left(\ZlatOracle^\top \MlatOracle \ZlatOracle \right)^{-1} \ZlatOracle^\top \MlatOracle \Y .
  \end{equation*}

  Applying the triangle inequality,
  \begin{equation} \label{eq:betahat2thetaOracle}
    \begin{aligned}
      \left\| \Qdes \betahattsls - \thetaOracleTSLS \right\|
       & \le
      \left\| \Qdes \left(\Zcheck^\top \Mcheck \Zcheck \right)^{-1} \Qdestop
      \left( \Mcheck \Zcheck \Qdestop - \MlatOracle \ZlatOracle \right)^\top \Y \right\| \\
       & ~~~~~~~~~+
      \left\|\left[\Qdes\left(\Zcheck^\top\Mcheck\Zcheck\right)^{-1} \Qdestop
        - \left(\ZlatOracle^\top \MlatOracle \ZlatOracle \right)^{-1} \right]
      \left( \MlatOracle \ZlatOracle \right)^\top \Y \right\| .
    \end{aligned} \end{equation}
  We will control each of these right-hand terms separately.

  Following an argument parallel to that given leading up to Equation~\eqref{eq:latcon:MZdiffs:done} in the proof of Lemma~\ref{lem:thetahatTSLS2thetaOracleTSLS},
  but using Lemmas~\ref{lem:McheckZcheck:forWeyl},~\ref{lem:McheckMlat:spectral} and~\ref{lem:ZcheckZlat:spectral}
  in place of Lemmas~\ref{lem:MhatZhat:forWeyl},~\ref{lem:MhatMlat:spectral} and~\ref{lem:ZhatZlat:spectral}, respectively, we obtain
  \begin{equation*}
    \left\| \Qdes (\Zcheck^\top \Mcheck \Zcheck)^{-1} \Qdestop
    \left( \Mcheck \Zcheck \Qdestop - \MlatOracle \ZlatOracle \right)^\top \Y
    \right\|
    = \op{ \frac{1}{\sqrt{n}} } .
  \end{equation*}

  Following an argument parallel to that leading to Equation~\eqref{eq:Zcovs:hat:done} in the proof of Lemma~\ref{lem:thetahatTSLS2thetaOracleTSLS},
  but using Lemmas~\ref{lem:McheckZcheck:forWeyl},~\ref{lem:GGtilde:quad},~\ref{lem:ZcheckZlat:spectral} and~\ref{lem:McheckMlat:spectral}
  in place of
  Lemmas~\ref{lem:MhatZhat:forWeyl},~\ref{lem:GhatGtilde:quad},~\ref{lem:ZhatZlat:spectral} and~\ref{lem:MhatMlat:spectral}, respectively, yields
  \begin{equation*}
    \left\| \Qdes \left[ (\Zcheck^\top \Mcheck \Zcheck )^{-1} \Qdestop
      - (\ZlatOracle^\top \MlatOracle \ZlatOracle)^{-1} \right]
    \left( \MlatOracle \ZlatOracle \right)^\top \Y \right\|
    = \op{ \frac{1}{\sqrt{n}} }.
  \end{equation*}

  Applying the above two displays to Equation~\eqref{eq:betahat2thetaOracle} and multiplying through by $\sqrt{n}$ completes the proof.
\end{proof}

\section{Peer Contagion Data Model} \label{apx:peercon}

Here we provide proof details for Theorems~\ref{thm:peertruepeerfit} and~\ref{thm:peertruelatfit}, which describe the behavior of our estimators when the responses $\Y$ are driven by contagion on the graph structure encoded by the observed network $\A$, i.e., what we have termed {\em peer contagion}, as given in Equation~\eqref{eq:lim-peer-red}.
The key challenge under the peer contagion model, as in the latent contagion model considered in Appendix~\ref{apx:latcon}, is that we do not have access to the latent positions $\Xpop$.
If we had access to this latent structure, we could construct oracle analogues of the estimators in Equation~\eqref{eq:def:betahattsls},
\begin{equation} \label{eq:def:betaOracle}
  \betaOracleTSLS = (\ZpeerOracle^\top \MpeerOracle \ZpeerOracle)^{-1}
  \ZpeerOracle^\top \MpeerOracle \Y,
\end{equation}
where $\MpeerOracle$ is as defined in Equation~\eqref{eq:def:MpeerOracle} and $\ZpeerOracle$ is given by
\begin{equation} \label{eq:def:ZpeerOracle}
  \ZpeerOracle = \begin{bmatrix} \1_n \; \W \; \Xpop \; \G \Y \end{bmatrix}
  \in \R^{n \times (p + d + 2)}.
\end{equation}

We begin by proving convergence of the peer contagion estimators under the peer contagion data model to establish Theorem~\ref{thm:peertruepeerfit}.
We give analogous proofs for the latent contagion estimators in Section~\ref{apx:peertruelatfit} below.

\subsection{Peer Contagion Estimators} \label{apx:peertruepeerfit}

Our proof of Theorem~\ref{thm:peertruepeerfit} relies on showing that when the model in Equation~\eqref{eq:lim-peer} holds, our estimators are close to those in Equation~\eqref{eq:def:betaOracle}.
Lemmas~\ref{lem:betahatTSLS2betaOracleTSLS} establishes this.

\begin{lemma}\label{lem:betahatTSLS2betaOracleTSLS}
  Under the peer contagion model in Equation~\eqref{eq:lim-peer-red}, suppose that Assumptions~\ref{assum:Apop:spectrum} through~\ref{assum:latentpositions} and Assumption~\ref{assum:peer-oracle} hold.
  Then there exists a sequence of orthogonal matrices $\Qdes \in \R^{(2+p+d)\times(2+p+d)}$ such that
  \begin{equation*}
    \sqrt{n} \paren*{ \Qdes \betahattsls - \betaOracleTSLS } = \op{ 1 } .
  \end{equation*}
\end{lemma}
\begin{proof}
Take $\Qdes$ to be as in Equation~\eqref{eq:def:Qdes}.
  Recalling the definitions from Equations~\eqref{eq:def:betahattsls} and~\eqref{eq:def:betaOracle} and using the fact that $\Qdes$ is orthogonal,
  \begin{equation*}
    \Qdes \betahattsls - \betaOracleTSLS
    =
    \Qdes \left(\Zcheck^\top \Mcheck \Zcheck \right)^{-1}
	\Qdestop \Qdes \Zcheck^\top \Mcheck \Y
    - (\ZpeerOracle^\top \MpeerOracle \ZpeerOracle)^{-1}
    \ZpeerOracle^\top \MpeerOracle \Y .
  \end{equation*}
  Applying the triangle inequality,
  \begin{equation} \label{eq:betahat2betaOracle}
  \begin{aligned}
  \left\| \Qdes \betahattsls - \betaOracleTSLS \right\|
   & \le
  \left\| \Qdes \left(\Zcheck^\top \Mcheck \Zcheck \right)^{-1} \Qdestop
    	\left( \Mcheck \Zcheck \Qdestop 
    	- \MpeerOracle \ZpeerOracle \right)^\top \Y \right\| \\
   & ~~~~~~~~~+
  \left\| \left[ \Qdes \left(\Zcheck^\top \Mcheck \Zcheck \right)^{-1} \Qdestop
    - \left(\ZpeerOracle^\top \MpeerOracle \ZpeerOracle \right)^{-1} \right]
  \left( \MpeerOracle \ZpeerOracle \right)^\top \Y \right\| .
  \end{aligned} \end{equation}
  We will control each of these right-hand terms separately.

  \paragraph{Controlling the first term in Equation~\eqref{eq:betahat2betaOracle}}
  \enspace

  By submultiplicativity and unitary invariance of the norm,
  \begin{equation*}
    \left\| \Qdes(\Zcheck^\top \Mcheck \Zcheck)^{-1} \Qdestop
    \left( \Mcheck \Zcheck \Qdestop - \MpeerOracle \ZpeerOracle \right)^\top \Y
    \right\|
    \le
    \frac{ \left\| \left( \Mcheck \Zcheck \Qdestop
		- \MpeerOracle \ZpeerOracle \right)^{\!\top} \Y \right\|}
    {\sigmamin^2( \Mcheck \Zcheck ) }
    \le
    \frac{ C \left\| \left( \Mcheck \Zcheck \Qdestop
		- \MpeerOracle \ZpeerOracle \right)^{\!\top} \Y \right\|}
    {\sigmamin^2( \MlatOracle \ZlatOracle ) },
  \end{equation*}
  where the second inequality follows from Weyl's inequality and Lemma~\ref{lem:McheckZcheck:forWeyl}.
  Using our assumption in Equation~\eqref{eq:assum:Zgrowth}, it follows that
  \begin{equation} \label{eq:peercon:peerfit:MZdiffs:Weyl}
    \left\| \Qdes (\Zcheck^\top \Mcheck \Zcheck)^{-1} \Qdestop
    \left( \Mcheck \Zcheck \Qdestop - \MpeerOracle \ZpeerOracle \right)^\top \Y
    \right\|
    \le
    \frac{ C \left\| \left( \Mcheck \Zcheck \Qdestop
		- \MpeerOracle \ZpeerOracle \right)^{\!\top} \Y \right\|}
    { n } .
  \end{equation}

  By the triangle inequality,
  \begin{equation} \label{eq:peercon:peerfit:MZdiffs:numer}
    \left\| \left( \Mcheck \Zcheck \Qdestop
		- \MpeerOracle \ZpeerOracle \right)^{\!\top} \Y \right\|
    \le
    \left\| \Qdes \Zcheck^\top \left(\Mcheck-\MpeerOracle\right) \Y \right\|
    + \left\| \left( \Zcheck \Qdestop -\ZpeerOracle \right)^\top
		\MpeerOracle \Y \right\|.
  \end{equation}
  Again using the triangle inequality,
  \begin{equation} \label{eq:peercon:peerfit:ZM:diff:Y:triangle}
    \left\| \Qdes \Zcheck^\top \left(\Mcheck-\MpeerOracle\right) \Y \right\|
    \le \left\| \left(\Zcheck \Qdestop -\ZpeerOracle\right)^\top 
		\left(\Mcheck-\MpeerOracle\right) \Y \right\|
    + \left\| \ZpeerOracle^\top \left(\Mcheck-\MpeerOracle\right) \Y \right\|.
  \end{equation}

  Using submultiplicativity followed by Lemmas~\ref{lem:ZcheckZpeer:spectral} and ~\ref{lem:McheckMpeer:spectral},
  \begin{equation*}
    \left\| \left(\Zcheck \Qdestop-\ZpeerOracle \right)^\top 
		:\left(\Mcheck-\MpeerOracle\right) \Y \right\|
    \le \left\| \Zcheck \Qdestop -\ZpeerOracle \right\|
    \left\| \Mcheck-\MpeerOracle \right\| \left\| \Y \right\|
    =
    \op{ \frac{ \sigmamin( \MlatOracle \ZlatOracle ) \| \Y \| }
		{ \sigmamin(\HlatOracle)} } .
  \end{equation*}
  Using the fact that $\MlatOracle$ is a projection and applying our growth assumption in Equation~\eqref{eq:assum:Zgrowth}, we have $\sigmamin(\MlatOracle \ZlatOracle) \le \| \ZlatOracle \| = \Op{\sqrt{n}}$, and thus
  \begin{equation} \label{eq:peercon:peerfit:ZM:diff:Y:triangleA}
    \left\| \left(\Zcheck \Qdestop -\ZpeerOracle \right)^\top 
		\left(\Mcheck-\MpeerOracle\right) \Y \right\|
    = \op{ \frac{ \| \Y \| }{ \sigmamin(\HlatOracle) } }
    = \op{ \sqrt{n} },
  \end{equation}
  where the second equality follows from  our growth assumption in Equation~\eqref{eq:assum:Hgrowth} and using Lemma~\ref{lem:Y:peercon} to ensure $\|\Y\|=\Op{\sqrt{n}}$.

  Using Lemma~\ref{lem:Y:peercon} again, write $\Y = \Ytilde + \zeta$, where $\Ytilde$ is independent of $\A - \Apop$ conditional on $\Xpop$ with
  \begin{equation} \label{eq:Ytilde}
    \left\| \Ytilde \right\| = \Op{ \sqrt{n} }
    ~\text{ and }~
    \| \zeta \| = \op{ \sqrt{n} }.
  \end{equation}
  Applying the triangle inequality,
  \begin{equation} \label{eq:peercon:peerfit:Mdiffquad:Ydecomp}
    \left\| \ZpeerOracle^\top \left(\Mcheck-\MpeerOracle\right) \Y \right\|
    \le
    \left\| \ZpeerOracle^\top \left(\Mcheck-\MpeerOracle\right) \Ytilde \right\|
    + \left\| \ZpeerOracle^\top \left(\Mcheck-\MpeerOracle\right) \zeta \right\|.
  \end{equation}
  By submultiplicativity, Lemma~\ref{lem:McheckMpeer:spectral} and the assumptions in Equations~\eqref{eq:assum:Hgrowth} and~\eqref{eq:assum:Zgrowth},
  \begin{equation*}
    \left\| \ZpeerOracle^\top \left(\Mcheck-\MpeerOracle\right) \Ytilde \right\|
    = \op{ \frac{ \| \ZlatOracle \| \| \Ytilde \| }{ \sigmamin(\HlatOracle) } }
    = \op{ \| \Ytilde \| }
    = \op{ \sqrt{n} },
  \end{equation*}
  where the last equality follows from $\|\Ytilde\| = \Op{ \sqrt{n} }$.
  Applying submultiplicativity, Lemma~\ref{lem:McheckMpeer:spectral}
  and the fact that $\zeta = \op{ \sqrt{n} }$,
  \begin{equation*}
    \left\| \ZpeerOracle^\top \left(\Mcheck-\MpeerOracle\right) \zeta \right\|
    = \op{ \frac{ \| \ZpeerOracle \| \sqrt{n} }{ \sigmamin( \HlatOracle ) } }
    = \op{ \sqrt{n} },
  \end{equation*}
  where the second equality follows from our growth assumptions in Equations~\eqref{eq:assum:Hgrowth} and~\eqref{eq:assum:Zgrowth}.
  Applying the above two display equations to Equation~\eqref{eq:peercon:peerfit:Mdiffquad:Ydecomp},
  \begin{equation*}
    \left\| \ZpeerOracle^\top \left(\Mcheck-\MpeerOracle\right) \Y \right\|
    = \op{ \sqrt{n} } .
  \end{equation*}
  Applying this and Equation~\eqref{eq:peercon:peerfit:ZM:diff:Y:triangleA}
  to Equation~\eqref{eq:peercon:peerfit:ZM:diff:Y:triangle},
  \begin{equation*}
    \left\| \Qdes \Zcheck^\top \left(\Mcheck-\MpeerOracle\right) \Y \right\|
    = \op{ \sqrt{n} } .
  \end{equation*}
  Finally, applying this to Equation~\eqref{eq:peercon:peerfit:MZdiffs:numer},
  \begin{equation} \label{eq:peercon:peerfit:MZdiffs:check:numer:intermezzo}
    \left\| \left(\Mcheck\Zcheck\Qdestop 
		- \MpeerOracle\ZpeerOracle\right)^{\!\top} \Y \right\|
    \le
    \left\| \left( \Zcheck \Qdestop -\ZpeerOracle \right)^\top 
		\MpeerOracle \Y \right\| + \op{ \sqrt{n} }.
  \end{equation}

  Recalling the structure of $\ZpeerOracle$ from Equation~\eqref{eq:def:ZpeerOracle} and applying the triangle inequality,
  \begin{equation*}
    \left\| \left(\Zcheck\Qdestop - \ZpeerOracle \right)^\top 
		\MpeerOracle \Y \right\|
    \le
    \| \left(\Xhat \Q^\top -\Xpop \right)^\top \MpeerOracle \Y \|
    + \left| \Y^\top \MpeerOracle \left(\G - \G \right) \Y \right|
    = \| \left(\Xhat \Q^\top-\Xpop \right)^\top \MpeerOracle \Y \| .
  \end{equation*}
  Applying the triangle inequality and recalling $\MlatOracle$ from Equation~\eqref{eq:def:MlatOracle},
  \begin{equation} \label{eq:peercon:peerfit:Zdiffs:triangle1}
    \left\| \left(\Zcheck \Qdestop - \ZpeerOracle \right)^\top 
		\MpeerOracle \Y \right\|
    \le \| \left(\Xhat \Q^\top-\Xpop \right)^\top 
		\left( \MpeerOracle - \MlatOracle \right) \Y \|
    + \| \left(\Xhat\Q^\top-\Xpop\right)^\top \MlatOracle \Y \| .
  \end{equation}

  By submultiplicativity followed by Lemmas~\ref{lem:XhatXTB:ctrl} and~\ref{lem:MpeerMlat:spectral},
  \begin{equation} \label{eq:peercon:peerfit:Zdiffs:submultTerm}
    \begin{aligned}
      \left\| \left(\Xhat\Q^{\!\top}\!-\!\Xpop\right)^{\!\top} \! \!
		\left( \MpeerOracle \!-\! \MlatOracle \right)\! \Y \right\|
       & \le
      \left\| \Xhat \Q^{\!\top} \!-\! \Xpop \right\|
      \left\| \MpeerOracle - \MlatOracle \right\|
      \left\| \Y \right\|   \\
       & =
      \op{
        \! \frac{ \| \Y \| }{ \sigmamin(\HlatOracle) } \!
        \left[\!
          \frac{ \sqrt{ \nu\! +\! b^2 } \sqrt{n} \log n }{ \sqrt{\spop_d} }
          \!+\!
          \frac{ \kappa(\Apop) \!(\nu\!+\!b^2\!) n^{3/2} \log^2 \! n }{ \spop_d^{3/2} }
          \right] \!
      }                     \\
       & = \op{ \sqrt{n} },
    \end{aligned} \end{equation}
  where we have used Lemma~\ref{lem:Y:peercon} to ensure that $\|\Y\| = \Op{\sqrt{n}}$ and the growth rates in Equations~\eqref{eq:assum:spectralConc},~\eqref{eq:assum:XandW:spectrumUB},~\eqref{eq:assum:XhatX:rate:2} and~\eqref{eq:assum:Hgrowth}.
  Applying this to Equation~\eqref{eq:peercon:peerfit:Zdiffs:triangle1},
  \begin{equation} \label{eq:peercon:peerfit:Zdiffs:triangle2}
    \left\| \left(\Zcheck \Qdestop - \ZpeerOracle \right)^\top 
		\MpeerOracle \Y \right\|
    \le \| \left(\Xhat \Q^\top-\Xpop\right)^\top \MlatOracle \Y \|
    + \op{ \sqrt{n} } .
  \end{equation}

  Recalling $\MlatOracle = \HlatOracle \left( \HlatOracle^\top \HlatOracle \right)^{-1} \HlatOracle^\top$ and using submultiplicativity,
  \begin{equation*}
    \left\| \left(\Xhat\Q^\top-\Xpop\right)^\top \MlatOracle \Y \right\|
    \le
    \frac{ \left\| \left(\Xhat\Q^\top-\Xpop\right)^\top \HlatOracle \right\|
      \left\| \HlatOracle \right\| \left\| \Y \right\| }
    { \sigmamin^2( \HlatOracle ) }
    \le
    \frac{ \kappa(\HlatOracle) \left\| (\Xhat\Q^\top-\Xpop)^\top 
		\HlatOracle \right\| \left\| \Y \right\| }
    { \sigmamin( \HlatOracle ) } .
  \end{equation*}
  Applying Lemma~\ref{lem:XhatXTB:ctrl} with $\mB = \HlatOracle$ and using Lemma~\ref{lem:Y:peercon} to ensure that $\| \Y \| = \Op{\sqrt{n}}$,
  \begin{equation*}
    \left\| \left(\Xhat\Q^\top-\Xpop\right)^\top \MlatOracle \Y \right\|
    \le
    \frac{ \kappa(\HlatOracle) \sqrt{n} \| \HlatOracle \| }{ \sigmamin(\HlatOracle) }
    \left[ \frac{ \sqrt{ \nu + b^2 } \log n }{ \sqrt{\spop_d} }
      +
      \frac{ \kappa(\Apop) (\nu+b^2) n \log^2 n }{ \spop_d^{3/2} }
      \right].
  \end{equation*}
  Applying our growth rates in Equations~\eqref{eq:assum:spectralConc},~\eqref{eq:assum:XandW:spectrumUB},~\eqref{eq:assum:XhatX:rate:2} and~\eqref{eq:assum:Hgrowth},
  \begin{equation*}
    \left\| (\Xhat\Q^\top-\Xpop)^\top \MlatOracle \Y \right\| 
	= \op{ \sqrt{n} } .
  \end{equation*}
  Applying this to Equation~\eqref{eq:peercon:peerfit:Zdiffs:triangle2},
  \begin{equation*}
    \left\| (\Zcheck\Qdestop - \ZpeerOracle )^\top \MpeerOracle \Y \right\|
    = \op{ \sqrt{n} } .
  \end{equation*}
  Applying this to Equation~\eqref{eq:peercon:peerfit:MZdiffs:check:numer:intermezzo} in turn,
  \begin{equation*}
    \left\| \left( \Mcheck \Zcheck \Qdestop - \MpeerOracle \ZpeerOracle \right)^{\!\top} \Y \right\|
    = \op{ \sqrt{n} } .
  \end{equation*}
  Finally, applying this bound to Equation~\eqref{eq:peercon:peerfit:MZdiffs:Weyl},
  \begin{equation} \label{eq:peercon:peerfit:MZdiffs:done}
    \left\| \Qdes (\Zcheck^\top \Mcheck \Zcheck)^{-1} \Qdestop
    \left( \Mcheck \Zcheck \Qdestop - \MpeerOracle \ZpeerOracle \right)^\top \Y
    \right\|
    = \op{ \frac{1}{\sqrt{n}} } .
  \end{equation}

  \paragraph{Controlling the second term in Equation~\eqref{eq:betahat2betaOracle}}
  \enspace

  By Lemmas~\ref{lem:McheckZcheck:forWeyl} and~\ref{lem:MpeerZpeer:forWeyl} along with our growth assumption in Equation~\eqref{eq:assum:Zgrowth}, both $\Zcheck^\top \Mcheck \Zcheck$ and $\ZpeerOracle^\top \MpeerOracle \ZpeerOracle$ are invertible with high probability for all $n$ suitably large.
  Thus, factoring appropriately, applying submultiplicativity and recalling that $\Mcheck$ and $\MpeerOracle$ are projections,
  \begin{equation*}
    \left\| \left[\Qdes\left(\Zcheck^\top\Mcheck\Zcheck\right)^{-1}\Qdestop
      - (\ZpeerOracle^\top \MpeerOracle \ZpeerOracle)^{-1} \right]
    \left( \MpeerOracle \ZpeerOracle \right)^\top \Y \right\|
    \le
    \frac{ \left\| \Qdes \Zcheck^\top \Mcheck \Zcheck \Qdestop
      - \ZpeerOracle^\top \MpeerOracle \ZpeerOracle \right\|
      \left\| \ZpeerOracle \right\| \left\| \Y \right\| }
    {\sigmamin^2(\Mcheck \Zcheck) \sigmamin^2(\MpeerOracle \ZpeerOracle) } .
  \end{equation*}
  Using Lemmas~\ref{lem:McheckZcheck:forWeyl} and~\ref{lem:MpeerZpeer:forWeyl} again along with Weyl's inequality, we can ensure that
  \begin{equation*}
    \sigmamin( \Mcheck \Zcheck ) \sigmamin( \MpeerOracle \ZpeerOracle ) 
	= \Omegap{ \sigmamin^2( \MlatOracle \ZlatOracle ) },
  \end{equation*}
  so that our growth assumption in Equation~\eqref{eq:assum:Zgrowth}, along with Lemma~\ref{lem:Y:peercon} yield
  \begin{equation*}
    \left\| \left[ \Qdes (\Zcheck^\top \Mcheck \Zcheck)^{-1} \Qdestop
      - (\ZpeerOracle^\top \MpeerOracle \ZpeerOracle)^{-1} \right]
    \left( \MpeerOracle \ZpeerOracle \right)^\top \Y \right\|
    \le
    \frac{ C \left\| \Qdes \Zcheck^\top \Mcheck \Zcheck \Qdestop
      - \ZpeerOracle^\top \MpeerOracle \ZpeerOracle \right\|
      \left\| \ZpeerOracle \right\| }
    { n^{3/2} } .
  \end{equation*}
  Applying Lemma~\ref{lem:ZpeerZlat:spectral}, Weyl's inequality and our growth assumption in Equation~\eqref{eq:assum:Zgrowth},
  \begin{equation} \label{eq:peercon:peerfit:Zcovs:start}
    \left\| \left[ \Qdes (\Zcheck^\top \Mcheck \Zcheck)^{-1} \Qdestop
      - (\ZpeerOracle^\top \MpeerOracle \ZpeerOracle)^{-1} \right]
    \left( \MpeerOracle \ZpeerOracle \right)^\top \Y \right\|
    \le
    \frac{ C \left\| \Qdes \Zcheck^\top \Mcheck \Zcheck \Qdestop
      - \ZpeerOracle^\top \MpeerOracle \ZpeerOracle \right\| }
    { n } .
  \end{equation}

  Adding and subtracting appropriate quantities and applying the triangle inequality,
  \begin{equation} \label{eq:peercon:peerfit:Zcovs:numertri} \begin{aligned}
    \left\| \Qdes \Zcheck^\top \Mcheck \Zcheck \Qdestop
    - \ZpeerOracle^\top \MpeerOracle \ZpeerOracle \right\|
    &\le
    2\left\| (\Zcheck \Qdestop -\ZpeerOracle)^\top 
		\MpeerOracle \ZpeerOracle \right\|
    +
    \left\| \Qdes \Zcheck^\top (\Mcheck-\MpeerOracle) \Zcheck \Qdestop \right\|\\
    &~~~~~~+
    \left\| (\Zcheck\Qdestop-\ZpeerOracle)^\top \MpeerOracle (\Zcheck\Qdestop-\ZpeerOracle) \right\| .
  \end{aligned} \end{equation}

  Decomposing $\Zcheck\Qdestop-\ZpeerOracle$ as in Equation~\eqref{eq:peercon:peerfit:Zdiffs:triangle1},
  \begin{equation*}
    \left\| (\Zcheck\Qdestop-\ZpeerOracle)^\top \MpeerOracle \ZpeerOracle \right\|
    \le
    \left\| (\Xhat\Q^\top-\Xpop)^\top \MpeerOracle \ZpeerOracle \right\|
    + \left| \ZpeerOracle^\top \MpeerOracle (\G-\G) \Y \right|
    = \left\| (\Xhat\Q^\top-\Xpop)^\top \MpeerOracle \ZpeerOracle \right\| .
  \end{equation*}
  Applying the triangle inequality,
  \begin{equation} \label{eq:peercon:peerfit:Zcovs:linearTerm}
    \left\| (\Zcheck\Qdestop-\ZpeerOracle)^\top 
		\MpeerOracle \ZpeerOracle \right\|
    \le \left\| (\Xhat\Q^\top-\Xpop)^\top
    \left( \MpeerOracle - \MlatOracle \right) \ZpeerOracle \right\|
    + \left\| (\Xhat\Q^\top-\Xpop)^\top \MlatOracle \ZpeerOracle \right\|.
  \end{equation}
  By submultiplicativity followed by Lemmas~\ref{lem:XhatXTB:ctrl} and~\ref{lem:MpeerMlat:spectral},
  \begin{equation*}
    \left\| (\Xhat\Q^{\!\top}-\Xpop)^{\!\top} \!
    \left( \MpeerOracle - \MlatOracle \right) \ZpeerOracle \right\|
    \le
    \op{ \frac{ \sqrt{n} \| \ZpeerOracle \| }{ \sigmamin(\HlatOracle)}
      \frac{ \sqrt{ \nu \!+\! b^2 } \log n }{ \sqrt{\spop_d} } }
    +
    \op{ \frac{ \| \ZpeerOracle \| }{ \sigmamin(\HlatOracle) }
      \frac{ \kappa(\Apop) (\nu\!+\!b^2) n \log^2 n }{ \spop_d^{3/2} } } ,
  \end{equation*}
  and the growth rates in Equations~\eqref{eq:assum:spectralConc},~\eqref{eq:assum:XandW:spectrumUB} and~\eqref{eq:assum:XhatX:rate:2} yield
  \begin{equation*}
    \left\| (\Xhat\Q^{\!\top}-\Xpop)^\top
    \left( \MpeerOracle - \MlatOracle \right) \ZpeerOracle \right\|
    = \op{ \frac{ \sqrt{n} \| \ZpeerOracle \| }{ \sigmamin(\HlatOracle)} }
    + \op{ \frac{ \| \ZpeerOracle \| }{ \sigmamin(\HlatOracle)} }
    = \op{ \sqrt{n} },
  \end{equation*}
  where the second equality follows from Lemma~\ref{lem:ZpeerZlat:spectral} and our growth assumptions in Equations~\eqref{eq:assum:Zgrowth} and~\eqref{eq:assum:Hgrowth}.
  Applying this to Equation~\eqref{eq:peercon:peerfit:Zcovs:linearTerm},
  \begin{equation*}
    \left\| (\Zcheck\Qdestop-\ZpeerOracle)^\top \MpeerOracle \ZpeerOracle \right\|
    \le \left\| (\Xhat\Q^\top-\Xpop)^\top \MlatOracle \ZpeerOracle \right\|
    + \op{ \sqrt{n} } .
  \end{equation*}
  Recalling the definition of $\MlatOracle$ from Equation~\eqref{eq:def:MlatOracle} and applying submultiplicativity,
  \begin{equation*}
    \left\| (\Zcheck\Qdestop-\ZpeerOracle)^\top \MpeerOracle \ZpeerOracle \right\|
    \le
    \frac{ \left\| (\Xhat\Q^\top-\Xpop)^\top \HlatOracle \right\|
      \kappa( \HlatOracle ) \left\| \ZpeerOracle \right\| }
    { \sigmamin( \HlatOracle ) }
    + \op{ \sqrt{n} } .
  \end{equation*}
  Applying Lemma~\ref{lem:XhatXTB:ctrl} with $\mB = \HlatOracle$ and using Equations~\eqref{eq:assum:spectralConc},~\eqref{eq:assum:XandW:spectrumUB},~\eqref{eq:assum:XhatX:rate:2} and~\eqref{eq:assum:Hgrowth},
  \begin{equation*}
    \left\| (\Zcheck\Qdestop-\ZpeerOracle)^\top \MpeerOracle \ZpeerOracle \right\|
    \le \op{ \left\| \ZpeerOracle \right\| } + \op{ \sqrt{n} } .
  \end{equation*}
  Applying Lemma~\ref{lem:ZhatZpeer:spectral} and our growth assumption in Equation~\eqref{eq:assum:Zgrowth},
  \begin{equation} \label{eq:peercon:peerfit:Zcovs:linearTerm:done}
    \left\| (\Zcheck\Qdestop-\ZpeerOracle)^\top \MpeerOracle \ZpeerOracle \right\|
    = \op{ \sqrt{n} },
  \end{equation}

  By submultiplicativity followed by Lemma~\ref{lem:ZcheckZlat:spectral}, Weyl's inequality, and the fact that $\sigmamin(\MlatOracle \ZlatOracle) \le \sigmamin( \ZlatOracle )$,
  \begin{equation*}
    \left\| \Qdes \Zcheck^\top ( \Mcheck-\MpeerOracle ) \Zcheck \Qdestop \right\|
    \le
    \left\| \Zcheck \right\|^2 \left\| \Mcheck-\MpeerOracle \right\|
    \le C \left\| \ZlatOracle \right\|^2 \left\| \Mcheck-\MpeerOracle \right\|,
  \end{equation*}
  Applying Lemma~\ref{lem:McheckMpeer:spectral} and our growth assumptions in Equations~\eqref{eq:assum:Hgrowth} and~\eqref{eq:assum:Zgrowth},
  \begin{equation} \label{eq:peercon:peerfit:Zcovs:Mdiffs:done}
    \left\| \Qdes \Zcheck^\top ( \Mcheck-\MpeerOracle ) \Zcheck \right\|
    = \op{ \sqrt{n} } .
  \end{equation}

  Using idempotence of $\MpeerOracle$,
  \begin{equation} \label{eq:peercon:peerfit:Zcovs:quad:start}
    \left\| (\Zcheck \Qdestop-\ZpeerOracle)^\top 
		\MpeerOracle (\Zcheck-\ZpeerOracle) \right\|
    = \left\| \MpeerOracle (\Zcheck\Qdestop-\ZpeerOracle) \right\|^2 .
  \end{equation}

  Applying the triangle inequality and using submultiplicativity,
  \begin{equation*}
    \left\| \MpeerOracle (\Zcheck\Qdestop-\ZpeerOracle) \right\|
    \le \left\| \MlatOracle (\Zcheck\Qdestop-\ZpeerOracle) \right\|
    + \left\| \MpeerOracle - \MlatOracle \right\|
    \left\| \Zcheck\Qdestop-\ZpeerOracle \right\|
    \le \left\| \MlatOracle (\Zcheck\Qdestop-\ZpeerOracle) \right\|
    + \op{ \frac{ \| \MlatOracle \ZlatOracle \| }
      { \sigmamin(\HlatOracle)} },
  \end{equation*}
  where the second inequality follows from Lemmas~\ref{lem:ZcheckZpeer:spectral} and~\ref{lem:MhatMlat:spectral}.
  Recalling that $\MlatOracle$ is a projection and applying our growth assumptions in Equations~\eqref{eq:assum:Zgrowth} and~\eqref{eq:assum:Hgrowth},
  \begin{equation*}
    \left\| \MpeerOracle (\Zcheck\Qdestop-\ZpeerOracle) \right\|
    \le \left\| \MlatOracle (\Zcheck\Qdestop-\ZpeerOracle) \right\| + \op{ 1 } .
  \end{equation*}
  Recalling the definition of $\MlatOracle$ from Equation~\eqref{eq:def:MlatOracle} and applying submultiplicativity,
  \begin{equation} \label{eq:peercon:peerfit:Zcovs:MZZ:almostDone}
    \left\| \MpeerOracle (\Zcheck\Qdestop-\ZpeerOracle) \right\|
    \le \frac{ \kappa( \HlatOracle ) }{ \sigmamin(\HlatOracle) }
    \left\| \HlatOracle^\top (\Zcheck\Qdestop-\ZpeerOracle) \right\|
    + \op{ 1 } .
  \end{equation}
  Recalling the structure of $\Zcheck$ and $\ZpeerOracle$ from Definition~\ref{def:feasibleEstimators:peer} and Equation~\eqref{eq:def:ZpeerOracle}, respectively,
  \begin{equation*}
    \left\| \HlatOracle^\top (\Zcheck\Qdestop-\ZpeerOracle) \right\|
    \le \left\| \HlatOracle^\top \left( \Xhat\Q^\top - \Xpop \right) \right\|
    + \left\| \HlatOracle^\top \left( \G - \G \right) \Y \right\|
    = \left\| \HlatOracle^\top \left( \Xhat\Q^\top - \Xpop \right) \right\|,
  \end{equation*}
  and Lemma~\ref{lem:XhatXTB:ctrl} implies, using our growth assumptions in Equations~\eqref{eq:assum:spectralConc},~\eqref{eq:assum:XandW:spectrumUB} and~\eqref{eq:assum:XhatX:rate:2},
  \begin{equation*}
    \left\| \HlatOracle^\top (\Zcheck\Qdestop-\ZpeerOracle) \right\|
    = \op{ \left\| \HlatOracle \right\| } .
  \end{equation*}
  Applying this to Equation~\eqref{eq:peercon:peerfit:Zcovs:MZZ:almostDone} and applying our growth assumption in Equation~\eqref{eq:assum:Hgrowth},
  \begin{equation*}
    \left\| \MpeerOracle (\Zcheck\Qdestop-\ZpeerOracle) \right\|
    = \op{ 1 } .
  \end{equation*}
  Applying this in turn to Equation~\eqref{eq:peercon:peerfit:Zcovs:quad:start},
  \begin{equation} \label{eq:peercon:peerfit:Zcovs:quad:done}
    \left\| (\Zcheck\Qdestop-\ZpeerOracle)^\top \MpeerOracle
    (\Zcheck\Qdestop-\ZpeerOracle) \right\|
    = \op{ 1 } .
  \end{equation}

  Applying Equations~\eqref{eq:peercon:peerfit:Zcovs:linearTerm:done},
  ~\eqref{eq:peercon:peerfit:Zcovs:Mdiffs:done}
  and~\eqref{eq:peercon:peerfit:Zcovs:quad:done}
  to Equation~\eqref{eq:peercon:peerfit:Zcovs:start},
  \begin{equation} \label{eq:peercon:peerfit:Zcovs:hat:done}
    \left\| \left[ \Qdes (\Zcheck^\top \Mcheck \Zcheck)^{-1} \Qdestop
      - (\ZpeerOracle^\top \MpeerOracle \ZpeerOracle)^{-1} \right]
    \left( \MpeerOracle \ZpeerOracle \right)^\top \Y \right\|
    = \op{ \frac{1}{\sqrt{n}} }.
  \end{equation}
  Applying this and Equation~\eqref{eq:peercon:peerfit:MZdiffs:done} to Equation~\eqref{eq:betahat2betaOracle} and multiplying through by $\sqrt{n}$ completes the proof.
\end{proof}

\subsection{Latent Contagion Estimators} \label{apx:peertruelatfit}

We now establish corresponding results for the behavior of the latent contagion estimators when the peer contagion model is true.

\begin{lemma}\label{lem:thetahatTSLS2betaOracleTSLS}
  Under the peer contagion model in Equation~\eqref{eq:lim-latent-red},
  suppose that Assumptions~\ref{assum:Apop:spectrum} through~\ref{assum:stronger4dhat} hold.
  Then there exists a sequence of orthogonal matrices $\Qdes \in \R^{(2+p+d)\times(2+p+d)}$ such that
  \begin{equation*}
    \sqrt{n} \paren*{ \Qdes \thetahattsls - \betaOracleTSLS } = \op{ 1 } .
  \end{equation*}
\end{lemma}
\begin{proof}
Take $\Qdes$ to be as in Equation~\eqref{eq:def:Qdes}.
  Recalling the definitions in Equations~\eqref{eq:def:thetahattsls} and~\eqref{eq:def:betaOracle} and using the fact that $\Qdes$ is orthogonal,
  \begin{equation*}
  \Qdes \thetahattsls - \betaOracleTSLS
   =
   \left[ \Qdes (\Zhat^\top \Mhat \Zhat)^{-1} \Qdestop \left( \Mhat \Zhat \Qdestop \right)^\top
      - (\ZpeerOracle^\top \MpeerOracle \ZpeerOracle)^{-1}
      \left( \MpeerOracle \ZpeerOracle \right)^\top \right] \Y .
  \end{equation*}
  Simplifying and applying the triangle inequality,
  \begin{equation} \label{eq:thetahat2betaoracle}
    \begin{aligned}
      \left\| \Qdes \thetahattsls - \betaOracleTSLS \right\|
       & \le
      \left\| \Qdes (\Zhat^\top \Mhat \Zhat)^{-1} \Qdestop
      \left( \Mhat \Zhat \Qdestop - \MpeerOracle \ZpeerOracle \right)^\top \Y \right\| \\
       & ~~~~~~+
      \left\| \left[ \Qdes (\Zhat^\top \Mhat \Zhat)^{-1} \Qdestop
        - (\ZpeerOracle^\top \MpeerOracle \ZpeerOracle)^{-1} \right]
      \left( \MpeerOracle \ZpeerOracle \right)^\top \Y \right\| .
    \end{aligned} \end{equation}

  We will control each of these right-hand terms separately.

  \paragraph{Controlling the first term in Equation~\eqref{eq:thetahat2betaoracle}}
  \enspace

  By submultiplicativity of the norm,
  \begin{equation*}
    \left\| \Qdes (\Zhat^\top \Mhat \Zhat)^{-1} \Qdestop
    \left( \Mhat \Zhat \Qdestop - \MpeerOracle \ZpeerOracle \right)^\top \Y
    \right\|
    \le
    \frac{ \left\| \left( \Mhat \Zhat \Qdestop 
		- \MpeerOracle \ZpeerOracle \right)^{\!\top} \Y \right\|}
    {\sigmamin^2( \Mhat \Zhat ) }
    \le
    \frac{C\left\| \left( \Mhat \Zhat \Qdestop 
		- \MpeerOracle \ZpeerOracle \right)^{\!\top} \Y \right\|}
    {\sigmamin^2( \MlatOracle \ZlatOracle ) },
  \end{equation*}
  where the second inequality follows from Weyl's inequality and Lemma~\ref{lem:MhatZhat:forWeyl}.
  Using the growth assumption in Equation~\eqref{eq:assum:Zgrowth}, it follows that
  \begin{equation} \label{eq:peercon:latfit:MZdiffs:Weyl}
    \left\| \Qdes (\Zhat^\top \Mhat \Zhat)^{-1} \Qdestop
    \left( \Mhat \Zhat \Qdestop - \MpeerOracle \ZpeerOracle \right)^\top \Y
    \right\|
    \le
    \frac{ C \left\| \left( \Mhat \Zhat \Qdestop 
		- \MpeerOracle \ZpeerOracle \right)^{\!\top} \Y \right\|}
    { n } .
  \end{equation}

  By the triangle inequality,
  \begin{equation} \label{eq:peercon:latfit:MZdiffs:numer}
    \left\| \left( \Mhat \Zhat \Qdestop 
		- \MpeerOracle \ZpeerOracle \right)^{\!\top} \Y \right\|
    \le
    \left\| \Qdes \Zhat^\top \left(\Mhat-\MpeerOracle\right) \Y \right\|
    + \left\| \left( \Zhat \Qdes -\ZpeerOracle \right)^\top \MpeerOracle \Y \right\|.
  \end{equation}
  Again using the triangle inequality,
  \begin{equation} \label{eq:peercon:latfit:ZhatMdiff:Y:triangle}
    \left\| \Qdes \Zhat^\top \left(\Mhat-\MpeerOracle\right) \Y \right\|
    \le \left\| \left(\Zhat\Qdestop -\ZlatOracle\right)^\top \left(\Mhat-\MpeerOracle\right) \Y \right\|
    + \left\| \ZlatOracle^\top \left(\Mhat-\MpeerOracle\right) \Y \right\|.
  \end{equation}
  Using Lemma~\ref{lem:Y:peercon} to write $\Y = \Ytilde + \zeta$, where $\Ytilde$ and $\zeta$ obey the growth rates in Equation~\eqref{eq:Ytilde} and $\Ytilde$ is independent of $\A - \Apop$ conditional on $\Xpop$,
  \begin{equation*}
    \left\| \ZlatOracle^\top \left(\Mhat-\MpeerOracle\right) \Y \right\|
    \le \left\| \ZlatOracle^\top \left(\Mhat-\MpeerOracle\right) \Ytilde \right\|
    +  \left\| \ZlatOracle^\top \left(\Mhat-\MpeerOracle\right) \zeta \right\| .
  \end{equation*}
  Controlling the first term using Lemma~\ref{lem:MhatMpeer:spectral} followed by the growth assumptions in Equations~\eqref{eq:assum:Hgrowth} and~\eqref{eq:assum:Zgrowth}, and using Lemma~\ref{lem:Y:peercon} to ensure that $\|\Y\| = \Op{ \sqrt{n} }$,
  \begin{equation*}
    \left\| \ZlatOracle^\top \left(\Mhat-\MpeerOracle\right) \Y \right\|
    \le \left\| \ZlatOracle^\top \left(\Mhat-\MpeerOracle\right) \zeta \right\|
    + \op{ \frac{ \| \ZlatOracle \| \| \Ytilde \| }{ \sigmamin(\HlatOracle) } }
    \le \left\| \ZlatOracle^\top \left(\Mhat-\MpeerOracle\right) \zeta \right\|
    + \op{ \sqrt{n} } .
  \end{equation*}
  Applying submultiplicativity to the first term followed by Lemma~\ref{lem:MhatMpeer:spectral}, recalling that
  \begin{equation*}
    \left\| \ZlatOracle^\top \left(\Mhat-\MpeerOracle\right) \Y \right\|
    \le
    \op{ \frac{ \left\| \ZlatOracle \right\| \sqrt{n} }{ \sigmamin(\HlatOracle) } }
    + \op{ \sqrt{n} }
    = \op{ \sqrt{n} } ,
  \end{equation*}
  where the equality follows from the growth assumptions in Equations~\eqref{eq:assum:Zgrowth} and~\eqref{eq:assum:Hgrowth}.

  Using submultiplicativity followed by Lemmas~\ref{lem:ZhatZpeer:spectral} and ~\ref{lem:MhatMpeer:spectral},
  \begin{equation*}
    \left\| \left(\Zhat \Qdestop-\ZlatOracle \right)^\top \left(\Mhat-\MpeerOracle\right) \Y \right\|
    \le \left\| \Zhat \Qdestop-\ZlatOracle \right\|
    \left\| \Mhat-\MpeerOracle \right\| \left\| \Y \right\|
    =
    \op{ \frac{ \sigmamin( \MlatOracle \ZlatOracle ) \| \Y \| }
		{ \sigmamin(\HlatOracle)} }
    = \op{ \sqrt{n} },
  \end{equation*}
  where the last equality follows using Lemma~\ref{lem:Y:peercon} to ensure $\|\Y\|=\Op{\sqrt{n}}$ and the growth assumption in Equation~\eqref{eq:assum:Hgrowth}.
  Applying the above two displays to Equation~\eqref{eq:peercon:latfit:ZhatMdiff:Y:triangle},
  \begin{equation*}
    \left\| \Qdes \Zhat^\top \left(\Mhat-\MpeerOracle\right) \Y \right\|
    = \op{ \sqrt{n} }.
  \end{equation*}
  Applying this to Equation~\eqref{eq:peercon:latfit:MZdiffs:numer},
  \begin{equation} \label{eq:peercon:latfit:MZdiffs:numer:intermezzo}
    \left\| \left( \Mhat \Zhat \Qdestop 
		- \MpeerOracle \ZpeerOracle \right)^{\!\top} \Y \right\|
    \le
    \left\|\left(\Zhat\Qdestop-\ZpeerOracle\right)^\top\MpeerOracle\Y\right\|
    + \op{ \sqrt{n} } .
  \end{equation}

  Applying the triangle inequality,
  \begin{equation*}
    \left\| \left( \Zhat\Qdestop-\ZpeerOracle \right)^\top 
		\MpeerOracle \Y \right\|
    \le
    \left\| \left( \Zhat\Qdestop-\ZpeerOracle \right)^\top 
		\MlatOracle \Y \right\|
    +
    \left\| \left( \Zhat\Qdestop-\ZpeerOracle \right)^\top
    		\left( \MpeerOracle - \MlatOracle \right) \Y \right\| .
  \end{equation*}
  Applying submultiplicativity to the second term followed by Lemmas~\ref{lem:ZhatZpeer:spectral},~\ref{lem:MpeerMlat:spectral} and using Lemma~\ref{lem:Y:peercon} to ensure $\| \Y \| = \Op{\sqrt{n}}$,
  \begin{equation*}
    \left\| \left( \Zhat\Qdestop-\ZpeerOracle \right)^\top \MpeerOracle \Y \right\|
    \le
    \left\| \left( \Zhat\Qdestop-\ZpeerOracle \right)^\top \MlatOracle \Y \right\|
    + \op{ \frac{ \sqrt{n} ~\| \MlatOracle \ZlatOracle \| }
      { \sigmamin( \HlatOracle ) } }
    =
    \left\| \left( \Zhat\Qdestop-\ZpeerOracle \right)^\top \MlatOracle \Y \right\|
    + \op{ \sqrt{n} },
  \end{equation*}
  where the second equality follows from our assumptions in Equations~\eqref{eq:assum:Zgrowth} and~\eqref{eq:assum:Hgrowth}.
  Applying this to Equation~\eqref{eq:peercon:latfit:MZdiffs:numer:intermezzo},
  \begin{equation} \label{eq:peercon:latfit:MZdiffs:numer:intermezzo:2}
    \left\| \left( \Mhat \Zhat \Qdestop - \MpeerOracle \ZpeerOracle \right)^{\!\top}
    \Y \right\|
    \le
    \left\| \left( \Zhat \Qdestop-\ZpeerOracle \right)^\top \MlatOracle \Y \right\|
    + \op{ \sqrt{n} } .
  \end{equation}

  Recalling the structure of $\ZpeerOracle$ from Equation~\eqref{eq:def:ZpeerOracle} and applying the triangle inequality,
  \begin{equation} \label{eq:peercon:latfit:Zdiffs:triangle}
    \left\| (\Zhat\Qdestop - \ZpeerOracle )^\top \MlatOracle \Y \right\|
    \le
    \| (\Xhat\Q^\top-\Xpop)^\top \MlatOracle \Y \|
    + \left| \Y^\top \MlatOracle (\Ghat - \G) \Y \right| .
  \end{equation}

  Recalling the definition of $\MlatOracle$ from Equation~\eqref{eq:def:MlatOracle}, and applying submultiplicativity,
  \begin{equation*}
    \| (\Xhat\Q^\top-\Xpop)^\top \MlatOracle \Y \|
    \le \frac{ \left\| (\Xhat\Q^\top-\Xpop)^\top \HlatOracle \right\|
      \left\| \HlatOracle \right\| \left\| \Y \right\| }
    { \sigmamin^2( \HlatOracle ) }
    \le C \left\| (\Xhat\Q^\top-\Xpop)^\top \HlatOracle \right\|,
  \end{equation*}
  where we have used our growth assumption in Equation~\eqref{eq:assum:Hgrowth} and used Lemma~\ref{lem:Y:peercon} to bound $\|\Y\| = \Op{ \sqrt{n} }$.
  Applying Lemma~\ref{lem:XhatXTB:ctrl} with $\mB = \HlatOracle$ and using our growth assumptions in Equations~\eqref{eq:assum:spectralConc},~\eqref{eq:assum:XandW:spectrumUB} and~\eqref{eq:assum:XhatX:rate:2},
  \begin{equation} \label{eq:peercon:latfit:Zdiffs:Xcomponent}
    \| (\Xhat\Q^\top-\Xpop)^\top \MlatOracle \Y \|
    = \op{ \left\| \HlatOracle \right\| }
    = \op{ \sqrt{n} } ,
  \end{equation}
  where the second equality follows from our assumptions in Equation~\eqref{eq:assum:Hgrowth}.
  Applying this to Equation~\eqref{eq:peercon:latfit:Zdiffs:triangle},
  \begin{equation*}
    \left\| (\Zhat \Qdestop - \ZpeerOracle )^\top \MlatOracle \Y \right\|
    \le
    \left| \Y^\top \MlatOracle (\Ghat - \G) \Y \right|
    + \op{ \sqrt{n} } .
  \end{equation*}
  Applying this to Equation~\eqref{eq:peercon:latfit:MZdiffs:numer:intermezzo:2},
  \begin{equation} \label{eq:peercon:latfit:MZdiffs:numer:justYGYleft}
    \left\| \left( \Mhat \Zhat \Qdestop 
		- \MpeerOracle \ZpeerOracle \right)^{\!\top} \Y \right\|
    \le
    \left| \Y^\top \MlatOracle (\Ghat - \G) \Y \right|
    + \op{ \sqrt{n} } .
  \end{equation}

  Recalling the definition of $\MlatOracle$ from Equation~\eqref{eq:def:MlatOracle} and applying submultiplicativity,
  \begin{equation*}
    \left| \Y^\top \MlatOracle (\Ghat - \G) \Y \right|
    \le \frac{ \| \Y \| \kappa( \HlatOracle ) }{ \sigmamin(\HlatOracle) }
    \left\| \HlatOracle^\top (\Ghat - \G) \Y \right\|
    \le C \left\| \HlatOracle^\top (\Ghat - \G) \Y \right\|,
  \end{equation*}
  where the second inequality follows from our assumptions in Equation~\eqref{eq:assum:Hgrowth} and using Lemma~\ref{lem:Y:peercon} to ensure $\| \Y \| = \Op{\sqrt{n} }$.
  Applying Lemma~\ref{lem:Y:peerquad} with $\vu = \HlatOracle$ and using our growth assumption in Equation~\eqref{eq:assum:Hgrowth},
  \begin{equation*}
    \left| \Y^\top \MlatOracle (\Ghat - \G) \Y \right|
    = \op{ \| \HlatOracle \| }
    = \op{ \sqrt{n} } .
  \end{equation*}
  Applying this to Equation~\eqref{eq:peercon:latfit:MZdiffs:numer:justYGYleft}
  and applying the resulting bound to Equation~\eqref{eq:peercon:latfit:MZdiffs:Weyl},
  \begin{equation} \label{eq:peercon:latfit:MZdiffs:done}
    \left\| \Qdes (\Zhat^\top \Mhat \Zhat)^{-1} \Qdestop
    \left( \Mhat \Zhat \Qdestop - \MpeerOracle \ZpeerOracle \right)^\top \Y
    \right\|
    = \op{ \frac{1}{\sqrt{n}} } .
  \end{equation}

  \paragraph{Controlling the second term in Equation~\eqref{eq:thetahat2betaoracle}}
  \enspace

  By Lemmas~\ref{lem:MhatZhat:forWeyl} and~\ref{lem:MpeerZpeer:forWeyl}  and our growth assumption in Equation~\eqref{eq:assum:Zgrowth}, both $\Zhat^\top \Mhat \Zhat$ and $\ZpeerOracle^\top \MpeerOracle \ZpeerOracle$ are invertible with high probability for all $n$ suitably large.
  Thus, factoring appropriately, applying submultiplicativity and recalling that $\Mhat$ and $\MpeerOracle$ are projections,
  \begin{equation*}
    \left\| \left[ \Qdes (\Zhat^\top \Mhat \Zhat)^{-1} \Qdestop
      - (\ZpeerOracle^\top \MpeerOracle \ZpeerOracle)^{-1} \right]
    \left( \MpeerOracle \ZpeerOracle \right)^\top \Y \right\|
    \le
    \frac{ \left\| \Qdes \Zhat^\top \Mhat \Zhat \Qdestop
      - \ZpeerOracle^\top \MpeerOracle \ZpeerOracle \right\|
      \left\| \ZpeerOracle \right\| \left\| \Y \right\| }
    { \sigmamin^2( \Mhat \Zhat ) \sigmamin^2( \MpeerOracle \ZpeerOracle ) } .
  \end{equation*}
  Using Lemmas~\ref{lem:MhatZhat:forWeyl} and~\ref{lem:MpeerZpeer:forWeyl} again along with Weyl's inequality, we can ensure that
  \begin{equation*}
    \sigmamin( \Mhat \Zhat ) \sigmamin(\MpeerOracle\ZpeerOracle)
    = \Omegap{ \sigmamin^2( \MlatOracle \ZlatOracle ) }.
  \end{equation*}
  Our growth assumptions in Equations~\eqref{eq:assum:Zgrowth} and~\eqref{eq:assum:Zgrowth}, as well as Lemma~\ref{lem:Y:peercon} then yield
  \begin{equation} \label{eq:peercon:latfit:Zcovs:start}
    \left\| \left[ \Qdes (\Zhat^\top \Mhat \Zhat)^{-1} \Qdestop
      - (\ZpeerOracle^\top \MpeerOracle \ZpeerOracle)^{-1} \right]
    \left( \MpeerOracle \ZpeerOracle \right)^\top \Y \right\|
    \le
    \frac{ C \left\| \Qdes \Zhat^\top \Mhat \Zhat \Qdestop
      - \ZpeerOracle^\top \MpeerOracle \ZpeerOracle \right\| }
    { n } .
  \end{equation}

  Adding and subtracting appropriate quantities and applying the triangle inequality,
  \begin{equation} \label{eq:peercon:latfit:Zcovs:numertri} \begin{aligned}
    \left\| \Qdes \Zhat^\top \Mhat \Zhat \Qdestop
    - \ZpeerOracle^\top \MpeerOracle \ZpeerOracle \right\|
    &\le
    2\left\| (\Zhat \Qdestop-\ZpeerOracle)^\top \MpeerOracle \ZpeerOracle \right\|
    +
    \left\| \Zhat^\top ( \Mhat-\MpeerOracle ) \Zhat \right\| \\
    &~~~~~~+
    \left\| (\Zhat \Qdestop-\ZpeerOracle)^\top \MpeerOracle (\Zhat\Qdestop-\ZpeerOracle) \right\| .
  \end{aligned} \end{equation}

  Applying the triangle inequality followed by submultiplicativity,
  \begin{equation*} \begin{aligned}
    \left\| (\Zhat\Qdestop-\ZpeerOracle)^\top \MpeerOracle \ZpeerOracle \right\|
    &\le \left\| (\Zhat\Qdestop-\ZpeerOracle)^\top \MlatOracle \ZpeerOracle \right\|
    + \left\| \Zhat\Qdestop-\ZpeerOracle \right\| \left\| \MpeerOracle - \MlatOracle \right\|
    \left\| \ZpeerOracle \right\| \\
    &= \left\| (\Zhat\Qdestop-\ZpeerOracle)^\top \MlatOracle \ZpeerOracle \right\|
    + \op{ \frac{ \| \MlatOracle \ZlatOracle \| \left\| \ZpeerOracle \right\| }{ \sigmamin(\HlatOracle) } }.
  \end{aligned} \end{equation*}
  Controlling $\| \ZpeerOracle \|$ with Lemma~\ref{lem:ZpeerZlat:spectral} and  Weyl's inequality and using using our growth assumptions in Equations~\eqref{eq:assum:Zgrowth} and~\eqref{eq:assum:Hgrowth},
  \begin{equation*}
    \left\| (\Zhat\Qdestop-\ZpeerOracle)^\top \MpeerOracle \ZpeerOracle \right\|
    \le
    \left\| (\Zhat\Qdestop-\ZpeerOracle)^\top \MlatOracle \ZpeerOracle \right\|
    + \op{ \sqrt{n} } .
  \end{equation*}

  Decomposing $\Zhat-\ZpeerOracle$ as in Equation~\eqref{eq:peercon:latfit:Zdiffs:triangle},
  \begin{equation} \label{eq:peercon:latfit:Zcovs:linearTerm:prelim}
    \left\| (\Zhat\Qdestop-\ZpeerOracle)^\top \MpeerOracle \ZpeerOracle \right\|
    \le
    \left\| (\Xhat\Q^\top-\Xpop)^\top \MlatOracle \ZpeerOracle \right\|
    + \left\| \ZpeerOracle^\top \MlatOracle (\Ghat-\G) \Y \right\| + \op{ \sqrt{n} } .
  \end{equation}
  Recalling the definition of $\MlatOracle$ from Equation~\eqref{eq:def:MlatOracle} and using submultiplicativity,
  \begin{equation*}
    \left\| (\Xhat\Q^\top-\Xpop)^\top \MlatOracle \ZpeerOracle \right\|
    \le \frac{ \kappa( \HlatOracle ) \left\| (\Xhat-\Xpop)^\top \HlatOracle \right\| \| \ZpeerOracle \| }
    { \sigmamin( \HlatOracle ) }
    = \op{ \| \ZpeerOracle \| },
  \end{equation*}
  where we have used Lemma~\ref{lem:XhatXTB:ctrl} with $\mB = \HlatOracle$ along with the growth rates in Equations~\eqref{eq:assum:spectralConc}, \eqref{eq:assum:XandW:spectrumUB}, \eqref{eq:assum:XhatX:rate:2}, and~\eqref{eq:assum:Hgrowth}.
  Applying Lemma~\ref{lem:ZpeerZlat:spectral}, Weyl's inequality and our growth assumption in Equation~\eqref{eq:assum:Zgrowth},
  \begin{equation*}
    \left\| (\Xhat\Q^\top-\Xpop)^\top \MlatOracle \ZpeerOracle \right\|
    = \op{ \sqrt{n} } .
  \end{equation*}
  Applying this to Equation~\eqref{eq:peercon:latfit:Zcovs:linearTerm:prelim},
  \begin{equation*}
    \left\| (\Zhat\Qdestop-\ZpeerOracle)^\top \MpeerOracle \ZpeerOracle \right\|
    \le
    \left\| \ZpeerOracle^\top \MlatOracle (\Ghat-\G) \Y \right\| + \op{ \sqrt{n} } .
  \end{equation*}
  Again using the definition of $\MlatOracle$ and submultiplicativity,
  \begin{equation*}
    \left\| (\Zhat\Qdestop-\ZpeerOracle)^\top \MpeerOracle \ZpeerOracle \right\|
    \le
    \frac{ \left\| \ZpeerOracle \right\| \kappa( \HlatOracle ) }{ \sigmamin(\HlatOracle) }
    \left\| \HlatOracle^\top (\Ghat-\G) \Y \right\| + \op{ \sqrt{n} }
    \le C  \left\| \HlatOracle^\top (\Ghat-\G) \Y \right\| + \op{ \sqrt{n} } ,
  \end{equation*}
  Applying Lemma~\ref{lem:Y:peerquad} with $\vu$ equal to each of the columns of $\HlatOracle$ and using our growth assumption in Equation~\eqref{eq:assum:Hgrowth},
  \begin{equation} \label{eq:peercon:latfit:Zcovs:linearTerm:done}
    \left\| (\Zhat\Qdestop-\ZpeerOracle)^\top \MpeerOracle \ZpeerOracle \right\|
    = \op{ \sqrt{n} } .
  \end{equation}

  By submultiplicativity,
  \begin{equation*}
    \left\| \Qdes \Zhat^\top ( \Mhat-\MpeerOracle ) \Zhat \Qdestop \right\|
    \le
    \left\| \Zhat \Qdestop \right\|^2 \left\| \Mhat-\MpeerOracle \right\|
    \le C \left\| \ZlatOracle \right\|^2 \left\| \Mhat-\MpeerOracle \right\|,
  \end{equation*}
  where the second inequality follows from Lemma~\ref{lem:ZhatZlat:spectral}, Weyl's inequality, and the fact that $\sigmamin(\MlatOracle \ZlatOracle) \le \sigmamin( \ZlatOracle )$, since $\MlatOracle$ is a projection.
  Applying Lemma~\ref{lem:MhatMpeer:spectral} and our growth assumptions in Equations~\eqref{eq:assum:Hgrowth} and~\eqref{eq:assum:Zgrowth},
  \begin{equation} \label{eq:peercon:latfit:Zcovs:Mdiffs:done}
    \left\| \Qdes \Zhat^\top ( \Mhat-\MpeerOracle ) \Zhat \Qdestop \right\|
    = \op{ \sqrt{n} } .
  \end{equation}

  Using idempotence of $\MpeerOracle$,
  \begin{equation} \label{eq:peercon:latfit:Zcovs:quad:start}
    \left\| (\Zhat\Qdestop-\ZpeerOracle)^\top \MpeerOracle (\Zhat\Qdestop-\ZpeerOracle) \right\|
    = \left\| \MpeerOracle (\Zhat\Qdestop-\ZpeerOracle) \right\|^2 .
  \end{equation}
  Applying the triangle inequality and submultiplicativity,
  \begin{equation*}
    \left\| \MpeerOracle (\Zhat\Qdestop-\ZpeerOracle) \right\|
    \le \left\| \MlatOracle \left( \Zhat\Qdestop - \ZpeerOracle \right) \right\|
    + \left\| \MpeerOracle-\MlatOracle \right\|
    \left\| \Zhat\Qdestop - \ZpeerOracle \right\|
    \le \left\| \MlatOracle \left( \Zhat\Qdestop - \ZpeerOracle \right) \right\|
    + \op{ 1 },
  \end{equation*}
  where the second inequality follows from Lemmas~\ref{lem:MpeerMlat:spectral} and~\ref{lem:ZhatZpeer:spectral} along with our growth assumptions in Equations~\eqref{eq:assum:Hgrowth} and~\eqref{eq:assum:Zgrowth}.
  Recalling the definition of $\MlatOracle$ from Equation~\eqref{eq:def:MlatOracle} and applying submultiplicativity,
  \begin{equation*}
    \left\| \MpeerOracle (\Zhat\Qdestop-\ZpeerOracle) \right\|
    \le
    \frac{ \kappa(\HlatOracle)
      \left\| \HlatOracle^\top (\Zhat\Qdestop-\ZpeerOracle) \right\| }
    { \sigmamin( \HpeerOracle ) }
    + \op{ 1 } .
  \end{equation*}
  An argument parallel to that leading up to Equation~\eqref{eq:peercon:latfit:Zcovs:linearTerm:done} yields
  \begin{equation*}
    \left\| \HlatOracle^\top (\Zhat\Qdestop-\ZpeerOracle) \right\|
    = \op{ \sqrt{n} },
  \end{equation*}
  from which our growth assumption in Equation~\eqref{eq:assum:Hgrowth} yield
  \begin{equation*}
    \left\| \MpeerOracle (\Zhat\Qdestop-\ZpeerOracle) \right\| = \op{ 1 } .
  \end{equation*}
  Applying this to Equation~\eqref{eq:peercon:latfit:Zcovs:quad:start},
  \begin{equation} \label{eq:peercon:latfit:Zcovs:quad:done}
    \left\| (\Zhat\Qdestop-\ZpeerOracle)^\top \MpeerOracle 
		(\Zhat\Qdestop-\ZpeerOracle) \right\|
    = \op{ 1 } .
  \end{equation}

  Applying Equations~\eqref{eq:peercon:latfit:Zcovs:linearTerm:done},
  ~\eqref{eq:peercon:latfit:Zcovs:Mdiffs:done}
  and~\eqref{eq:peercon:latfit:Zcovs:quad:done}
  to Equation~\eqref{eq:peercon:latfit:Zcovs:numertri},
  \begin{equation*}
    \left\| \Qdes \Zcheck^\top \Mcheck \Zcheck \Qdestop
    - \ZpeerOracle^\top \MpeerOracle \ZpeerOracle \right\|
    = \op{ \sqrt{n} } .
  \end{equation*}
  Applying this to Equation~\eqref{eq:peercon:latfit:Zcovs:start},
  \begin{equation} \label{eq:peercon:latfit:Zcovs:hat:done}
    \left\| \left[ \Qdes (\Zhat^\top \Mhat \Zhat)^{-1} \Qdestop
      - (\ZpeerOracle^\top \MpeerOracle \ZpeerOracle)^{-1} \right]
    \left( \MpeerOracle \ZpeerOracle \right)^\top \Y \right\|
    = \op{ \frac{1}{\sqrt{n}} }.
  \end{equation}
  Applying this and Equation~\eqref{eq:peercon:latfit:MZdiffs:done} to Equation~\eqref{eq:thetahat2betaoracle} and multiplying through by $\sqrt{n}$ completes the proof.
\end{proof}

\section{Projection Equivalence} \label{apx:equivalence}

Here, we prove Theorem~\ref{thm:proj-equivalence}, which concerns asymptotic equivalence of the projection parameters defined in Equation~\eqref{eq:def:projectionParams}.
We handle the two different model settings of Theorem~\ref{thm:proj-equivalence} in two lemmas below.
The latent contagion result, in which the responses are generated as in Equation~\eqref{eq:lim-latent}, is established in Lemma~\ref{lem:proj-equiv:latent}.
The peer contagion result, in which the responses are generated as in Equation~\eqref{eq:lim-peer}, is given in Lemma~\ref{lem:proj-equiv:peer}.
We note that our results on (approximate) projection equivalence do not require the same assumptions as our first four Theorems, and instead hold under slightly weaker assumptions.
For example, instead of the subgamma edge behavior of Definition~\ref{def:subgamma-network}, we require only that the edge noise has bounded second moment.

\begin{lemma} \label{lem:proj-equiv:GGtilde:quaddiff:vbounded:nonasy}
Suppose that $\vv \in \R^n$ is such that $\mA-\mApop$ is independent of $\vv$ conditional on $\Xpop$.
Then
\begin{equation*}
\bbE \left\| (\G-\Gtilde) \vv \right\|^2
\le 
C n \nu_n \sum_{i=1}^n 
	\bbE \frac{ \left\| \vv \right\|_\infty^2 }{ \dtilde_i^2 }.
\end{equation*}
\end{lemma}
\begin{proof}
Applying the triangle inequality,
\begin{equation} \label{eq:proj-equiv:GGtilde:quaddiff:tri}
\bbE \left\| \left( \mG - \Gtilde \right) \vv \right\|^2
\le C \bbE \left\|\mDtilde^{-1} \left(\mA-\mApop\right) \vv \right\|^2
+ C \bbE \left\|\left(\D^{-1} - \mDtilde^{-1}\right)\mA \vv \right\|^2 .
\end{equation}
Expanding the norm,
\begin{equation*}
\bbE \left\| \mDtilde^{-1} \! \left(\mA\!-\!\mApop\right) \vv \right\|^2
= \sum_{i=1}^n 
\bbE \left[ \frac{1}{\dtilde_i}\sum_{j=1}^n (\mA\!-\!\mApop)_{ij} v_j \right]^2
= \sum_{i=1}^n \bbE \frac{1}{\dtilde_i^2}
		\left[ \sum_{j=1}^n (\mA\!-\!\mApop)_{ij} v_j \right]^2 .
\end{equation*}
Since the entries of $\mA-\mApop$ are independent (up to symmetry) and mean zero conditional on the latent positions, expanding the square yields
\begin{equation} \label{eq:proj-equiv:GGtilde:quaddiff:term1:done}
\bbE \left\| \mDtilde^{-1} \left(\mA-\mApop\right) \vv \right\|^2
= \sum_{i=1}^n \sum_{j=1}^n
		\bbE \frac{1}{\dtilde_i^2} (\mA\!-\!\mApop)_{ij}^2 v_j^2 
\le C \nu_n \sum_{i=1}^n \sum_{j=1}^n
	\bbE \frac{ v_j^2 }{\dtilde_i^2}
= C \nu_n \bbE \| \vv \|^2 \sum_{i=1}^n \frac{1}{\dtilde_i^2} .
\end{equation}

Factorizing and using the fact that $\mG=\D^{-1}\mA$ is a transition matrix,
\begin{equation*} \begin{aligned}
\bbE \left\| \left( \D^{-1} - \mDtilde^{-1}\right)\mA \vv \right\|^2
&= \bbE \left\| \mDtilde^{-1}\left(\mDtilde-\D\right) \mG \vv \right\|^2
= \sum_{i=1}^n \bbE \left( \frac{d_i - \dtilde_i}{\dtilde_i} \right)^2
			\left( \mG \vv \right)_i^2 \\
&\le C \sum_{i=1}^n \bbE \left( \frac{d_i - \dtilde_i}{\dtilde_i} \right)^2
			\left\| \vv \right\|_\infty^2 .
		\end{aligned} \end{equation*}
Using conditional independence of the entries of $\mA-\mApop$ along with Equation~\eqref{eq:assum:edgeVariance},
\begin{equation*}
\bbE \left\| \left( \D^{-1} - \mDtilde^{-1}\right)\mA \vv \right\|^2
\le C n \nu_n \sum_{i=1}^n \bbE \frac{ \left\| \vv \right\|_\infty^2 }
			{\dtilde_i^2}.
\end{equation*}
Applying this and Equation~\eqref{eq:proj-equiv:GGtilde:quaddiff:term1:done} to Equation~\eqref{eq:proj-equiv:GGtilde:quaddiff:tri} and using the trivial upper bound $\| \vv \|^2 \le n \| \vv \|_\infty^2$ completes the proof.
\end{proof}

\begin{lemma} \label{lem:proj-equiv:GGtilde:quaddiff:meanzero:nonasy}
Suppose that $\vv \in \R^n$ has uncorrelated, mean-zero entries with second moments bounded by $\bbE v_i^2 \le C \sigmavv^2$, and that $\mA-\mApop$ is independent of $\vv$ conditional on $\Xpop$.
Then
\begin{equation*}
\bbE \left\| (\G-\Gtilde) \vv \right\|^2
\le C \sigmavv^2 n \nu_n \sum_{i=1}^n \bbE \frac{1}{\dtilde_i^2} .
\end{equation*}
\end{lemma}
\begin{proof}
Expanding the square and using the fact that $\vv$ has mean-zero uncorrelated entries,
\begin{equation*} \begin{aligned}
\bbE \left\| (\G-\Gtilde) \vv \right\|^2
&=
\bbE \vv^\top (\G-\Gtilde)^\top (\G-\Gtilde) \vv
= \sum_{i=1}^n 
	\bbE \left[ (\G-\Gtilde)^\top (\G-\Gtilde) \right]_{ii} v_i^2 \\
&\le C \sigmavv^2 \sum_{i=1}^n \sum_{j=1}^n \bbE (\G-\Gtilde)_{ji}^2 ,
\end{aligned} \end{equation*}
where the inequality follows from our assumption that the entries of $\vv$ have bounded second moments.
It follows that, by the triangle inequality,
\begin{equation} \label{eq:proj-equiv:latcon:Gdiffeps:tri}
\bbE \left\| (\G-\Gtilde) \vv \right\|^2
\le 
C \sigmavv^2
\left[
\bbE \left\| \mDtilde^{-1} \left( \mA - \mApop \right) \right\|_F^2
        + 
\bbE \left\| \left(\D^{-1}-\mDtilde^{-1} \right) \mA \right\|_F^2
\right] .
\end{equation}

Expanding the norm, using conditional independence of the edges and Equation~\eqref{eq:assum:edgeVariance},
\begin{equation} \label{eq:proj-equiv:latcon:Gdiffeps:tri:term1:done}
\bbE \left\| \mDtilde^{-1} \left( \mA - \mApop \right) \right\|_F^2
= \sum_{i=1}^n \sum_{j=1}^n \bbE \frac{1}{\dtilde_i^2}
        \bbE\left[ \left( \mA - \mApop \right)_{i,j}^2 \mid \Xpop \right]
\le C n \nu_n \sum_{i=1}^n \bbE \frac{1}{\dtilde_i^2} .
\end{equation}

Using the fact that the rows of $\G = \D^{-1} \A$ sum to $1$ and all entries are between $0$ and $1$,
\begin{equation*}
\bbE \left\| \left(\D^{-1}-\mDtilde^{-1} \right) \mA \right\|_F^2
= \bbE \left\| \left(\mI-\mD \mDtilde^{-1} \right) \mG \right\|_F^2
= \sum_{i=1}^n \sum_{j=1}^n
\bbE \frac{(\dtilde_i-d_i)^2}{\dtilde_i^2} G_{ij}^2 
\le \sum_{i=1}^n \bbE \frac{(\dtilde_i-d_i)^2}{\dtilde_i^2} .
\end{equation*}
By conditional independence of the edges along with Equation~\eqref{eq:assum:edgeVariance},
\begin{equation*}
        \bbE \left\| \left(\D^{-1}-\mDtilde^{-1} \right) \mA \right\|_F^2
        \le C n \nu_n \sum_{i=1}^n \bbE \frac{1}{\dtilde_i^2} .
\end{equation*}
        Applying this and Equation~\eqref{eq:proj-equiv:latcon:Gdiffeps:tri:term1:done} to Equation~\eqref{eq:proj-equiv:latcon:Gdiffeps:tri} completes the proof.
\end{proof}

\begin{lemma} \label{lem:proj-equiv:GGtilde:quaddiff:transpose:nonasy}
Suppose that $\vv \in \R^n$ is such that $\mA-\mApop$ is independent of $\vv$ conditional on $\Xpop$.
Then
\begin{equation*}
\bbE \left\| (\G-\Gtilde)^\top \vv \right\|^2
\le 
C n \nu_n \sum_{i=1}^n \bbE \frac{ v_i^2 }{ \dtilde_i^2 }.
\end{equation*}
\end{lemma}
\begin{proof}
Applying the triangle inequality,
\begin{equation} \label{eq:proj-equiv:GGtilde:quaddiff:transpose:tri}
\bbE \left\| \left( \mG - \Gtilde \right)^\top \vv \right\|^2
\le 
C \bbE \left\|\left(\mA-\mApop\right)^\top \mDtilde^{-1} \vv \right\|^2
+ C \bbE \left\|\mA^\top \left(\D^{-1} - \mDtilde^{-1}\right) \vv \right\|^2 .
\end{equation}
Expanding the norm and taking advantage of the conditional edge independence structure of $\mA-\mApop$,
\begin{equation*}
\bbE \left\| \left(\mA\!-\!\mApop\right)^\top \! \mDtilde^{-1} \vv \right\|^2
= \sum_{i=1}^n
\bbE \left[ \sum_{j=1}^n \frac{(\mA\!-\!\mApop)_{ij}v_j }{\dtilde_j} \right]^2
= \sum_{i=1}^n \sum_{j=1}^n \bbE \frac{(\mA\!-\!\mApop)_{ij}^2 v_j^2 }
	{\dtilde_j^2} .
\end{equation*}
Applying Equation~\eqref{eq:assum:edgeVariance},
\begin{equation} \label{eq:proj-equiv:GGtilde:quaddiff:transpose:term1:done}
\bbE \left\| \left(\mA\!-\!\mApop\right)^\top \! \mDtilde^{-1} \vv \right\|^2
\le
C n \nu_n \sum_{j=1}^n \bbE \frac{ v_j^2 }{\dtilde_j^2} .
\end{equation}

Factorizing and using the fact that $\mG=\D^{-1}\mA$ is a transition matrix,
\begin{equation*}
\bbE \left\| \mA^\top \left( \D^{-1} - \mDtilde^{-1}\right) \vv \right\|^2
= \bbE \left\| \mG^\top \left(\mDtilde-\D\right) \mDtilde^{-1} \vv \right\|^2
\le \bbE \left\| \left(\mDtilde-\D\right) \mDtilde^{-1} \vv \right\|^2,
\end{equation*}
where the inequality follows from submultiplicativity and the fact that $\|\mG\| \le 1$.
Using conditional independence of the entries of $\mA-\mApop$ along with Equation~\eqref{eq:assum:edgeVariance},
\begin{equation*}
\bbE \left\| \mA^\top \left( \D^{-1} - \mDtilde^{-1}\right) \vv \right\|^2
\le \sum_{i=1}^n \bbE \frac{ (\dtilde_i-d_i)^2 v_i^2 }{ \dtilde_i^2 }
\le C n \nu_n \sum_{i=1}^n \frac{ v_i^2 }{ \dtilde_i^2 }.
\end{equation*}
Applying this and Equation~\eqref{eq:proj-equiv:GGtilde:quaddiff:transpose:term1:done} to Equation~\eqref{eq:proj-equiv:GGtilde:quaddiff:transpose:tri} completes the proof.
\end{proof}

\begin{lemma} \label{lem:proj-equiv:GGtilde:lineardiff:bounded:nonasy}
Let $\vu,\vv \in \R^n$ be such that $\mA-\mApop$ is independent of $\vu$ and $\vv$ conditional on $\Xpop$.
Then
\begin{equation*}
\left| \bbE \vu^\top (\G-\Gtilde)^\top \vv \right|
\le C n \nu_n \sum_{i=1}^n \bbE
	\frac{\left\| \vu \right\|_\infty \left\| \vv \right\|_\infty}
		{\dtilde_i^2} .
\end{equation*}
\end{lemma}
\begin{proof}
Recalling $\G=\D^{-1} \mA$ and $\Gtilde = \Dtilde^{-1} \mApop$, 
\begin{equation*} \begin{aligned}
\bbE \vu^\top \left( \G - \Gtilde \right)^\top \vv
&= \bbE \vu^\top \mA \left( \D^{-1} - \Dtilde^{-1} \right) \vv
 + \bbE \vu^\top \left( \mA - \mApop \right) \Dtilde^{-1} \vv \\
& = \bbE \vb^\top \mA \left( \D^{-1} - \Dtilde^{-1} \right) \vv,
\end{aligned} \end{equation*}
where the second equality follows from our conditional independence and mean zero assumptions on $\mA-\mApop$.
Writing
\begin{equation*}
\mA \left( \D^{-1} - \Dtilde^{-1} \right)
= \mA \D^{-1} \left( \mI - \D \Dtilde^{-1} \right)
= \mG^\top \left( \mI - \D \Dtilde^{-1} \right),
\end{equation*}
it follows that
\begin{equation} \label{eq:proj-equiv:GGtilde:linearterm:nonasy:introGtilde} 
\begin{aligned}
\bbE \vu^\top \left( \G - \Gtilde \right)^\top \vv
&= \bbE \vu^\top \G^\top \left( \mI - \D \Dtilde^{-1} \right) \vv \\
&= \bbE \vu^\top \left( \mG - \Gtilde \right)^\top
	\left( \mI - \D \Dtilde^{-1} \right) \vv,
\end{aligned} \end{equation}
where the second equality follows from the fact that
\begin{equation*}
\bbE \vu^\top \Gtilde^\top \left( \mI - \D \Dtilde^{-1} \right) \vv
= 0,
\end{equation*}
again because $\mA-\mApop$ (and thus also $\mD-\mDtilde$) is mean zero and independent of $\vu$ and $\vv$ conditional on $\Xpop$.
Adding and subtracting appropriate quantities in Equation~\eqref{eq:proj-equiv:GGtilde:linearterm:nonasy:introGtilde} yields
\begin{equation} \label{eq:proj-equiv:linearterm:nonasy:lastSplit}
\begin{aligned}
\bbE \vu^\top \! \left( \G \!-\! \Gtilde \right)^\top \! \vv
&= \bbE \vu^\top \! \left( \mA - \mApop \right) \Dtilde^{-1}
	\! \left( \mI \!-\! \D \Dtilde^{-1} \!\right) \vv
 + \bbE \vu^\top \G^\top \! \left( \mI - \D \Dtilde^{-1} \right)^2 \vv .
\end{aligned} \end{equation}

Expanding the matrix-vector products and applying the triangle inequality,
\begin{equation*} \begin{aligned}
\left| \bbE \vu^\top \mG^\top \left( \mI - \D \Dtilde^{-1} \right)^2 \vv
	\right|
&\le
\sum_{i=1}^n \bbE \left|
	\frac{ (d_i-\dtilde_i)^2 (\mG \vu)_i v_i }{ \dtilde_i^2 } \right|
&\le \sum_{i=1}^n \bbE \left\| \vu \right\|_\infty \left\| \vv \right\|_\infty
	\frac{ (d_i - \dtilde_i)^2 }{ \dtilde_i^2 } ,
\end{aligned} \end{equation*}
where the second inequality follows from the fact that $\G$ is a transition matrix.
By Assumption~\ref{assum:equiv:edges},
\begin{equation} \label{eq:proj-equiv:linearterm:nonasy:ID2:done}
\bbE \vu^\top \mG^\top \left( \mI - \D \Dtilde^{-1} \right)^2 \vv
\le
C n \nu_n \sum_{i=1}^n \bbE
\frac{\left\| \vu \right\|_\infty \left\| \vv \right\|_\infty}{\dtilde_i^2} .
\end{equation}

Expanding the matrix-vector products,
\begin{equation*}
\bbE \vu^\top \left( \mA - \mApop \right) \Dtilde^{-1}
	\left( \mI - \D \Dtilde^{-1} \right) \vv
= \sum_{i=1}^n \sum_{j=1}^n
        \bbE
        \frac{ (d_i-\dtilde_i) u_i }{ \dtilde_i^2 }
        \left(\mA - \mApop \right)_{ij} v_j .
\end{equation*}
Expanding $d_i-\dtilde_i = \sum_j (\mA-\mApop)_{ij}$ and applying our conditional edge independence assumptions,
\begin{equation*} 
\left| \bbE \vu^\top \left( \mA - \mApop \right) \Dtilde^{-1}
        \left( \mI - \D \Dtilde^{-1} \right) \vv \right|
= \left| \sum_{i=1}^n \sum_{j=1}^n \bbE \frac{ u_i v_j }{ \dtilde_i^2 }
			\left(\mA \!-\! \mApop \right)_{ij}^2 \right|.
\end{equation*}
By the triangle inequality and our assumption in Equation~\eqref{eq:assum:edgeVariance},
\begin{equation} \label{eq:proj-equiv:linearterm:nonasy:AP:done}
\left| \bbE \vu^\top \left( \mA - \mApop \right) \Dtilde^{-1}
        \left( \mI - \D \Dtilde^{-1} \right) \vv \right|
\le C n \nu_n  \sum_{i=1}^n
\bbE \frac{ \left\| \vu \right\|_\infty \left\| \vv \right\|_\infty }
		{ \dtilde_i^2 } .
\end{equation}
Applying the triangle inequality to Equation~\eqref{eq:proj-equiv:linearterm:nonasy:lastSplit} followed by Equations~\eqref{eq:proj-equiv:linearterm:nonasy:ID2:done} and~\eqref{eq:proj-equiv:linearterm:nonasy:AP:done} completes the proof.
\end{proof}

\begin{lemma} \label{lem:proj-equiv:GGtilde:lineardiff:meanzero:nonasy}
Let $\vu,\vv \in \R^n$ be such that $\mA-\mApop$ is independent of $\vu$ and $\vv$ conditional on $\Xpop$.
Suppose further that $\max_i \bbE \left| u_i v_i \right| \le C \sigmauv^2$ and that for all $i \neq j$, $\bbE u_i v_j = 0$.
Then
\begin{equation*}
\left| \bbE \vu^\top (\G-\Gtilde)^\top \vv \right|
\le C \sigmauv^2 n \nu_n \sum_{i=1}^n \bbE \frac{1}{\dtilde_i^2} .
\end{equation*}
\end{lemma}
\begin{proof}
Recalling the definitions of $\G$ and $\Gtilde$,
\begin{equation*} \begin{aligned}
\bbE \vu^\top (\G-\Gtilde)^\top \Gtilde \vv 
&= \bbE \vu^\top (\mA-\mApop)^\top \Dtilde^{-1} \vv
+ \bbE \vu^\top \mA^\top \left(\D^{-1}- \Dtilde^{-1}\right) \vv \\
&= \bbE \vu^\top \mA^\top \left(\D^{-1}- \Dtilde^{-1}\right) \vv,
\end{aligned} \end{equation*}
where the second equality follows from the fact that $\mA-\mApop$ is mean zero and independent of $\vu,\vv$ conditional on $\Xpop$.
Factorizing, then adding and subtracting appropriate quantities,
\begin{equation*} \begin{aligned}
\bbE \vu^{\top} \!(\G\!-\!\Gtilde)^{\!\top} \vv 
&= \bbE \vu^\top \G^\top 
	\left(\mI - \D \mDtilde^{-1} \right) \vv \\
&= 
\bbE \vu^\top \left(\G-\Gtilde\right)^\top \!
	\left(\mI - \D \mDtilde^{-1} \right) \vv 
+ \bbE \vu^\top \Gtilde^\top 
	\left(\mI - \D \mDtilde^{-1} \right) \vv \\
&=
\bbE \vu^\top \left(\G-\Gtilde\right)^\top 
	\left(\mI - \D \mDtilde^{-1} \right) \vv ,
\end{aligned} \end{equation*}
where the final equality again follows from the fact that $\mA-\mApop$ is mean zero conditional on $\Xpop$, and thus $\D-\Dtilde$ is, too.
Again expanding
\begin{equation*}
\G-\Gtilde = \left(\D^{-1}-\Dtilde^{-1}\right)\mA
		+ \Dtilde^{-1} \left(\mA-\mApop \right)
\end{equation*}
and writing $\G = \D^{-1} \mA$, 
\begin{equation} 
\label{eq:proj-equiv:GGtilde:lineardiff:meanzero:nonasy:almostThere}
\begin{aligned}
\left| \bbE \vu^{\top} \!(\G\!-\!\Gtilde)^{\!\top} \vv \right|
&\le \left| \bbE \vu^\top \mG^\top \!
	\left(\mI - \D \mDtilde^{-1} \right)^2 \vv \right| \\
&~~~~~~+ 
\left| \bbE \vu^\top \left(\mA-\mApop\right)^\top \mDtilde^{-1}
		\left(\mI - \D \mDtilde^{-1} \right) \vv \right| .
\end{aligned} \end{equation}

Expanding the quadratic form and using our assumption that $\bbE u_i v_j = 0$ for $i \neq j$,
\begin{equation*} \begin{aligned}
\bbE \vu^\top \left(\mA-\mApop\right)^\top \mDtilde^{-1}
		\left(\mI - \D \mDtilde^{-1} \right) \vv
&=
\sum_{j=1}^n \bbE \left[ \left(\mA-\mApop\right)^\top \mDtilde^{-1}
                \left(\mI - \D \mDtilde^{-1} \right) \right]_{jj}
		u_j v_j \\
&=
\sum_{j=1}^n 
\bbE \left(\mA-\mApop\right)_{jj}
		\frac{\dtilde_j-d_j}{\dtilde_j^2} u_j v_j .
\end{aligned} \end{equation*}
Using the (conditional) edge independence structure of $\mA-\mApop$,
\begin{equation*} \begin{aligned}
\left| \bbE \vu^\top \left(\mA-\mApop\right)^\top \mDtilde^{-1}
		\left(\mI - \D \mDtilde^{-1} \right) \vv \right|
&= \left| \sum_{j=1}^n 
\bbE \frac{ \left(\mA-\mApop\right)_{jj}^2 }{ \dtilde_j^2}
		u_j v_j \right| \\
&\le C \nu_n \sum_{j=1}^n
\bbE \frac{ \left| u_j v_j \right| }{ \dtilde_j^2},
\end{aligned} \end{equation*}
where the inequality follows from the triangle inequality and Equation~\eqref{eq:assum:edgeVariance}.
Using our assumption bounding $\bbE |u_j v_j|$, 
\begin{equation}
\label{eq:proj-equiv:GGtilde:lineardiff:meanzero:nonasy:step1}
\bbE \vu^\top \left(\mA-\mApop\right)^\top \mDtilde^{-1}
		\left(\mI - \D \mDtilde^{-1} \right) \vv
\le C \sigmauv^2 \nu_n \sum_{i=1}^n \bbE \frac{ 1 }{ \dtilde_i^2}.
\end{equation}

Again expanding the quadratic form and rearranging,
\begin{equation*} \begin{aligned}
\left| \bbE \vu^\top \mG^\top 
	\left(\mI - \D \mDtilde^{-1} \right)^2 \vv \right|
&=
\left| \sum_{j=1}^n \sum_{i=1}^n
\bbE \frac{(d_i-\dtilde_i)^2}{\dtilde_i^2} G_{ij} u_j v_j 
\right| \\
&\le
C n \nu_n \sum_{j=1}^n \sum_{i=1}^n
\bbE \frac{1}{\dtilde_i^2} G_{ij} \left| u_j v_j \right|,
\end{aligned} \end{equation*}
where we have used Equation~\eqref{eq:assum:edgeVariance}.
Using the fact that $\G$ is a transition matrix along with our assumption bounding $\bbE |u_j v_j|$, it follows that
\begin{equation} 
\label{eq:proj-equiv:GGtilde:lineardiff:meanzero:nonasy:step2}
\left| \bbE \vu^\top \mG^\top 
	\left(\mI - \D \mDtilde^{-1} \right)^2 \vv \right|
\le
C \sigmauv^2 n \nu_n \sum_{i=1}^n \bbE \frac{1}{\dtilde_i^2} .
\end{equation}
Applying Equations~\eqref{eq:proj-equiv:GGtilde:lineardiff:meanzero:nonasy:step1}
and~\eqref{eq:proj-equiv:GGtilde:lineardiff:meanzero:nonasy:step2}
to Equation~\eqref{eq:proj-equiv:GGtilde:lineardiff:meanzero:nonasy:almostThere}
completes the proof.
\end{proof}

\subsection{Projection equivalence under latent contagion}

Here we prove projection equivalence when the true data generating model is latent contagion, as described in Equation~\eqref{eq:lim-latent-red}.
For ease of notation, define the sample covariance matrices
\begin{equation} \label{eq:def:Sigmamxs}
	\SiglatOracle  = \frac{1}{n} \ZlatOracle^\top \ZlatOracle
	~\text{ and }~
	\SigpeerOracle = \frac{1}{n} \ZpeerOracle^\top \ZpeerOracle.
\end{equation}
as well as
\begin{equation} \label{eq:def:XitildeL}
	\mXitilde = \left( \mI - \theta_y \Gtilde \right)^{-1} \in \R^{n \times n}
	~\text{ and }~
	\mLtilde = \1_n \thetanaught + \W \thetaw + \Xpop \thetax.
\end{equation}
We observe that under the latent contagion model of Equation~\eqref{eq:lim-latent-red}, we have
\begin{equation} \label{eq:Y:XitildeL:latent}
\Y = \mXitilde\left(\mLtilde + \be \right) .
\end{equation}

\begin{lemma} \label{lem:proj-equiv:latent}
	Under the latent contagion model in Equation~\eqref{eq:lim-latent-red}, suppose that Assumptions~\ref{assum:latent-oracle},~\ref{assum:equiv:limcov},~\ref{assum:equiv:edges} and~\ref{assum:equiv:LandX} hold.
	Then
	\begin{equation*}
	\norm*{ \left[ \bbE \ZlatOracle^\top \ZlatOracle \right]^{-1} \bbE \brac*{\ZlatOracle^\top \Y} - \left[ \bbE \ZpeerOracle^\top \ZpeerOracle \right]^{-1} \bbE \brac*{\ZpeerOracle^\top \Y} }
	= o( n^{-1/2} ).
	\end{equation*} 
\end{lemma}

In preparation for a proof of Lemma~\ref{lem:proj-equiv:latent}, we establish a handful of technical results.

\begin{lemma} \label{lem:proj-equiv:latcon:linearTerm}
	Under the conditions of Lemma~\ref{lem:proj-equiv:latent},
	\begin{equation*}
		\frac{1}{n} \left\| \bbE \ZlatOracle^\top \Y \right\|
		\le C .
	\end{equation*}
\end{lemma}
\begin{proof}
	Recalling the structure of $\ZlatOracle$ from Equation~\eqref{eq:def:ZlatOracle},
	\begin{equation*} \begin{aligned}
	\left\| \frac{1}{n} \bbE \brac*{\ZlatOracle^\top \Y} \right\|
	 & \le
	\frac{1}{n} \left\| \bbE \1_n^\top \Y \right\|
	+
	\frac{1}{n} \left\| \bbE \W^\top \Y \right\|
	+
	\frac{1}{n} \left\| \bbE \Xpop^\top \Y \right\|
	+
	\frac{1}{n} \left\| \bbE \Y^\top \Gtilde^\top \Y \right\| .
	\end{aligned} \end{equation*}
	Expanding $\Y = \mXitilde\left(\mLtilde + \be \right)$ as in Equation~\eqref{eq:Y:XitildeL:latent} and using the fact that $\be$ is conditionally independent of $\mA$ conditional on $\Xpop$ under the latent contagion model, the moment bounds in Assumption~\ref{assum:latent-oracle} imply
	\begin{equation*} \begin{aligned}
	\left\| \frac{1}{n} \bbE \brac*{\ZlatOracle^\top \Y} \right\|
	& \le C +
	\frac{1}{n} \left\| \bbE \Y^\top \Gtilde^\top \Y \right\| .
	\end{aligned} \end{equation*}
	Again recalling the decomposition of $\Y$ and using the independence structure of the latent contagion model,
	\begin{equation*} \begin{aligned}
	\left\| \frac{1}{n} \bbE \brac*{\ZlatOracle^\top \Y} \right\|
	 & \le
	C
	+
	\frac{1}{n}
	\bbE \left\| \mLtilde^\top \mXitilde^\top \Gtilde^\top \mXitilde \mLtilde\right\|
	+
	\frac{1}{n}
	\bbE \left\| \be^\top \mXitilde^\top \Gtilde^\top \mXitilde \be \right\| \\
	 & \le
	C
	+ \frac{C}{n} \bbE \left\| \mLtilde \right\|^2
	+ \frac{C}{n} \bbE \left\| \be \right\|^2,
	\end{aligned} \end{equation*}
	where the second inequality follows from submultiplicativity and Lemma~\ref{lem:IbgyG:invertible} (note that while this result is stated for $\G$, it holds for $\Gtilde$, as well).
	Items~\ref{item:assum:oracle:boundedCovariates} and ~\ref{item:assum:oracle:epsilonMoments} in Assumption~\ref{assum:latent-oracle} then imply
	\begin{equation*}
		\left\| \frac{1}{n} \bbE \brac*{\ZlatOracle^\top \Y} \right\|
		\le C,
	\end{equation*}
	as we set out to show.
\end{proof}

\begin{lemma} \label{lem:proj-equiv:latcon:linearTerm:diff}
	Under the conditions of Lemma~\ref{lem:proj-equiv:latent},
	\begin{equation*}
		\left\| \bbE \left( \ZpeerOracle - \ZlatOracle \right)^\top
		\ZlatOracle \right\|
		= o( \sqrt{n} ) .
	\end{equation*}
\end{lemma}
\begin{proof}
	Recalling the structure of $\ZpeerOracle$ and $\ZlatOracle$ from Equations~\eqref{eq:def:ZpeerOracle} and~\eqref{eq:def:ZlatOracle}, respectively,
	\begin{equation*} 
		\left\| \bbE \left( \ZpeerOracle - \ZlatOracle \right)^\top
		\ZlatOracle \right\|
		=
		\left| \bbE \Y^\top \left( \G - \Gtilde \right)^\top
		\Gtilde \Y \right| .
	\end{equation*} 
	Expanding $\Y$ and recalling the notation established in Equation~\eqref{eq:def:XitildeL}, using the fact that $\be$ is mean zero and conditionally independent of $\mA-\mApop$ along with the triangle inequality,
	\begin{equation*} 
	\left\| \bbE \left( \ZpeerOracle - \ZlatOracle \right)^\top
	\ZlatOracle \right\|
	\le
	\left| \bbE \mLtilde^\top \mXitilde^\top
			\left( \G - \Gtilde \right)^\top \Gtilde
			\mXitilde \mLtilde \right|
	+ \left| \bbE \be^\top \mXitilde^\top
			\left( \G - \Gtilde \right)^\top \Gtilde
			\mXitilde \be \right| .
	\end{equation*} 
	Controlling the first right-hand term by
	Lemma~\ref{lem:proj-equiv:GGtilde:lineardiff:bounded:nonasy}
	with $\vu=\vv = \mXitilde \mLtilde$
	and controlling the second term by
	Lemma~\ref{lem:proj-equiv:GGtilde:lineardiff:meanzero:nonasy}
	with $\vu = \vv = \mXitilde \be$,
	\begin{equation*} 
	\left\| \bbE \left( \ZpeerOracle - \ZlatOracle \right)^\top
	\ZlatOracle \right\|
	\le
	C n \nu_n \sum_{i=1}^n \bbE
	\frac{\left\| \mXitilde \mLtilde \right\|_\infty^2
		+ \sigmaeps^2 }
		{\dtilde_i^2} .
	\end{equation*}
	Using
	$\| \mXitilde \mLtilde \|_\infty \le C\| \mLtilde\|_\infty$
	along with Assumptions~\ref{assum:peer-oracle}
	and~\ref{assum:equiv:LandX}
	completes the proof.
\end{proof}

\begin{lemma} \label{lem:proj-equiv:latcon:quadterm:diff}
	Under the conditions of Lemma~\ref{lem:proj-equiv:latent},
	\begin{equation*}
		\left\| \bbE
		\left(\ZpeerOracle - \ZlatOracle\right)^\top
		\left( \ZpeerOracle - \ZlatOracle \right) \right\|
		= o( \sqrt{n} ).
	\end{equation*}
\end{lemma}
\begin{proof}
	Recalling the structure of $\ZpeerOracle$ and $\ZlatOracle$ from Equations~\eqref{eq:def:ZpeerOracle} and~\eqref{eq:def:ZlatOracle}, respectively,
	\begin{equation} \label{eq:proj-equiv:latcon:quadterm:diff:start}
		\left\| \bbE \left(\ZpeerOracle - \ZlatOracle\right)^\top
		\left( \ZpeerOracle - \ZlatOracle \right) \right\|
		= \left| \bbE \Y^\top \left( \mG - \Gtilde \right)^\top
		\left( \mG - \Gtilde \right) \Y \right| .
	\end{equation}
	Under the latent contagion model in Equation~\eqref{eq:lim-latent}, we have $\Y = \mXitilde \left( \mLtilde + \be \right)$, where $\mXitilde$ and $\mLtilde$ are as in Equation~\eqref{eq:def:XitildeL}.
	Since $\be$ is mean zero and independent of $\mA$ conditional on $\Xpop$,
	\begin{equation*}
		\bbE\Y^\top \left( \mG - \Gtilde \right)^\top
		\left( \mG - \Gtilde \right) \Y
		= \bbE \left\| \left( \mG - \Gtilde \right) \mXitilde \mLtilde \right\|^2
		+ \bbE \left\| \left( \mG - \Gtilde \right) \mXitilde \be \right\|^2 ,
	\end{equation*}
	so that applying the triangle inequality to Equation~\eqref{eq:proj-equiv:latcon:quadterm:diff:start} yields
	\begin{equation*} 
	\left\| \bbE \left(\ZpeerOracle - \ZlatOracle\right)^\top
	\left( \ZpeerOracle - \ZlatOracle \right) \right\|
	\le 
	\bbE \left\| \left( \mG - \Gtilde \right) 
			\mXitilde \mLtilde \right\|^2
	+ \bbE \left\| \left( \mG - \Gtilde \right) 
			\mXitilde \be \right\|^2 .
	\end{equation*}
	Using Lemma~\ref{lem:proj-equiv:GGtilde:quaddiff:vbounded:nonasy}
	with $\vv = \mXitilde \mLtilde$ to control the first term
	and using
	Lemma~\ref{lem:proj-equiv:GGtilde:quaddiff:meanzero:nonasy}
	with $\vv = \mXitilde \be$,
	\begin{equation*} 
	\left\| \bbE \left(\ZpeerOracle - \ZlatOracle\right)^\top
	\left( \ZpeerOracle - \ZlatOracle \right) \right\|
	\le 
	C n \nu_n \sum_{i=1}^n \bbE 
	\frac{\| \mXitilde \mLtilde \|_\infty^2 + \sigmaeps^2}{\dtilde_i^2} .
	\end{equation*}
	Using 
	$\| \mXitilde \mLtilde \|_\infty \le C\| \mLtilde\|_\infty$
	along with Assumptions~\ref{assum:peer-oracle}
	and~\ref{assum:equiv:LandX}
	completes the proof.
\end{proof}

\begin{lemma} \label{lem:proj-equiv:latcon:covmxs}
Under the conditions of Lemma~\ref{lem:proj-equiv:latent},
let $\SigpeerOracle$ and $\SiglatOracle$ be as in Equation~\eqref{eq:def:Sigmamxs}.
Then
\begin{equation*}
\left\| \bbE \SigpeerOracle - \bbE \SiglatOracle \right\|
= o( n^{-1/2} ) .
\end{equation*}
\end{lemma}
\begin{proof}
	Adding and subtracting appropriate quantities and applying the triangle inequality,
	\begin{equation*}
		\left\| \bbE \SigpeerOracle - \bbE \SiglatOracle \right\|
		\le
		\frac{2}{n} \left\| \bbE \left( \ZpeerOracle - \ZlatOracle \right)^\top
		\ZlatOracle \right\|
		+ \frac{1}{n} \left\| \bbE \left(\ZpeerOracle - \ZlatOracle\right)^\top \left( \ZpeerOracle - \ZlatOracle \right)
		\right\| .
	\end{equation*}
	Applying Lemmas~\ref{lem:proj-equiv:latcon:linearTerm:diff} and~\ref{lem:proj-equiv:latcon:quadterm:diff} yields the result.
\end{proof}

\begin{lemma} \label{lem:proj-equiv:latent:peerCovmx}
	Under the conditions of Lemma~\ref{lem:proj-equiv:latent},
	$\bbE \SigpeerOracle$ is invertible and
	\begin{equation*}
		\left\| \left[ \bbE \SigpeerOracle \right]^{-1} \right\|
		\le C .
	\end{equation*}
\end{lemma}
\begin{proof}
	By Lemma~\ref{lem:proj-equiv:latcon:covmxs},
	\begin{equation*}
		\left\| \bbE \SigpeerOracle - \bbE \SiglatOracle  \right\|
		= o( n^{-1/2} ) .
	\end{equation*}
	Combining this with our assumption that $\bbE \SiglatOracle$ converges to an invertible matrix completes the proof.
\end{proof}

\begin{lemma} \label{lem:proj-equiv:latcon:precMxs}
	Under the conditions of Lemma~\ref{lem:proj-equiv:latent},
	$\SiglatOracle$ and $\SigpeerOracle$, as defined in Equation~\eqref{eq:def:Sigmamxs},
	are both invertible and obey
	\begin{equation*}
		\left\|
		\left( \bbE \SigpeerOracle \right)^{-1}
		- \left( \bbE \SiglatOracle \right)^{-1}
		\right\|
		= o( n^{-1/2} ) .
	\end{equation*}
\end{lemma}
\begin{proof}
	By Lemma~\ref{lem:proj-equiv:latent:peerCovmx} and Assumption~\ref{assum:equiv:limcov}, both $\bbE \SigpeerOracle$ and $\bbE \SiglatOracle$ are invertible for $n$ suitably large.
	Applying submultiplicativity,
	\begin{equation*} \begin{aligned}
			\left\| \left( \bbE \SigpeerOracle \right)^{-1}
			- \left( \bbE \SiglatOracle \right)^{-1} \right\|
			 & \le
			\left\| \left( \bbE \SigpeerOracle \right)^{-1}  \right\|
			\left\| \left( \bbE \SiglatOracle \right)^{-1} \right\|
			\left\| \bbE \SigpeerOracle - \bbE \SiglatOracle \right\|.
		\end{aligned} \end{equation*}
	Applying Lemma~\ref{lem:proj-equiv:latent:peerCovmx} and Assumption~\ref{assum:equiv:limcov} again, it follows that
	\begin{equation*}
		\left\| \left( \bbE \SigpeerOracle \right)^{-1}
		- \left( \bbE \SiglatOracle \right)^{-1} \right\|
		\le C \left\| \bbE \SigpeerOracle - \bbE \SiglatOracle \right\|,
	\end{equation*}
	and Lemma~\ref{lem:proj-equiv:latcon:covmxs} yields our desired result.
\end{proof}

\begin{proof}[Proof of Lemma~\ref{lem:proj-equiv:latent}]
	Recalling the definitions from Equation~\eqref{eq:def:Sigmamxs},
	\begin{equation*}
		\norm*{\bbE \brac*{\ZlatOracle^\top \ZlatOracle}^{-1} \bbE \brac*{\ZlatOracle^\top \Y} - \bbE \brac*{\ZpeerOracle^\top \ZpeerOracle}^{-1} \bbE \brac*{\ZpeerOracle^\top \Y}}
		=
		\norm*{
			\left[ \bbE \SiglatOracle \right]^{-1}
			\bbE \frac{ \ZlatOracle^\top \Y }{ n }
			- \left[ \bbE \SigpeerOracle \right]^{-1}
			\bbE \frac{ \ZpeerOracle^\top \Y}{n} }.
	\end{equation*}

	By the triangle inequality, it will suffice for us to show that
	\begin{equation} \label{eq:proj-equiv:latent:term1}
		\left\| \left( \bbE \SigpeerOracle \right)^{-1}
		\bbE \frac{ \left( \ZpeerOracle - \ZlatOracle \right)^\top \Y }{n}
		\right\| = o( n^{-1/2} )
	\end{equation}
	and
	\begin{equation} \label{eq:proj-equiv:latent:term2}
		\left\|
		\left[ \left( \bbE \SigpeerOracle \right)^{-1}
			- \left( \bbE \SiglatOracle \right)^{-1} \right]
		\bbE \frac{ \ZlatOracle^\top \Y}{n} \right\|
		= o( n^{-1/2} ) .
	\end{equation}

	To see Equation~\eqref{eq:proj-equiv:latent:term2}, observe that by submultiplicativity,
	\begin{equation*}
		\left\|
		\left[ \left( \bbE \SigpeerOracle \right)^{-1}
			- \left( \bbE \SiglatOracle \right)^{-1} \right]
		\bbE \frac{ \ZlatOracle^\top \Y}{n} \right\|
		\le
		\left\|\left( \bbE \SigpeerOracle \right)^{-1}
		- \left( \bbE \SiglatOracle \right)^{-1} \right\|
		\left\| \bbE \frac{ \ZlatOracle^\top \Y}{n} \right\|,
	\end{equation*}
	from which Lemmas~\ref{lem:proj-equiv:latcon:linearTerm} and~\ref{lem:proj-equiv:latcon:precMxs} yield Equation~\eqref{eq:proj-equiv:latent:term2}.

	To see Equation~\eqref{eq:proj-equiv:latent:term1}, note that submultiplicativity and Lemma~\ref{lem:proj-equiv:latent:peerCovmx} imply
	\begin{equation*}
		\left\| 
		\left( \bbE \SiglatOracle \right)^{-1} 
		\bbE 
		\frac{ \left( \ZpeerOracle - \ZlatOracle \right)^\top \Y }{n}
		\right\|
		\le
		C \left\| \frac{1}{n}
		\bbE \left( \ZpeerOracle - \ZlatOracle \right)^\top \Y
		\right\| .
	\end{equation*}
	An argument nearly identical to that used to prove Lemma~\ref{lem:proj-equiv:latcon:linearTerm:diff} yields Equation~\eqref{eq:proj-equiv:latent:term1}, completing the proof.
\end{proof}

\subsection{Projection equivalence under peer contagion}

A result analogous to Lemma~\ref{lem:proj-equiv:latent} holds under the peer contagion model, as described in Equation~\eqref{eq:lim-peer-red}.
For ease of notation, define
\begin{equation} \label{eq:def:XiL}
	\mXi = \left( \mI - \beta_y \G \right)^{-1} \in \R^{n \times n}
	~\text{ and }~
	\mLtilde = \1_n \thetanaught + \W \thetaw + \Xpop \thetax,
\end{equation}
and observe that under the peer contagion model of Equation~\eqref{eq:lim-peer-red}, we have
\begin{equation} \label{eq:Y:XiL:peer}
\Y = \mXi\left(\mLtilde + \be \right) .
\end{equation}

\begin{lemma} \label{lem:proj-equiv:peer}
	Under the peer contagion model in Equation~\eqref{eq:lim-peer-red}, suppose that Assumptions~\ref{assum:peer-oracle},~\ref{assum:equiv:limcov},~\ref{assum:equiv:edges} and~\ref{assum:equiv:LandX} hold.
	Then
	\begin{equation*}
	\left\| \left( \bbE \ZpeerOracle^\top \ZpeerOracle \right)^{-1} 
				\bbE \ZpeerOracle^\top \Y 
			- \left(\bbE \ZlatOracle^\top \ZlatOracle \right)^{-1} 
				\bbE \ZlatOracle^\top \Y
	\right\|
	= o( n^{-1/2} ) .
	\end{equation*}
\end{lemma}

As with our proof of Lemma~\ref{lem:proj-equiv:latent}, our proof of this result applies a triangle inequality and controls the two resulting terms separately.
Toward this end, we establish a handful of technical results before giving a proof of Lemma~\ref{lem:proj-equiv:peer}.

\begin{lemma} \label{lem:proj-equiv:peercon:linearTerm}
	Under the conditions of Lemma~\ref{lem:proj-equiv:peer},
	\begin{equation*}
		\frac{1}{n} \left\| \bbE \ZpeerOracle^\top \Y \right\|
		\le C .
	\end{equation*}
\end{lemma}
\begin{proof}
	Recalling the structure of $\ZpeerOracle$ from Equation~\eqref{eq:def:ZpeerOracle},
	\begin{equation*} \begin{aligned}
	\left\| \frac{1}{n} \bbE \brac*{\ZpeerOracle^\top \Y} \right\|
	 & \le
	\frac{1}{n} \left\| \bbE \1_n^\top \Y \right\|
	+
	\frac{1}{n} \left\| \bbE \W^\top \Y \right\|
	+
	\frac{1}{n} \left\| \bbE \Xpop^\top \Y \right\|
	+
	\frac{1}{n} \left\| \bbE \Y^\top \G^\top \Y \right\| .
	\end{aligned} \end{equation*}
	Writing $\Y = \mXi\left(\mLtilde + \be \right)$ as in Equation~\eqref{eq:Y:XiL:peer} and using the fact that $\be$ is mean zero,
	\begin{equation*} \begin{aligned}
	\left\| \frac{1}{n} \bbE \brac*{\ZpeerOracle^\top \Y} \right\|
	 & \le
	\frac{1}{n} \left\| \bbE \1_n^\top \mXi \mLtilde \right\|
	+
	\frac{1}{n} \left\| \bbE \W^\top \mXi \mLtilde \right\|
	+
	\frac{1}{n} \left\| \bbE \Xpop^\top \mXi \mLtilde \right\|
	+
	\frac{1}{n} \left\| \bbE \Y^\top \G^\top \Y \right\| .
	\end{aligned} \end{equation*}
	The moment bounds in Assumption~\ref{assum:peer-oracle} then imply
	\begin{equation*} \begin{aligned}
	\left\| \frac{1}{n} \bbE \brac*{\ZpeerOracle^\top \Y} \right\|
	& \le C +
	\frac{1}{n} \left\| \bbE \Y^\top \G^\top \Y \right\| .
	\end{aligned} \end{equation*}
	Again recalling the decomposition of $\Y$ and using the independence structure of the peer contagion model and the fact that $\be$ is mean zero,
	\begin{equation*} \begin{aligned}
	\left\| \frac{1}{n} \bbE \brac*{\ZpeerOracle^\top \Y} \right\|
	& \le
	C
	+
	\frac{1}{n}
	\bbE \left\| \mLtilde^\top \mXi^\top \G^\top \mXi \mLtilde\right\|
	+
	\frac{1}{n}
	\bbE \left\| \be^\top \mXi^\top \G^\top \mXi \be \right\| \\
	& \le
	C
	+ \frac{C}{n} \bbE \left\| \mLtilde \right\|^2
	+ \frac{C}{n} \bbE \left\| \be \right\|^2,
	\end{aligned} \end{equation*}
	where the second inequality follows from submultiplicativity and Lemma~\ref{lem:IbgyG:invertible}.
	Assumption~\ref{assum:peer-oracle} then implies
	\begin{equation*}
	\left\| \frac{1}{n} \bbE \brac*{\ZpeerOracle^\top \Y} \right\|
	\le C,
	\end{equation*}
	as we set out to show.
\end{proof}

Under the peer contagion model, we have $\Y = \mXi (\mLtilde + \be)$,
in which $\mXi$, as defined in Equation~\eqref{eq:def:XiL},
depends on $\mG$.
To account for the random variation of $\mG$ about $\Gtilde$, we will require a handful of technical lemmas before we control
\begin{equation*}
	\left\| \bbE \left( \ZlatOracle - \ZpeerOracle \right)^\top
			\ZpeerOracle \right\|
\end{equation*}
in Lemma~\ref{lem:proj-equiv:peercon:linearTerm:diff}.
We will make repeated use of the fact that for any non-negative integer $s$,
\begin{equation} \label{eq:powerDiff:noncommute}
\G^{s} - \Gtilde^{s}
= \sum_{\ell=0}^{s-1} \G^{s-\ell-1} (\G-\Gtilde) \Gtilde^\ell .
\end{equation}

\begin{lemma} \label{lem:proj-equiv:peercon:qrRecurse:bounded}
Under the setting of Lemma~\ref{lem:proj-equiv:peer},
suppose that $\vb$ is such that $\mA-\mApop$ is independent of $\vb$ conditional on $\Xpop$.
Then, for any non-negative integers $q$ and $r$,
\begin{equation*}
\left| \bbE \left[ \G^q \vb \right]^\top \left( \G - \Gtilde \right)^\top
			\G^r \vb \right|
\le Cq r n \nu_n \sum_{i=1}^n 
    \bbE \frac{ \left\| \vb \right\|_\infty^2 }{ \dtilde_i^2 } ,
\end{equation*}
where $C$ is a constant not depending on $q$ and $r$.
\end{lemma}
\begin{proof}
Using Equation~\eqref{eq:powerDiff:noncommute} with $s = r$,
the triangle inequality yields
\begin{equation} \label{eq:proj-equiv:peer:qr:step1} \begin{aligned}
\left| \bbE \left[ \G^q \vb \right]^\top 
        \left( \Gtilde - \G \right)^\top \G^r \vb \right|
&\le \left| \bbE \left[ \G^q \vb \right]^\top 
		\left( \Gtilde - \G \right)^\top \Gtilde^{r+1} \vb \right| \\
&~~~~~~+
\left| \sum_{\ell=0}^{r-1} 
\bbE \left[ \G^q \vb \right]^\top 
                \left( \Gtilde - \G \right)^\top
		\G^{r-\ell-1} (\G-\Gtilde) \Gtilde^\ell \vb \right| .
\end{aligned} \end{equation}
Using Equation~\eqref{eq:powerDiff:noncommute} with $s=q$,
\begin{equation*} \begin{aligned}
\left|
\bbE \left[ \G^q \vb \right]^\top \!
		\!\left( \Gtilde \!-\! \G \right)^{\!\top} 
		\! \Gtilde^r \vb \right|
&\le 
\left| \bbE \left[ \Gtilde^q \vb \right]^\top 
                \left( \Gtilde - \G \right)^\top \Gtilde^r \vb \right| \\
&~~~~~~+
\left| 
\bbE \left[ \sum_{m=0}^{q-1} \!\G^{q-1-m}(\G-\Gtilde)\Gtilde^m \vb \right]^\top 
		\!\!\!\!
                \left( \Gtilde - \G \right)^{\!\top} \! \Gtilde^r \vb
\right| .
\end{aligned} \end{equation*}
Applying Lemma~\ref{lem:proj-equiv:GGtilde:lineardiff:bounded:nonasy} to the first right-hand term with $\vu =  \Gtilde^q \vb$ and $\vv = \Gtilde^r \vb$,
\begin{equation} \label{eq:proj-equiv:peer:qr:RHterm1} \begin{aligned}
\left|
\bbE \left[ \G^q \vb \right]^\top \!
		\!\left( \Gtilde \!-\! \G \right)^{\!\top} 
		\! \Gtilde^r \vb \right|
&\le C n \nu_n \sum_{i=1}^n \bbE
        \frac{\left\| \Gtilde^q \vb \right\|_\infty 
		\left\| \Gtilde^r \vb \right\|_\infty}{\dtilde_i^2} \\
&~~~~~~+
\left| 
\bbE \left[ \sum_{m=0}^{q-1} \!\G^{q-1-m}(\G-\Gtilde)\Gtilde^m \vb \right]^\top 
		\!\!\!\!
                \left( \Gtilde - \G \right)^{\!\top} \! \Gtilde^r \vb
\right| .
\end{aligned} \end{equation}

Rearranging slightly,
\begin{equation*} \begin{aligned}
&\bbE \left[ \sum_{m=0}^{q-1} \G^{q-1-m}(\G-\Gtilde)\Gtilde^m \vb \right]^\top 
                \left( \Gtilde - \G \right)^\top \Gtilde^r \vb \\
&~~~~~~=
\sum_{m=0}^{q-1} \bbE \left[ (\G-\Gtilde)\Gtilde^m \vb \right]^\top
			\left[ \G^{q-1-m} \right]^\top
			\left( \Gtilde - \G \right)^\top \Gtilde^r \vb,
\end{aligned} \end{equation*}
Applying the triangle inequality followed by Cauchy-Schwarz,
\begin{equation*} \begin{aligned}
&\left|
 \bbE \left[ \sum_{m=0}^{q-1} \G^{q-1-m}(\G-\Gtilde)\Gtilde^m \vb \right]^\top 
                \left( \Gtilde - \G \right)^\top \Gtilde^r \vb \right| \\
&~~~~~~\le
\sum_{m=0}^{q-1}
\sqrt{ \bbE \left\| (\G-\Gtilde)\Gtilde^m \vb \right\|^2 }
\sqrt{ \bbE \left\| \left[ \G^{q-1-m} \right]^\top
		\left( \Gtilde-\G \right)^\top \Gtilde^r \vb \right\|^2 }\\
&~~~~~~\le
\sum_{m=0}^{q-1}
\sqrt{ \bbE \left\| (\G-\Gtilde)\Gtilde^m \vb \right\|^2 }
\sqrt{ \bbE \left\| 
		\left( \Gtilde-\G \right)^\top \Gtilde^r \vb \right\|^2 },
\end{aligned} \end{equation*}
where the second inequality follows from submultiplicativity and the fact that $\|\G\|\le1$.
Applying this bound to Equation~\eqref{eq:proj-equiv:peer:qr:RHterm1},
\begin{equation*} \begin{aligned}
\left|
\bbE \left[ \G^q \vb \right]^\top 
		\left( \Gtilde - \G \right)^\top \Gtilde^r \vb \right|
&\le C n \nu_n \sum_{i=1}^n \bbE
        \frac{\left\| \Gtilde^q \vb \right\|_\infty 
		\left\| \Gtilde^r \vb \right\|_\infty}{\dtilde_i^2} \\
&~~~~~~+
\sum_{m=0}^{q-1}
\sqrt{ \bbE \left\| (\G-\Gtilde)\Gtilde^m \vb \right\|^2 }
\sqrt{ \bbE \left\| 
		\left( \Gtilde-\G \right)^\top \Gtilde^r \vb \right\|^2 }.
\end{aligned} \end{equation*}
Applying Lemma~\ref{lem:proj-equiv:GGtilde:quaddiff:vbounded:nonasy} with $\vv = \Gtilde^r \vb$ and
Lemma~\ref{lem:proj-equiv:GGtilde:quaddiff:transpose:nonasy} with $\vv = \Gtilde^m \vb$ for $m=0,1,2\dots,q-1$, and noting that
\begin{equation*}
( \Gtilde^s \vb )_i^2 \le \| \Gtilde^s \vb \|_\infty^2 \le \| \vb \|_\infty^2
\end{equation*}
for any integer $s \ge 0$, since $\Gtilde$ is an averaging operator,
\begin{equation*} \begin{aligned}
\left|
\bbE \left[ \G^q \vb \right]^\top 
		\left( \Gtilde - \G \right)^\top \Gtilde^r \vb \right|
&\le
C (q+1) n \nu_n \sum_{i=1}^n 
        \bbE \frac{ \left\| \vb \right\|_\infty^2 }{ \dtilde_i^2 }.
\end{aligned} \end{equation*}
Applying this bound to Equation~\eqref{eq:proj-equiv:peer:qr:step1},
\begin{equation} \label{eq:proj-equiv:peer:qr:step2} \begin{aligned}
\left| \bbE \left[ \G^q \vb \right]^\top 
        \left( \Gtilde - \G \right)^\top \G^r \vb \right|
&\le 
C (q+1) n \nu_n \sum_{i=1}^n 
        \bbE \frac{ \left\| \vb \right\|_\infty^2 }{ \dtilde_i^2 } \\
&~~~~~~+
\left| \sum_{\ell=0}^{r-1}
\bbE \left[ \G^q \vb \right]^\top 
                \left( \Gtilde - \G \right)^\top
		\G^{r-\ell-1} (\G-\Gtilde) \Gtilde^\ell \vb \right| .
\end{aligned} \end{equation}

Again using Equation~\eqref{eq:powerDiff:noncommute} with $s=q$ and applying the triangle inequality,
\begin{equation} \label{eq:proj-equiv:GGtilde:ellsum} \begin{aligned}
&\left| \sum_{\ell=0}^{r-1}
\bbE \!\left[ \!\left( \Gtilde \!-\! \G \right) \G^q \vb \right]^\top 
		\! \G^{r-\ell-1} (\G\!-\!\Gtilde) \Gtilde^\ell \vb \right| \\ 
~~~~~~&\le
\sum_{\ell=0}^{r-1}
\left| \bbE \left[ \left( \Gtilde - \G \right) \Gtilde^q \vb \right]^\top 
                \G^{r-\ell-1} (\G-\Gtilde) \Gtilde^\ell \vb \right| \\
&~~~~~~+
\sum_{\ell=0}^{r-1} \sum_{m=0}^{q-1}
\left| \bbE \left[ \G^{q-m-1} \! (\G-\Gtilde) \Gtilde^m \vb \right]^\top 
                \left( \Gtilde \!-\! \G \right)^\top \!
                \G^{r-\ell-1} (\G-\Gtilde) \Gtilde^\ell \vb \right|
\end{aligned} \end{equation}
Applying the Cauchy-Schwarz inequality,
\begin{equation*} \begin{aligned}
& \sum_{\ell=0}^{r-1}
\left| \bbE \left[ \left( \Gtilde - \G \right) \G^q \vb \right]^\top 
		\G^{r-\ell-1} (\G-\Gtilde) \Gtilde^\ell \vb \right| \\
&~~~\le
\sum_{\ell=0}^{r-1}
\sqrt{ \bbE \left\| \left( \Gtilde - \G \right) \Gtilde^q \vb \right\|^2 }
\sqrt{ \bbE \left\| \G^{r-\ell-1} (\G-\Gtilde) \Gtilde^\ell \vb \right\|^2 } \\
&~~~~~~~~~+
\sum_{\ell=0}^{r-1} \sum_{m=0}^{q-1}
\sqrt{ \bbE \left\| \left( \Gtilde - \G \right) \G^{q-m-1} (\G-\Gtilde) \Gtilde^m \vb \right\|^2 }
\sqrt{ \bbE \left\| \G^{r-\ell-1} (\G-\Gtilde) \Gtilde^\ell \vb \right\|^2 } \\
&~~~\le
\sum_{\ell=0}^{r-1}
\sqrt{ \bbE \left\| \left( \Gtilde - \G \right) \Gtilde^q \vb \right\|^2 }
\sqrt{ \bbE \left\| (\G-\Gtilde) \Gtilde^\ell \vb \right\|^2 } \\
&~~~~~~~~~+
2\sum_{\ell=0}^{r-1} \sum_{m=0}^{q-1}
\sqrt{ \bbE \left\| (\G-\Gtilde) \Gtilde^m \vb \right\|^2 }
\sqrt{ \bbE \left\| (\G-\Gtilde) \Gtilde^\ell \vb \right\|^2 } ,
\end{aligned} \end{equation*}
where the second inequality follows from submultiplicativity and the bounds $\| \G \| \le 1$ and $\|\G-\Gtilde\| \le 2$.
Applying Lemma~\ref{lem:proj-equiv:GGtilde:quaddiff:vbounded:nonasy} with $\vv = \Gtilde^s \vb$ for $s=0,1,2,\dots,(r-1) \vee (q-1)$,
\begin{equation*} \begin{aligned}
\sum_{\ell=0}^{r-1}
\left| \bbE \left[ \left( \Gtilde - \G \right) \G^q \vb \right]^\top 
		\G^{r-\ell-1} (\G-\Gtilde) \Gtilde^\ell \vb \right|
\le
Cq r n \nu_n \sum_{i=1}^n 
    \bbE \frac{ \left\| \vb \right\|_\infty^2 }{ \dtilde_i^2 },
\end{aligned} \end{equation*}
where we have again used the fact that $\Gtilde$ is an averaging operator to write $\| \Gtilde^s \vb \|_\infty \le \| \vb \|_\infty$.
Applying this bound to Equation~\eqref{eq:proj-equiv:GGtilde:ellsum},
\begin{equation} \label{eq:proj-equiv:GGtilde:emellsum} \begin{aligned}
&\left| \sum_{\ell=0}^{r-1}
\bbE \!\left[ \!\left( \Gtilde \!-\! \G \right) \G^q \vb \right]^\top 
		\! \G^{r-\ell-1} (\G\!-\!\Gtilde) \Gtilde^\ell \vb \right| \\
&~~~~\le
Cq r n \nu_n \sum_{i=1}^n 
    \bbE \frac{ \left\| \vb \right\|_\infty^2 }{ \dtilde_i^2 } \\
&~~~~~~~~~+
\sum_{\ell=0}^{r-1} \sum_{m=0}^{q-1}
\left| \bbE \left[ \G^{q-m-1} \! (\G-\Gtilde) \Gtilde^m \vb \right]^\top 
                \left( \Gtilde \!-\! \G \right)^\top \!
                \G^{r-\ell-1} (\G-\Gtilde) \Gtilde^\ell \vb \right| .
\end{aligned} \end{equation}

By the Cauchy-Schwarz inequality,
\begin{equation*} \begin{aligned}
& \sum_{\ell=0}^{r-1} \sum_{m=0}^{q-1}
\left| \bbE \left[ \G^{q-m-1} \! (\G-\Gtilde) \Gtilde^m \vb \right]^\top 
                \left( \Gtilde \!-\! \G \right)^\top \!
                \G^{r-\ell-1} (\G-\Gtilde) \Gtilde^\ell \vb \right| \\
&~~~\le
\sum_{\ell=0}^{r-1} \sum_{m=0}^{q-1}
\sqrt{ \bbE \left\| \G^{q-m-1} \! (\G-\Gtilde) \Gtilde^m \vb \right\|^2 }
\sqrt{ \bbE \left\| 
                \left( \Gtilde \!-\! \G \right)^\top \!
                \G^{r-\ell-1} (\G-\Gtilde) \Gtilde^\ell \vb \right\|^2 } \\
&~~~\le
2\sum_{\ell=0}^{r-1} \sum_{m=0}^{q-1}
\sqrt{ \bbE \left\| (\G-\Gtilde) \Gtilde^m \vb \right\|^2 }
\sqrt{ \bbE \left\| (\G-\Gtilde) \Gtilde^\ell \vb \right\|^2 } ,
\end{aligned} \end{equation*}
where the second inequality follows from submultiplicativity and the fact that both $\G$ and $\Gtilde$ are transition matrices.
Applying Lemma~\ref{lem:proj-equiv:GGtilde:quaddiff:vbounded:nonasy} with $\vv = \Gtilde^s \vb$ for $s=0,1,2\dots,(r-1)\vee (q-1)$, and noting that $\| \Gtilde^s \vb \|_\infty \le \| \vb \|_\infty$, since $\Gtilde$ is an averaging operator,
\begin{equation*} \begin{aligned}
\sum_{\ell=0}^{r-1} \sum_{m=0}^{q-1}
\left| \bbE \left[ \G^{q-m-1} \! (\G-\Gtilde) \Gtilde^m \vb \right]^\top 
                \left( \Gtilde \!-\! \G \right)^\top \!
                \G^{r-\ell-1} (\G-\Gtilde) \Gtilde^\ell \vb \right| 
\le
C q r n \nu_n \sum_{i=1}^n 
	\bbE \frac{ \left\| \vv \right\|_\infty^2 }{ \dtilde_i^2 }.
\end{aligned} \end{equation*}
Applying this to Equation~\eqref{eq:proj-equiv:GGtilde:emellsum},
\begin{equation*} \begin{aligned}
\left| \sum_{\ell=0}^{r-1}
\bbE \!\left[ \!\left( \Gtilde \!-\! \G \right) \G^q \vb \right]^\top 
		\! \G^{r-\ell-1} (\G\!-\!\Gtilde) \Gtilde^\ell \vb \right| 
\le
Cq r n \nu_n \sum_{i=1}^n 
    \bbE \frac{ \left\| \vb \right\|_\infty^2 }{ \dtilde_i^2 } .
\end{aligned} \end{equation*}
Applying this to Equation~\eqref{eq:proj-equiv:peer:qr:step2} in turn,
\begin{equation*}
\left| \bbE \left[ \G^q \vb \right]^\top 
        \left( \Gtilde - \G \right)^\top \G^r \vb \right|
\le 
Cq r n \nu_n \sum_{i=1}^n 
    \bbE \frac{ \left\| \vb \right\|_\infty^2 }{ \dtilde_i^2 } ,
\end{equation*}
as we set out to show.
\end{proof}

\begin{lemma} \label{lem:proj-equiv:peercon:qrRecurse:meanzero}
Under the setting of Lemma~\ref{lem:proj-equiv:peer},
suppose that $\vb$ has mean-zero, uncorrelated entries with
$\max_i \bbE b_i^2 \le \sigmabb^2$,
and suppose that $\mA-\mApop$ is independent of $\vb$ conditional on $\Xpop$.
Then, for any non-negative integers $q$ and $r$,
\begin{equation*}
\left| \bbE \left[ \G^q \vb \right]^\top \left( \G - \Gtilde \right)^\top
			\G^r \vb \right|
\le C \sigmabb^2 q r n \nu_n \sum_{i=1}^n 
    \bbE \frac{ 1 }{ \dtilde_i^2 } ,
\end{equation*}
where $C$ is a constant not depending on $q$ and $r$.
\end{lemma}
\begin{proof}
The proof follows an argument parallel to that used to prove Lemma~\ref{lem:proj-equiv:peercon:qrRecurse:bounded},
but using Lemmas~\ref{lem:proj-equiv:GGtilde:lineardiff:meanzero:nonasy}
and~\ref{lem:proj-equiv:GGtilde:quaddiff:meanzero:nonasy}
in place of,
respectively,
Lemmas~\ref{lem:proj-equiv:GGtilde:lineardiff:bounded:nonasy}
and~\ref{lem:proj-equiv:GGtilde:quaddiff:vbounded:nonasy}.
\end{proof}

\begin{lemma} \label{lem:proj-equiv:peercon:linearTerm:diff}
	Under the conditions of Lemma~\ref{lem:proj-equiv:peer},
	\begin{equation*}
	\left\| \bbE \left( \ZlatOracle - \ZpeerOracle \right)^\top
			\ZpeerOracle \right\|
	= o( \sqrt{n} ) .
	\end{equation*}
\end{lemma}
\begin{proof}
	Recalling the structure of $\ZpeerOracle$ and $\ZlatOracle$ from Equations~\eqref{eq:def:ZpeerOracle} and~\eqref{eq:def:ZlatOracle}, respectively,
	\begin{equation*} 
	\left\| \bbE \left( \ZlatOracle - \ZpeerOracle \right)^\top
							\ZpeerOracle \right\|
	= \left| \bbE \Y^\top \left( \Gtilde - \G \right)^\top
				\G \Y \right| .
	\end{equation*} 
Applying the Neumann expansion to $\mXi$ followed by the triangle inequality,
\begin{equation*} \begin{aligned}
\left\| \bbE \left( \ZlatOracle - \ZpeerOracle \right)^\top
				\ZpeerOracle \right\|
&\le
\sum_{q=0}^\infty \sum_{r=0}^\infty \left| \betay \right|^{q+r}
	\left| \bbE
	\left( \mLtilde + \be \right)^\top \left[ \G^q \right]^\top
	\left( \G - \Gtilde \right)^\top
	\G^{r+1} \left( \mLtilde + \be \right) \right| \\
&\le
\sum_{q=0}^\infty \sum_{r=0}^\infty \left| \betay \right|^{q+r}
	\left[
	\left| \bbE \left( \G^q \mLtilde \right)^{\!\top} \!\!
		\left( \G \!-\! \Gtilde \right)^{\!\top} \! 
		\G^{r+1} \mLtilde \right|
	+
	\left| \bbE \left( \G^q \be \right)^{\!\top} \!\!
		\left( \G \!-\! \Gtilde \right)^{\!\top} \!
		\G^{r+1} \be \right| \right] .
\end{aligned} \end{equation*}
Applying Lemma~\ref{lem:proj-equiv:peercon:qrRecurse:bounded} with $\vb = \mLtilde$ and Lemma~\ref{lem:proj-equiv:peercon:qrRecurse:meanzero} with $\vb = \be$,
\begin{equation*} \begin{aligned}
\left\| \bbE \left( \ZlatOracle - \ZpeerOracle \right)^\top
					\ZpeerOracle \right\|
&\le
C \left[ n \nu_n \sum_{i=1}^n \bbE \frac{\|\mLtilde\|_\infty^2+\sigmaeps^2}                      { \dtilde_i^2 } \right]
\sum_{q=0}^\infty \sum_{r=0}^\infty q (r+1) \left| \betay \right|^{q+r+1} \\
&\le
\frac{C |\betay|^2}{(1-\betay)^4}
\left[ n \nu_n \sum_{i=1}^n \bbE \frac{\|\mLtilde\|_\infty^2+\sigmaeps^2}                      { \dtilde_i^2 } \right] ,
\end{aligned}\end{equation*}
where we have used the fact that $|\betay| < 1$.
Since $\betay$ is assumed constant with respect to $n$,
Assumptions~\ref{assum:peer-oracle} and~\ref{assum:equiv:LandX} complete the proof.
\end{proof}

\begin{lemma} \label{lem:proj-equiv:peercon:quadterm:diff}
Under the conditions of Lemma~\ref{lem:proj-equiv:peer},
\begin{equation*}
\left\| \bbE
\left(\ZpeerOracle - \ZlatOracle\right)^\top
\left( \ZpeerOracle - \ZlatOracle \right) \right\|
= o( \sqrt{n} ).
\end{equation*}
\end{lemma}
\begin{proof}
Recalling the structure of $\ZpeerOracle$ and $\ZlatOracle$ from Equations~\eqref{eq:def:ZpeerOracle} and~\eqref{eq:def:ZlatOracle}, respectively,
\begin{equation*} 
\left\| \bbE \left(\ZpeerOracle - \ZlatOracle\right)^\top
		\left( \ZpeerOracle - \ZlatOracle \right) \right\|
= \left| \bbE \Y^\top \left( \mG - \Gtilde \right)^\top
		\left( \mG - \Gtilde \right) \Y \right| .
\end{equation*}
Under the peer contagion model in Equation~\eqref{eq:lim-peer}, we have $\Y = \mXi \left( \mLtilde + \be \right)$, where $\mXi$ and $\mLtilde$ are as in Equation~\eqref{eq:def:XiL}.
Applying the Neumann expansion to $\mXi$ followed by the triangle inequality,
\begin{equation*} \begin{aligned}
&\left\| \bbE \left(\ZpeerOracle - \ZlatOracle\right)^{\!\top} \!
		\! \left( \ZpeerOracle - \ZlatOracle \right) \right\|
\le
\sum_{q=0}^\infty \sum_{r=0}^\infty
|\betay|^{q+r}
\left| \bbE \left(\mLtilde \!+\! \be \right)^\top \!
		\left[ \G^q \right]^\top \! \left( \mG - \Gtilde \right)^\top
		\! \left( \mG - \Gtilde \right)
		\G^r \left( \mLtilde \!+\! \be \right) \right| .
\end{aligned} \end{equation*}
Using an argument similar to that of the proof of Lemma~\ref{lem:proj-equiv:peercon:linearTerm:diff} to replace $\G^q$ and $\G^r$ with $\Gtilde^q$ and $\Gtilde^r$, respectively,
\begin{equation*} \begin{aligned}
&\left\| \bbE \left(\ZpeerOracle - \ZlatOracle\right)^{\!\top} \!
		\left( \ZpeerOracle - \ZlatOracle \right) \right\|
\le
\frac{ C \betay^2 }{(1-\betay)^4 }
n \nu_n \sum_{i=1}^n 
    \bbE \frac{ \sigmaeps^2 + \left\| \mLtilde \right\|_\infty^2 }
	{ \dtilde_i^2 } ,
\end{aligned} \end{equation*}
and Assumptions~\ref{assum:peer-oracle} and~\ref{assum:equiv:LandX} complete the proof.
\end{proof}

\begin{lemma} \label{lem:proj-equiv:peercon:covmxs}
Under the conditions of Lemma~\ref{lem:proj-equiv:peer},
let $\SigpeerOracle$ and $\SiglatOracle$ be as in Equation~\eqref{eq:def:Sigmamxs}.
Then
\begin{equation*}
\left\| \bbE \SigpeerOracle - \bbE \SiglatOracle \right\|
= o( n^{-1/2} ) .
\end{equation*}
\end{lemma}
\begin{proof}
Adding and subtracting appropriate quantities and applying the triangle inequality,
\begin{equation*}
\left\| \bbE \SigpeerOracle - \bbE \SiglatOracle \right\|
\le
\frac{2}{n} \left\| \bbE \left( \ZpeerOracle - \ZlatOracle \right)^\top
			\ZpeerOracle \right\|
+ \frac{1}{n} \left\| \bbE \left(\ZpeerOracle - \ZlatOracle\right)^\top 
			\left( \ZpeerOracle - \ZlatOracle \right)
		\right\| .
\end{equation*}
Applying Lemmas~\ref{lem:proj-equiv:peercon:linearTerm:diff} and~\ref{lem:proj-equiv:peercon:quadterm:diff} yields the result.
\end{proof}

\begin{lemma} \label{lem:proj-equiv:peercon:precmxBounded}
Under the conditions of Lemma~\ref{lem:proj-equiv:peer}, $\SigpeerOracle$ and $\SiglatOracle$, as defined in Equation~\eqref{eq:def:Sigmamxs}, are both invertible and obey
\begin{equation*}
\max\left\{ \left\| \left( \bbE \SigpeerOracle \right)^{-1} \right\|,
		\left\| \left( \bbE \SiglatOracle \right)^{-1} \right\|
	\right\} \le C .
\end{equation*}
\end{lemma}
\begin{proof}
By Assumption~\ref{assum:equiv:limcov}, $\bbE \SiglatOracle$ converges to an invertible matrix, so that the bound on $\| \left( \bbE \SiglatOracle \right)^{-1} \|$ is immediate.
The bound on $\| \left( \bbE \SigpeerOracle \right)^{-1} \|$ then follows from Lemma~\ref{lem:proj-equiv:peercon:covmxs}.
\end{proof}

\begin{lemma} \label{lem:proj-equiv:peercon:precMxs}
Under the conditions of Lemma~\ref{lem:proj-equiv:peer}, with $\SiglatOracle$ and $\SigpeerOracle$ as defined in Equation~\eqref{eq:def:Sigmamxs}, 
\begin{equation*}
\left\| \left( \bbE \SigpeerOracle \right)^{-1}
		- \left( \bbE \SiglatOracle \right)^{-1} \right\|
= o( n^{-1/2} ) .
\end{equation*}
\end{lemma}
\begin{proof}
By Lemma~\ref{lem:proj-equiv:peercon:precmxBounded} and Assumption~\ref{assum:equiv:limcov}, both $\bbE \SigpeerOracle$ and $\bbE \SiglatOracle$ are invertible for $n$ suitably large.
Applying submultiplicativity,
\begin{equation*} \begin{aligned}
\left\| \left( \bbE \SigpeerOracle \right)^{-1}
		- \left( \bbE \SiglatOracle \right)^{-1} \right\|
& \le
\left\| \left( \bbE \SigpeerOracle \right)^{-1}  \right\|
\left\| \left( \bbE \SiglatOracle \right)^{-1} \right\|
\left\| \bbE \SigpeerOracle - \bbE \SiglatOracle \right\|.
\end{aligned} \end{equation*}
Applying Lemma~\ref{lem:proj-equiv:peercon:precmxBounded} again,
\begin{equation*}
\left\| \left( \bbE \SigpeerOracle \right)^{-1}
- \left( \bbE \SiglatOracle \right)^{-1} \right\|
\le C \left\| \bbE \SigpeerOracle - \bbE \SiglatOracle \right\|,
\end{equation*}
and Lemma~\ref{lem:proj-equiv:peercon:covmxs} yields our desired result.
\end{proof}

\begin{proof}[Proof of Lemma~\ref{lem:proj-equiv:peer}]
Applying the triangle inequality,
\begin{equation} \label{eq:proj-equiv:peer:maintri} \begin{aligned}
\left\| \left[ \bbE \ZpeerOracle^\top \ZpeerOracle \right]^{\!-1} \!\!
	\bbE \ZpeerOracle^\top \Y
	- \! \left[\! \bbE \ZlatOracle^\top \ZlatOracle \right]^{\!-1} \!\!
		\bbE \ZlatOracle^\top \Y \right\|
&\le
\left\| \left[ \!\left( \bbE \ZpeerOracle^\top \ZpeerOracle \right)^{\!-1}
	\!\!-\! \left( \bbE \ZlatOracle^\top \ZlatOracle \right)^{\!-1}\!
	\right] \bbE \ZpeerOracle^\top \Y \right\| \\
&~~~~~~+ \left\| \left( \bbE \ZlatOracle^\top \ZlatOracle \right)^{\!-1}
		\! \left[ \bbE \ZpeerOracle^\top \Y 
		- \bbE \ZlatOracle^\top \Y \right] \right\| .
\end{aligned} \end{equation}

Applying submultiplicativity and recalling the defitinitions of $\SigpeerOracle$ and $\SiglatOracle$ from Equation~\eqref{eq:def:Sigmamxs},
\begin{equation*}
\left\| \left( \bbE \ZlatOracle^\top \ZlatOracle \right)^{\!-1}
		\! \left[ \bbE \ZpeerOracle^\top \Y 
		- \bbE \ZlatOracle^\top \Y \right] \right\|
\le
\left\| \left( \bbE \SiglatOracle \right)^{\!-1} \right\|
\left\| \frac{1}{n} \bbE
		\left( \ZpeerOracle - \ZlatOracle \right)^\top \Y 
		\right\| .
\end{equation*}
Applying Lemma~\ref{lem:proj-equiv:peercon:precmxBounded}
and recalling the structure of $\ZlatOracle$ and $\ZpeerOracle$ from Equations~\eqref{eq:def:ZlatOracle} and~\eqref{eq:def:ZpeerOracle}, respectively,
\begin{equation} \label{eq:proj-equiv:peer:triTerm2:intermezzo}
\left\| \left( \bbE \ZlatOracle^\top \ZlatOracle \right)^{\!-1}
		\! \left[ \bbE \ZpeerOracle^\top \Y 
		- \bbE \ZlatOracle^\top \Y \right] \right\|
\le \frac{C}{n} \left|
	\bbE \Y^\top \left( \G - \Gtilde \right)^\top \Y 
	\right|.
\end{equation}

Applying a Neumann expansion and using Lemmas~\ref{lem:proj-equiv:peercon:qrRecurse:bounded} and~\ref{lem:proj-equiv:peercon:qrRecurse:meanzero} as in the proof of Lemma~\ref{lem:proj-equiv:peercon:linearTerm:diff},
\begin{equation*} \begin{aligned}
\left| \bbE \Y^\top \! \left( \!\G - \Gtilde \right)^{\!\top} \! \Y 
	\right|
&\le
\frac{C |\betay|^2}{(1-\betay)^4}
\left[ n \nu_n \sum_{i=1}^n \bbE \frac{\|\mLtilde\|_\infty^2+\sigmaeps^2}                      { \dtilde_i^2 } \right] .
\end{aligned}\end{equation*}
Since $\betay$ is fixed with respect to $n$, Assumptions~\ref{assum:peer-oracle}
        and~\ref{assum:equiv:LandX} yield
\begin{equation*} 
\left| \bbE \Y^\top \! \left( \!\G - \Gtilde \right)^{\!\top} \! \Y 
	\right|
= o( \sqrt{n} ) .
\end{equation*}
Applying this to Equation~\eqref{eq:proj-equiv:peer:triTerm2:intermezzo},
\begin{equation} \label{eq:proj-equiv:peer:triTerm2}
\left\| \left( \bbE \ZlatOracle^\top \ZlatOracle \right)^{\!-1}
		\! \left[ \bbE \ZpeerOracle^\top \Y 
		- \bbE \ZlatOracle^\top \Y \right] \right\|
= o( n^{-1/2} ) .
\end{equation}

Recalling the definitions from Equation~\eqref{eq:def:Sigmamxs}, Lemma~\ref{lem:proj-equiv:peercon:linearTerm} yields
\begin{equation*} \begin{aligned}
\left\| \left[ \!\left( \bbE \ZpeerOracle^\top \ZpeerOracle \right)^{\!-1}
	\!\!-\! \left( \bbE \ZlatOracle^\top \ZlatOracle \right)^{\!-1}\!
	\right] \bbE \ZpeerOracle^\top \Y \right\|
&\le
\left\| \left( \bbE \SigpeerOracle \right)^{\!-1}
        \!\!-\! \left( \bbE \SiglatOracle \right)^{\!-1}
	\right\|
\left\| \bbE \frac{ \ZpeerOracle^\top \Y }{ n } \right\| \\
&\le
C \left\| \left( \bbE \SigpeerOracle \right)^{\!-1}
        \!\!-\! \left( \bbE \SiglatOracle \right)^{\!-1}
	\right\| .
\end{aligned} \end{equation*}
Applying Lemma~\ref{lem:proj-equiv:peercon:precMxs},
\begin{equation} \label{eq:proj-equiv:peer:triTerm1}
\left\| \left[ \!\left( \bbE \ZpeerOracle^\top \ZpeerOracle \right)^{\!-1}
	\!\!-\! \left( \bbE \ZlatOracle^\top \ZlatOracle \right)^{\!-1}\!
	\right] \bbE \ZpeerOracle^\top \Y \right\| \\
= o( n^{-1/2} ) .
\end{equation}

Applying Equations~\eqref{eq:proj-equiv:peer:triTerm2} and~\eqref{eq:proj-equiv:peer:triTerm1} to Equation~\eqref{eq:proj-equiv:peer:maintri} completes the proof.
\end{proof}

\newpage

\printbibliography

\end{document}